\documentclass[11pt]{article}
 
\usepackage{fullpage,graphicx,amssymb,amsmath,url,xcolor}

\newtheorem{remark}{Remark}
\newtheorem{definition}{Definition}
\newtheorem{theorem}{Theorem}
\newtheorem{assumption}{Assumption}
\newtheorem{lemma}{Lemma}
\newtheorem{proposition}{Proposition}

\newenvironment{proof}{Proof. \ }{\ \hfill QED.}

\def\leftsup#1{{}^{#1}}
\def\calLC{\mathcal{LC}}
\def\lc{\mathrm{lc}}
\def\wc{\mathrm{wc}}
\def\HM{\mathrm{HM}}
\def\JS{\mathrm{JS}}
\def\TV{\mathrm{TV}}
\def\LC{\mathrm{LC}}
\def\BC{\mathrm{BC}}
\def\Bhat{\mathrm{Bhat}}
\def\dz{\mathrm{d}z}
\def\Im{\mathrm{Im}}
\def\Re{\mathrm{Re}}
\def\FR{\mathrm{FR}}
\def\SL{\mathrm{SL}}

\def\KC{\mathrm{KC}}
\def\bbC{\mathbb{C}}

\def\bbD{\mathbb{D}}

\def\calL{\mathcal{L}}
\def\calH{\mathcal{H}}

\def\arccosh{\mathrm{arccosh}}

\def\dl{\mathrm{d}l}
\def\ds{\mathrm{d}s}

\def\dtheta{\mathrm{d}\theta}
\def\JS{\mathrm{JS}}

\def\Cauchy{\mathrm{Cauchy}}

\def\calL{\mathcal{L}}

\def\calC{\mathcal{C}}

\def\LS{\mathrm{LS}}

\def\dmu{\mathrm{d}\mu}
\def\bbH{\mathbb{H}}
\def\dx{\mathrm{d}x}
\def\dy{\mathrm{d}y}
\def\KL{\mathrm{KL}}

\def\bbR{{\mathbb{R}}}

\def\calX{\mathcal{X}}
\def\calY{\mathcal{Y}}
\def\st{\ :\ }
\def\SO{\mathrm{SO}}

\def\mattwotwo#1#2#3#4{\left[\begin{array}{cc}#1 & #2\\ #3 & #4\end{array}\right]}

\def\cc{\mathrm{cc}}
\def\Cauchy{\mathrm{Cauchy}}

\begin{document}
\title{On $f$-divergences between Cauchy distributions\thanks{The second author was supported by JSPS KAKENHI 19K14549.}
}

\author{
Frank Nielsen\\
Sony Computer Science Laboratories Inc.\\
E-mail: {\tt Frank.Nielsen@acm.org}\\
\and 
Kazuki Okamura\\
Department of Mathematics, Faculty of Science, Shizuoka University\\
E-mail:  {\tt okamura.kazuki@shizuoka.ac.jp}
}
 
 \date{}

\sloppy
\maketitle              
\begin{abstract}
We prove that the $f$-divergences between univariate Cauchy distributions are all symmetric, and can be expressed as strictly increasing   scalar functions of the symmetric chi-squared divergence. 
We report the corresponding scalar functions for the total variation distance, the Kullback-Leibler divergence, 
 the squared Hellinger divergence, and the Jensen-Shannon divergence among others.
Next, we give conditions to expand the $f$-divergences as converging infinite series of higher-order power chi divergences, and illustrate the criterion for converging Taylor series expressing the $f$-divergences between Cauchy distributions.
We then show that the symmetric property of $f$-divergences holds for multivariate location-scale families with prescribed matrix scales provided that the standard density is even which includes the cases of the multivariate normal and Cauchy families. 
However, the $f$-divergences between multivariate Cauchy densities with different scale matrices are shown asymmetric. 
Finally, we present several metrizations of $f$-divergences between univariate Cauchy distributions 
and further report geometric embedding properties of the Kullback-Leibler divergence.
\end{abstract}

\setcounter{tocdepth}{1}
\tableofcontents

\section{Introduction}

Let $\bbR$, $\bbR_{+}$  and $\bbR_{++}$ be  the sets of real numbers, non-negative real numbers, and  positive real numbers, respectively.
The probability density function of a Cauchy distribution (also called a Lorentzian distribution~\cite{filipovic2012two} in physics) is
$$
p_{l,s}(x) := \frac{1}{\pi s\left(1+\left(\frac{x-l}{s}\right)^2\right)}=\frac{s}{\pi (s^2+(x-l)^2)},
$$
where $l\in\bbR$ denotes the location parameter and $s \in \bbR_{++}$ the scale parameter of the Cauchy distribution, and $x\in\bbR$.
The space of Cauchy distributions form a location-scale family 
$$
\calC=\left\{p_{l,s}(x):=\frac{1}{s} p\left(\frac{x-l}{s}\right) \st (l,s)\in \bbR\times\bbR_{++}\right\},
$$ 
with standard density 
\begin{equation}\label{eq:stdcauchy}
p(x):=\frac{1}{\pi(1+x^2)}.
\end{equation}
To measure the dissimilarity between two continuous probability distributions $P$ and $Q$,
we consider the class of statistical $f$-divergences~\cite{Csiszar-1967,fdivchi-2013} between their corresponding probability densities functions $p(x)$ and $q(x)$ assumed to be strictly positive on $\bbR$:
$$
I_f(p:q) :=\int_{\bbR} p(x) f\left(\frac{q(x)}{p(x)}\right) \dx,
$$
where $f(u)$ is a convex function on $(0,\infty)$, strictly convex at $u=1$ (to ensure reflexivity $I_f(p:q)=0$ iff $p=q$), and satisfying $f(1)=0$ (to ensure positive-definiteness $I_f(p,q)\geq 0$ since by Jensen's inequality we have $I_f(p:q)\geq f(1)=0$).
The Kullback-Leibler divergence (KLD also called relative entropy) is an $f$-divergence obtained for $f_\KL(u)=-\log u$.
In general, the $f$-divergences are oriented dissimilarities: $I_f(p:q) \not=I_f(q:p)$ (eg., the KLD). 
The reverse $f$-divergence $I_f(q:p)$ can be obtained as a forward $f$-divergence for the conjugate function $f^*(u):=uf\left(\frac{1}{u}\right)$ (convex with $f^*(1)=0$): $I_f(q:p)=I_{f^*}(p:q)$. 
We have $I_f=I_g$ when there exists $\lambda\in\bbR$ such that $f(u)=g(u)+\lambda(u-1)$.
Thus an $f$-divergence is symmetric when there exists a real $\lambda$ such that $f(u)=uf\left(\frac{1}{u}\right)+\lambda(u-1)$, and $f$-divergences can always be symmetrized by taking the generator 
$s_f(u)=\frac{1}{2}(f(u)+uf\left(\frac{1}{u}\right))$.
In general, calculating the definite integrals of $f$-divergences is non trivial:
For example, the formula for the KLD between Cauchy densities was only recently obtained~\cite{KLCauchy-2019}:
\begin{eqnarray*}
D_\KL(p_{l_1,s_1}:p_{l_2,s_2}) &:=& I_{f_\KL}(p:q)=\int p_{l_1,s_1}(x)\log\frac{p_{l_1,s_1}(x)}{p_{l_2,s_2}(x)}\dx\\
&=& \log\left( \frac{\left(s_{1}+s_{2}\right)^{2}+\left(l_{1}-l_{2}\right)^{2}}{4 s_{1} s_{2}}\right).
\end{eqnarray*}
Let $\lambda=(\lambda_1=l,\lambda_2=s)$. 
Then we can rewrite the KLD formula as
\begin{equation}\label{eq:kldcauchy}
D_\KL(p_{\lambda_1}:p_{\lambda_2}) = \log \left( 1+ \frac{1}{2}\chi(\lambda_1,\lambda_2)\right),
\end{equation}
where
$$
\chi(\lambda,\lambda'):=\frac{(\lambda_1-\lambda_1')^2+(\lambda_2-\lambda_2')^2}{2\lambda_2\lambda_2'}=\frac{\|\lambda-\lambda'\|^2}{2\lambda_2\lambda_2'}.
$$ 
See \eqref{eq:chidef} for complex representations.

We observe that the KLD between Cauchy distributions is symmetric: $D_\KL(p_{l_1,s_1}:p_{l_2,s_2}) =D_\KL(p_{l_2,s_2}:p_{l_1,s_1})$.
Let 
$$
D_\chi^N(p:q):=\int \frac{(p(x)-q(x))^2}{q(x)}\dx \textup{ and } D_\chi^P(p:q):=\int \frac{(p(x)-q(x))^2}{p(x)}\dx
$$ 
denote the Neyman and Pearson chi-squared divergences between densities $p(x)$ and $q(x)$.
These divergences are $f$-divergences~\cite{fdivchi-2013} for the generators $f_\chi^P(u)=(u-1)^2$ 
and $f_\chi^N(u)=\frac{1}{u}(u-1)^2$, respectively.
The $\chi^2$-divergences between Cauchy densities are symmetric~\cite{CauchyVoronoi-2020}:
$$
D_\chi(p_{\lambda_1}:p_{\lambda_2}):= D_\chi^N(p_{\lambda_1}:p_{\lambda_2})=D_\chi^{P} (p_{\lambda_1}:p_{\lambda_2})=\chi(\lambda_1,\lambda_2), 
$$
hence the naming of the function $\chi(\cdot,\cdot)$. 
Notice that we have
$$
\chi(p_{\lambda_1}:p_{\lambda_2})=\rho(\lambda_1)\rho(\lambda_2)\, \frac{1}{2} D_E^2(\lambda_1,\lambda_2),
$$
where $D_E(\lambda_1, \lambda_2):=\sqrt{(\lambda_2-\lambda_1)^\top(\lambda_2-\lambda_1)}$ 
and $(\lambda_2-\lambda_1)^\top$ denotes the transpose of the vector $(\lambda_2-\lambda_1)$. 
That is, the function $\chi$ is a conformal half squared Euclidean divergence~\cite{nock2015conformal,tJ-2015} 
 with conformal factor $\rho(\lambda):=\frac{1}{\lambda_2}$.
When the Neyman and Pearson chi-squared divergences are not symmetric, we define the chi-squared symmetric divergence as 
$$
D_\chi(p:q)=D_\chi^N(p:q)+D_\chi^P(p:q)=\int \frac{(p(x)+q(x))(p(x)-q(x))^{2}}{p(x) q(x)} \dx.
$$

In this work, we first prove in \S\ref{sec:fdivsymmetic} that all $f$-divergences between univariate Cauchy distributions 
are symmetric (Theorem~\ref{thm:fdivsymmetric}) and can be expressed as a strictly increasing scalar function of the chi-squared divergence (Theorem~\ref{thm:fdivchisquared}). 
We illustrate this result by reporting the corresponding functions for the total variation distance, the Kullback-Leibler divergence, 
the LeCam-Vincze divergence, the squared Hellinger divergence, and the Jensen-Shannon divergence.
Further results for the $f$-divergences between the circular Cauchy, wrapped Cauchy and log-Cauchy distributions based on the invariance properties of 
the $f$-divergences are presented in~\S\ref{sec:logcauchy}.
We report conditions to expand the $f$-divergences as infinite series of higher-order chi divergences and instantiate the results for the Cauchy distributions in~\S\ref{sec:Taylor}.
In~\S\ref{sec:fdivasymmetric}, we then show that the symmetric property  of $f$-divergence holds for multivariate location-scale families including the normal and Cauchy families with prescribed matrix scales provided that the standard density is even, but does not hold for general case of different matrix scales.
We consider metrizations of the square roots of the KLD and the Bhattacharyya divergences in \S\ref{sec:metrization}. 
Finally in \S\ref{sec:properties} we investigate geometric properties of these metrics. 

In the appendix, we first recall the information geometry of the Cauchy family in \S\ref{app:sympstatmdf}, explain the relationship of Cauchy distributions with the M\"obius and Boole transformations in \S\ref{sec:Knight}, give alternative simpler proofs of the Kullback-Leibler divergence (\S\ref{sec:KLDCauchy}) and chi-squared divergence (\S\ref{sec:chisquared}) between Cauchy distributions, report a closed-form formula for the total variation distance between densities of a location-scale family in \S\ref{sec:tvlf}.
In \S\ref{sec:cei}, we also recall the complete elliptic integrals, which are used in the proof of the metrization of the square root of the Bhattacharyya divergence.
We discuss isometric embedding into a Hilbert space of the square root of the KLD in \S\ref{sec:FH}. 
We finally give a code snippet for calculating some converging truncated Taylor series of $f$-divergences between Cauchy distributions in \S\ref{sec:maximataylor}.

\section{Symmetric property of the $f$-divergences between univariate Cauchy distributions}\label{sec:fdivsymmetic}

Consider the location-scale non-abelian group $\LS(2)$ which can be represented as a matrix group~\cite{infproj-2021}.
A group element $g_{l,s}$ is represented by a matrix element $M_{l,s}=\mattwotwo{s}{l}{0}{1}$ for $(l,s)\in\bbR\times\bbR_{++}$.
The group operation $g_{l_{12},s_{12}}=g_{l_1,s_1}\times g_{l_2,s_2}$ corresponds to a matrix multiplication $M_{l_{12},s_{12}}=M_{l_1,s_1}\times M_{l_2,s_2}$ (with the group identity element $g_{0,1}$ being the matrix identity).
A location-scale family is defined by the action of the location-group on a standard density $p(x)= p_{0,1}(x)$.
That is, density $p_{l,s}(x)=g_{l,s}.p(x)$ where `.' denotes the action.
We have the following invariance for the $f$-divergences between any two densities of a location-scale family~\cite{infproj-2021} (including the Cauchy family):
$$
I_f(g.p_{l_1,s_1} : g.p_{l_2,s_2}) = I_f(p_{l_1,s_1} : p_{l_2,s_2}), \forall g\in \LS(2).
$$
Thus we have 
$$
I_f(p_{l_1,s_1} : p_{l_2,s_2}) =  I_f\left( p : p_{ \frac{l_2-l_1}{s_1} , \frac{s_2}{s_1}}\right)
=I_f\left(p_{ \frac{l_1-l_2}{s_2},\frac{s_1}{s_2}}:p\right).
$$ 
Therefore, we may always consider the calculation of the $f$-divergence between the standard density and another density of the location-scale family.
For example, we check that 
$$
\chi((l_1,s_1),(l_2,s_2))=\chi\left((0,1),\left(\frac{l_2-l_1}{s_1},\frac{s_2}{s_1}\right)\right)
$$ 
since
$\chi((0,1),(l,s))=\frac{(s-1)^2+l^2}{2s}$.
If we assume that the standard density $p$ is such that $E_p[X]= \int x p(x) \dx = 0$ and $E_p[X^2]=\int x^2 p(x) \dx =1$ (hence unit variance), then the random variable $Y= \mu+\sigma X$ has mean $E[Y]=\mu$ and standard deviation $\sigma(Y)=\sqrt{E[(Y-\mu)^2]}=\sigma$.
However, the expectation and variance of Cauchy distributions are not defined, hence we preferred $(l,s)$ parameterization over the $(\mu,\sigma^2)$ parameterization, where $l$ denotes the median and $s$ the probable error for the Cauchy location-scale family~\cite{McCullagh1993}.

\subsection{$f$-divergences between densities of a location family}

Let us first prove that $f$-divergences between densities of a location family with {\em even} standard density are symmetric:

\begin{proposition}\label{prop:fdivlocation}
Let $\calL_p=\{p(x-l)\ :\ l\in\bbR\}$ denote a location family with even standard density (i.e., $p(-x)=p(x)$) on the support $\calX=\bbR$.
Then all $f$-divergences between two densities $p_{l_1}$ and $p_{l_2}$ of $\calL$ are symmetric: 
$I_f(p_{l_1}:p_{l_2})=I_f(p_{l_2}:p_{l_1})$.
\end{proposition}

\begin{proof}
Consider the change of variable $l_1-x=y-l_2$ (so that $x-l_2=l_1-y$) with $\dx=-\dy$ and let us use the property that $p(z-l_1)=p(l_1-z)$ since $p$ is an even standard density.
We have:

\begin{eqnarray*}
I_f(p_{l_1}:p_{l_2}) &:=& \int_{-\infty}^{+\infty} p(x-l_1) f\left(\frac{p(x-l_2)}{p(x-l_1)}\right) \dx,\\
&=& \int_{+\infty}^{-\infty} p(l_1-x) f\left(\frac{p(x-l_2)}{p(l_1-x)}\right)  (-\dy),\\
&=&  \int_{-\infty}^{+\infty} p(y-l_2) f\left(\frac{p(x-l_2)}{p(y-l_2)}\right)  \dy,\\
&=& \int_{-\infty}^{+\infty} p(y-l_2) f\left(\frac{p(l_1-y)}{p(y-l_2)}\right)  \dy,\\
&=& \int_{-\infty}^{+\infty} p(y-l_2) f\left(\frac{p(y-l_1)}{p(y-l_2)}\right)  \dy,\\
&=:& I_f(p_{l_2}:p_{l_1}).
\end{eqnarray*}
\end{proof}

Thus $f$-divergences between location Cauchy densities are symmetric since $p(x)=p(-x)$ for the standard Cauchy density of Eq.~\ref{eq:stdcauchy}.

\subsection{$f$-divergences between Cauchy distributions are symmetric}

Let $\|\lambda\|=\sqrt{\lambda_1^2+\lambda_2^2}$ denote the Euclidean norm of a 2D vector $\lambda=(\lambda_1,\lambda_2)\in\bbR^2$.
We state the main theorem:
\begin{theorem}\label{thm:fdivsymmetric}
All $f$-divergences between univariate Cauchy distributions $p_\lambda$ and $p_{\lambda'}$ with $\lambda=(l,s)$ and $\lambda'=(l',s')$ are symmetric and can be expressed as 
$$
I_f(p_{\lambda} : p_{\lambda'}) = h_f\left(\chi(\lambda,\lambda')\right)$$ 
where
$$
\chi(\lambda,\lambda'):= \frac{\|\lambda-\lambda'\|^2}{2\lambda_2\lambda_2'}
$$ 
and $h_f:\bbR_{+}\rightarrow\bbR_+$ is a function (with $h_f(0)=0$). 
\end{theorem}

The proof does not yield explicit closed-form formula for the $f$-divergences as it can be in general difficult to calculate in closed forms, and relies on
 McCullagh's complex parametrization~\cite{McCullagh1993} $p_\theta$ of the parameter of the Cauchy density $p_{l,s}$ with $\theta=l+is$:
$$
p_\theta(x)=\frac{|\Im(\theta)|}{\pi |x - \theta|^2},
$$ 
since $|x-(l+is)|^2=((x-l)+is)((x-l)-is)=(x-l)^2+s^2$.
The parameter space $\theta$ is the complex plane $\bbC$ where we identify $\bar\theta$ with $\theta$, and the Cauchy distributions are degenerated to Dirac distributions $\delta_{l}(x)$ whenever $s=0$.

We make use of the special linear group $\SL(2, \bbR)$ for $\theta$ the complex parameter:
$$
\SL(2, \bbR) := \left\{ \mattwotwo{a}{b}{c}{d} \st a,b,c,d\in\bbR, ad-bc=1  \right\}.
$$
 
Let $A.\theta:=\frac{a\theta+b}{c\theta+d}$ (real linear fractional transformations) be the action of $A=\mattwotwo{a}{b}{c}{d}\in\SL(2, \bbR)$.
McCullagh proved that if $X\sim\Cauchy(\theta)$ then $A.X\sim\Cauchy\left(A.\theta\right)$, where $\theta\in\bbC$ is identified with $\bar{\theta}$ (hence $\lambda(\theta)=(\Re(\theta),|\Im(\theta)|)$).
For example, if $X\sim\Cauchy(is)$ then $\frac{1}{X}\sim\Cauchy(\frac{1}{is})=\Cauchy(-\frac{i}{s})\equiv \Cauchy(\frac{1}{s})$.
Using the $\lambda=(l,s)$ parameterization, we have
\begin{eqnarray*}
l_A &=&\frac{\left(a l+b\right)\left(cl+d\right)+a c s^{2}}{\left(cl+d\right)^{2}+c^{2} s^{2}},\\
s_A &=& \left|\frac{(a d-b c) s}{\left(c l+d\right)^{2}+c^{2} s^{2}}\right|.
\end{eqnarray*}

We can also define an action of $\SL(2, \mathbb{R})$ to the real line $\bbR$ by 
$x \mapsto \frac{ax+b}{cx+d}, \ \ \ x \in \bbR$,
where we interpret $-\frac{d}{c} \mapsto \frac{a}{c}$
if $c \not= 0$. 
We remark that $d \not= 0$ if $c = 0$. 
This map is bijective between $\bbR$. 
We have the following invariance:

\begin{lemma}[Invariance of Cauchy $f$-divergence under $\SL(2, \bbR)$]\label{lemma:finv}
For any $A\in\SL(2, \bbR)$ and $\theta_1,\theta_2\in\bbH$, we have
$$
I_f(p_{A.\theta_1}:p_{A.\theta_2})=I_f(p_{\theta_1}:p_{\theta_2}).
$$
\end{lemma}

\begin{proof} 
We prove the invariance by the change of variable in the integral. 
Let $D(\theta_1: \theta_2) := I_f \left(p_{\theta_1} : p_{\theta_2} \right)$.
We have 
$$
D(A.\theta_1: A.\theta_2) = \int_{\bbR} 
\frac{\Im(A.\theta_1)}{\pi |x- A.\theta_1|^2} f\left( \frac{\Im(A.\theta_2)  |x- A.\theta_1|^2}{\Im(A.\theta_1)  |x- A.\theta_2|^2}\right)  \dx. 
$$

Since $A \in \SL(2, \mathbb{R})$, 
we have 
$$
\Im(A.\theta_i) = \frac{\Im(\theta_i)}{|c\theta_i + d|^2}, \quad i\in\{1,2\}.
$$

If $x = A.y$ then 
$\dx = \frac{\dy}{|cy+d|^2}$,
and 
$$
\left|A.y- A.\theta_i \right|^2 = \frac{|y - \theta_i|^2}{|cy + d|^2\ |c\theta_i + d|^2}, \quad i\in\{1,2\}
.$$
Hence we get:
\begin{eqnarray*}
\int_{\bbR} f\left( \frac{\Im(A.\theta_2)  |x- A.\theta_1|^2}{\Im(A.\theta_1)  
|x- A.\theta_2|^2}\right) \frac{\Im(A.\theta_2)}{\pi |x- A.\theta_1|^2} \dx &=&
 \int_{\bbR} f\left( \frac{\Im(\theta_2)  |y- \theta_1|^2}{\Im(\theta_1)  |y- \theta_2|^2}\right) \frac{\Im(\theta_2)}{\pi |y- \theta_2|^2} \dy,\\
&=& I_f \left(p_{\theta_1} : p_{\theta_2} \right).
\end{eqnarray*}
\end{proof}

Let us notice that the Cauchy family is the only univariate location-scale family that is also closed by inversion~\cite{KnightCauchy-1976}:
That is, if $X\sim\Cauchy(l,s)$ then $\frac{1}{X}\sim\Cauchy(l',s')$.
 Therefore our results are specific to the Cauchy family and not to any other location-scale family.
However the characterization by \cite{KnightCauchy-1976} yields some applications. 
See Appendix \ref{sec:Knight} for details.

We now prove Theorem~\ref{thm:fdivsymmetric} using the notion of maximal invariants of Eaton~\cite{Eaton-1989} (Chapter 2) that will be discussed in \S\ref{sec:maxinvariant}.

Let us rewrite the function $\chi$ with complex arguments as:
\begin{equation}\label{eq:chidef} 
\chi(z,w) := \frac{|z-w|^2}{2\,\Im(z) \Im(w)},\quad z,w\in\bbC. 
\end{equation}

\begin{proposition}[McCullagh~\cite{McCullagh1993}] \label{prop:McCullagh-mi}
The function $\chi$ defined in Eq. \ref{eq:chidef} is a {\em maximal invariant} for the action of the special linear group $\SL(2, \bbR)$ 
to $\bbH \times \bbH$ defined by 
$$
A.(z,w) := \left(\frac{az+b}{cz+d},  \frac{aw+b}{cw+d}\right), \quad
 A = \mattwotwo{a}{b}{c}{d} \in \SL(2, \bbR), \ z, w \in \bbH.
$$ 
That is, we have
$$
\chi(A.z, A.w) = \chi(z,w), \quad A \in \SL(2, \bbR), \ z, w  \in \bbH,
$$
and it holds that for every $z, w, z^{\prime}, w^{\prime} \in \mathbb{H}$ satisfying that $\chi(z^{\prime}, w^{\prime}) = \chi(z,w)$, there exists $A \in \SL(2, \bbR)$ such that $(A.z, A.w) = (z^{\prime}, w^{\prime})$. 
\end{proposition}

By Lemma \ref{lemma:finv} and Theorem 2.3 of \cite{Eaton-1989}, 
there exists a unique function $h_f: [0, \infty)  \to [0, \infty)$ such that  $h_f(\chi(z,w))=D(z,w)$ for all $z,w\in\bbH$.

\begin{theorem}\label{thm:fdivchisquared}
The $f$-divergence between two univariate Cauchy densities is symmetric and expressed as a function of the chi-squared divergence:
\begin{equation}\label{rep}
I_f(p_{\theta_1}:p_{\theta_2}) = I_f(p_{\theta_2}:p_{\theta_1}) = h_f(\chi(\theta_1,\theta_2)), \quad \theta_1,\theta_2 \in \bbH.
\end{equation} 
\end{theorem} 
Therefore we have proven that  the $f$-divergences between univariate Cauchy densities are all symmetric.
Note that we have $h_f=h_{f^*}$.
In general, the $f$-divergences between two Cauchy mixtures $m(x)=\sum_{i=1}^{k} w_i p_{l_i,s_i}(x)$ and 
$m'(x)=\sum_{i=1}^{k'} w_i' p_{l_i',s_i'}(x)$ are asymmetric (i.e., $I_f(m:m')\not=I_f(m':m)$) except when $k=k'=1$.

Similarly, 
we proved in Proposition~\ref{prop:fdivlocation} that all $f$-divergences between two densities of a location family with even standard density are symmetric. 
These $f$-divergences $I_f[p_{l_1}:p_{l_2}]$ can be expressed as a function $k_f$ of the the absolute value $|l_1-l_2|$:
$$
I_f(p_{l_1}:p_{l_2})=I_f(p_{l_1-l_2}:p)=I_f(p:p_{l_2-l_1})=k_f(|l_1-l_2|).
$$
Since the Cauchy standard density is even, we have for a prescribed scale subfamily
 $I_f(p_{l_1,s}:p_{l_2,s})=h_f(\chi((l_1,s),(l_2,s)))=k_{f,s}(|l_1-l_2|)$.
By the definition of $\chi$, 
it follows that we have $k_{f,s}(u)=h_f\left(\frac{u}{2s^2}\right)$.

\begin{remark}\label{rmk:connection}
It has been shown that Amari's dual $\pm\alpha$-connections~\cite{IG-2016} $\leftsup{\alpha}\Gamma$ all coincide with the Levi-Civita metric connection~\cite{locationscaleMfdMitchell-1988}. That is, the $\alpha$-geometry coincides with the Fisher-Rao geometry for the Cauchy family~\cite{CauchyVoronoi-2020}, for all $\alpha\in\bbR$ (see Appendix~\ref{sec:igls}).
Moreover, Eguchi~\cite{Eguchi1983, Eguchi1992} showed how to build an information-geometric dualistic structure
 $(M,\leftsup{D}g,\leftsup{D}\nabla,\leftsup{D}\nabla^*)$  from any arbitrary smooth divergence $D$, consisting of a pair of torsion-free affine connections $(\leftsup{D}\nabla,\leftsup{D}\nabla^*)$ coupled to the metric tensor $\leftsup{D}g$ so that we have 
$\leftsup{D}\nabla^*=\leftsup{D^*}\nabla$ where $D^*(p:q):=D(q,p)$ denotes the reverse divergence.
When the divergence $D$ is a $f$-divergence, it can be shown that the induced connections are $\alpha$-connections,
$\leftsup{D}\nabla=\leftsup{\alpha}\nabla$ and $\leftsup{D}\nabla^*=\leftsup{-\alpha}\nabla$ with $\alpha=3+2\frac{f'''(1)}{f''(1)}$, and the metric tensor $\leftsup{D}g=\frac{1}{f''(1)}\leftsup{F}g$ is proportional to the Fisher information metric tensor $\leftsup{F}g$~\cite{EIG-2020}.
(Notice that Amari defined {\em standard $f$-divergences} in~\cite{IG-2016} by fixing their their scalings so that $f''(1)=1$.)
Since $f$-divergences are symmetric for Cauchy distributions, we have $\leftsup{I_f}\nabla=\leftsup{I_{f^*}}\nabla=\leftsup{I_f}\nabla^*$.
\end{remark}


\begin{remark}
Of course, not all statistical divergences between Cauchy densities are symmetric.
For example, consider the statistical $q$-divergence~\cite{IG-2016} for a scalar $q\in [1,3)$:
$$
D_q (p:r):=\frac{1}{(1-q)Z_q(p)}\left(1-\int p^q(x) r^{1-q}(x) \dmu(x)\right),
$$
where $Z_q(p):= \int p^q(x) \dmu(x)$.
Then the $2$-divergence between two Cauchy densities $p_{\lambda_1}$ and $p_{\lambda_2}$ (with $\lambda_i=(l_i,s_i)$) 
is available in closed-form (as a corresponding Bregman divergence~\cite{IG-2016}):
$$
D_2(p_{\lambda_1}:p_{\lambda_2})=\frac{\pi}{s_2} \|\lambda_1-\lambda_2\|^2.
$$
Thus $D_2(p_{\lambda_1}:p_{\lambda_2})\not=D_2(p_{\lambda_2}:p_{\lambda_1})$ when $s_1\not=s_2$.
\end{remark}

Note that since $I_f(p_{\theta_2}:p_{\theta_1}) = h_f(\chi(\theta_1,\theta_1))$, Lemma~\ref{lemma:finv} can {\it a posteriori} be checked for 
the chi-squared divergence: For any $A\in\SL(2, \bbR)$ and $\theta\in\bbH$, we have
$$
\chi(p_{A.\theta_1}:p_{A.\theta_2})=\chi(p_{\theta_1}:p_{\theta_2}),
$$
and therefore for any $f$-divergence, since we have 
$I_f(p_{A.\theta_1}:p_{A.\theta_2})=I_f(p_{\theta_1}:p_{\theta_2})$ since 
$$
I_f(p_{A.\theta_2}:p_{A.\theta_1}) = h_f(\chi(A.\theta_1,A.\theta_1))= h_f(\chi(\theta_1,\theta_1)) 
=I_f(p_{\theta_2}:p_{\theta_1}).
$$
To prove that $\chi(p_{A.\theta_1}:p_{A.\theta_2})=\chi(p_{\theta_1}:p_{\theta_2})$, let us first recall that
$\Im(A.\theta)=\frac{\Im(\theta)}{|c\theta+d|^2}$
and
$|A.\theta_1-A.\theta_2|^2=\frac{|\theta_1-\theta_2|^2}{|c\theta_1+d|^2 \ |c\theta_2+d|^2}$.
Thus we have
\begin{eqnarray*}
\chi(A.\theta_1,A.\theta_2) &=& \frac{ |A.\theta_1-A.\theta_2|^2 }{ 2\,\Im(A.\theta_1)\Im(A.\theta_2) },\\
&=&
\frac{ |\theta_1-\theta_2|^2 |c\theta_1+d|^2\ |c\theta_2+d|^2}{ |c\theta_1+d|^2 \ |c\theta_2+d|^2\ 2\,\Im(\theta_1)\Im(\theta_2)},\\
&=& \frac{|\theta_1-\theta_2|^2}{ 2\, \Im(\theta_1)\Im(\theta_2)}=\chi(\theta_1,\theta_2).
\end{eqnarray*}

Alternatively, we may also define a bivariate function $g_f(l,s)$ so that using the action of the location-scale group, we have:
$$
h_f(\chi(\theta_1,\theta_2)) = g_f\left(\frac{l_1-l_2}{s_2},\frac{s_1}{s_2}\right),
$$
where $\theta_1=l_1+is_1$ and $\theta_2=l_2+is_2$.
When the function $h_f$ is not explicitly known, we may estimate the $f$-divergences using Monte Carlo importance samplings~\cite{infproj-2021}.

\subsection{Strictly increasing function $h_f$}

We have proven that  
$I_f(p_{\theta_1}:p_{\theta_2}) = I_f(p_{\theta_2}:p_{\theta_1}) = h_f(\chi(\theta_1,\theta_2))$.
Let us prove now that the function $h_f$ is is a strictly increasing function.

\begin{theorem}
Let $f : (0, \infty) \to \mathbb R$ be a convex function such that $f(1) = 0$ and $f \in C^1 ((0,1)) \cap C^1((1, \infty))$ and $f^{\prime}(x) < f^{\prime}(y)$ for every $x < 1 < y$. 
Let $D_f (\lambda : \lambda^{\prime})$ be the $f$-divergence between $p_{\lambda}$ and $p_{\lambda^{\prime}}$, specifically, 
\[ D_f (\lambda : \lambda^{\prime}) = \int_{\mathbb R} p_{\lambda}(x) f\left( \frac{p_{\lambda^{\prime}}(x)}{p_{\lambda}(x)} \right)dx. \] 
Let $\chi$ be McCullagh's maximal invariant. 
Let $h_f  : (0, \infty) \to [0, \infty)$ be the function such that 
\[ h_f (\chi(\lambda, \lambda^{\prime})) = D_f (\lambda : \lambda^{\prime}), \ \lambda, \lambda^{\prime} \in \mathbb{H}. \]
Then, $h_f$ is a strictly increasing function. 
\end{theorem}

The assumption of $f$ is complicated as we would like to cover the important case of the TV distance.

\begin{proof}
Let $u \ge 0$. 
Let $\lambda = i$ and  $\lambda^{\prime} = u+i$. 
Then, $\chi(i, u+i) = \frac{u^2}{2}$ and hence, 
\[ h_f \left(  \frac{u^2}{2} \right) = D_f (p_{i}:p_{u+i}). \]
Hence it suffices to show that $F_1 (u) := D_f (p_{i}:p_{u+i})$ is a strictly increasing function. 
We see that 
\[ F_1(u) = \int_{\mathbb R} \frac{1}{\pi (x^2 + 1)} f \left(  \frac{x^2+ 1}{(x-u)^2 + 1}\right) dx. \]

Then, 
\begin{lemma}\label{lem:exchange}
\[ F_1^{\prime}(u) = \int_{\mathbb R} \frac{2(x-u)}{\pi ((x-u)^2 + 1)} f^{\prime} \left(  \frac{x^2+ 1}{(x-u)^2 + 1}\right) dx, \ u > 0, \]
where we let $ f^{\prime} (  1 ) =0$. 
\end{lemma}

By the change-of-variable formula, 
\[ F_1^{\prime}(u) = \frac{2}{\pi} \int_{\mathbb R} \frac{x}{x^2 + 1} f^{\prime} \left(  \frac{(x+u)^2+ 1}{x^2 + 1}\right) dx, \ u > 0,  \]

We also see that 
\[  \int_{\mathbb R} \frac{x}{x^2 + 1} f^{\prime} \left(  \frac{(x+u)^2+ 1}{x^2 + 1}\right) dx\]
\[ =  \int_{0}^{\infty} \frac{x}{x^2 + 1} f^{\prime} \left(  \frac{(x+u)^2+ 1}{x^2 + 1}\right) dx +  \int_{-\infty}^{0} \frac{x}{x^2 + 1} f^{\prime} \left(  \frac{(x+u)^2+ 1}{x^2 + 1}\right) dx\]
\[ =  \int_{0}^{\infty} \frac{x}{x^2 + 1} \left(f^{\prime} \left(  \frac{(x+u)^2+ 1}{x^2 + 1}\right) - f^{\prime} \left(  \frac{(x-u)^2+ 1}{x^2 + 1}\right) \right)dx. \]

Since $f$ is convex and $x, u > 0$, 
it holds that 
\[ f^{\prime} \left(  \frac{(x+u)^2+ 1}{x^2 + 1}\right) \ge f^{\prime} \left(  \frac{(x-u)^2+ 1}{x^2 + 1}\right), \]
for every $x > 0$ except $x=u/2$. 
By the assumption,  
\[ f^{\prime} \left(  \frac{(x+u)^2+ 1}{x^2 + 1}\right) > f^{\prime} \left(  \frac{(x-u)^2+ 1}{x^2 + 1}\right), \ \ x \in \left(\frac{99}{100}u, \frac{101}{100}u\right). \]
Hence, 
\[ \int_{0}^{\infty} \frac{x}{x^2 + 1} \left(f^{\prime} \left(  \frac{(x+u)^2+ 1}{x^2 + 1}\right) - f^{\prime} \left(  \frac{(x-u)^2+ 1}{x^2 + 1}\right) \right)dx > 0.  \]
\end{proof}

\begin{proof}[Proof of Lemma \ref{lem:exchange}]
We show this assertion for $u = u_0 > 0$. 

\begin{lemma}\label{lem:meanvalue}
For every $c > 1$, 
\[ R_{f,c} := \sup_{1/c < a < b < c} \left|\frac{f(b) - f(a)}{b-a}\right| < +\infty. \]
\end{lemma}

\begin{proof}
We first remark that 
\[ \max_{x \in [1/c, c]} |f^{\prime}(x)| \le \max\left\{|f^{\prime}(1/c)|, |f^{\prime}(c)| \right\}, \]
since $f$ is convex. 

Assume that $a < b \le 1$ or $1 \le a < b$. 
Then, by the mean-value theorem, 
\[  \left|\frac{f(b) - f(a)}{b-a}\right| = |f^{\prime}(\xi)| \le \max\left\{|f^{\prime}(1/c)|, |f^{\prime}(c)| \right\}. \]

Finally we assume that $a < 1 < b$. 
Then, 
\[ \left|\frac{f(b) - f(a)}{b-a}\right| \le \left|\frac{f(1) - f(a)}{1-a}\right| + \left|\frac{f(b) - f(1)}{b-1}\right| \le 2 \max\left\{|f^{\prime}(1/c)|, |f^{\prime}(c)| \right\}.  \]
\end{proof}

For each fixed $u > 0$, 
\[ \frac{1}{1 + u + u^2} \le \frac{x^2+ 1}{(x-u)^2 + 1} \le 1 + u + u^2. \]
Hence, for some $c_0 > 1$, 
\[ \frac{1}{c_0} < \inf_{x \in \mathbb R, u \in  \left(\frac{99}{100}u_0, \frac{101}{100}u_0\right)} \frac{x^2+ 1}{(x-u)^2 + 1} \le \sup_{x \in \mathbb R, u \in  \left(\frac{99}{100}u_0, \frac{101}{100}u_0\right)} \frac{x^2+ 1}{(x-u)^2 + 1} < c_0. \]

Assume that $0 < |h| < u_0/100$. 
Then, by Lemma \ref{lem:meanvalue}, 
\[ \frac{1}{x^2 + 1} \left| \frac{1}{h} \left( f \left(  \frac{x^2+ 1}{(x-u_0 -h)^2 + 1}\right) - f \left(  \frac{x^2+ 1}{(x-u_0)^2 + 1}\right)\right) \right| \]
\[ \le \frac{R_{f, c_0}}{x^2 + 1} \left| \frac{1}{h} \left(\frac{x^2+ 1}{(x-u_0-h)^2 + 1} - \frac{x^2+ 1}{(x-u_0)^2 + 1}\right) \right| \]
\[= R_{f, c_0} \frac{2|x-u_0 -h| + u_0/100}{((x-u_0-h)^2 + 1)((x-u_0)^2 + 1)} \le R_{f, c_0} \frac{1+  u_0/100}{(x-u_0)^2 + 1}. \] 

We see that for $x \ne u_0/2$, 
\[  \frac{1}{x^2 + 1} \lim_{h \to 0} \frac{1}{h} \left( f \left(  \frac{x^2+ 1}{(x-u_0 -h)^2 + 1}\right) - f \left(  \frac{x^2+ 1}{(x-u_0)^2 + 1}\right)\right) \]
\[ = \frac{2(x-u_0)}{(x-u_0)^2 + 1} f^{\prime} \left(  \frac{x^2+ 1}{(x-u_0)^2 + 1}\right).  \]

Now the lemma follows from the dominated convergence theorem.
\end{proof}

\begin{remark}
It follows that the Chebyshev center~\cite{ChebyshevAlphaDiv-2020} $p_{\lambda^*}$ of a 
set of $n$ Cauchy distributions $p_{\lambda_1},\ldots, p_{\lambda_n}$ with respect to any $f$-divergence does not depend on the generator $f$ with 
$\lambda^*=\arg\min_\lambda \max_i I_f(p_{\lambda_i}:p_\lambda)$ since
\begin{eqnarray*}
\arg\min_\lambda \max_i I_f(p_{\lambda_i}:p_\lambda) &=&\arg\min_\lambda \max_i h_f(\chi(\lambda_i,\lambda)),\\
&=&  \arg\min_\lambda \max_i \chi(\lambda_i,\lambda),\\
&=& \arg\min_\lambda \max_i \frac{\|\lambda_i-\lambda\|^2}{\lambda_i}.
\end{eqnarray*}
Similarly, the Cauchy Voronoi diagrams with respect to $f$-divergences all coincide~\cite{CauchyVoronoi-2020}.
\end{remark}

\begin{remark}\label{rmk:conjecture}
It is interesting to consider whether if the symmetry of $f$-divergence between a location-scale family on $\mathbb{R}$ holds for every $f$, then, the family is limited to Cauchy or not. 
If this is true, then, it implies the characterization of the Cauchy distribution by \cite{KnightCauchy-1976} and  \cite{dLocationCauchyGroup-1987}. 
See Proposition \ref{prop:div-inv} in Appendix. 
\end{remark}

\subsection{Some illustrating examples}

\subsubsection{The Kullback-Leibler divergence}

It was proven in~\cite{KLCauchy-2019} that
\begin{eqnarray*}
D_\KL(p_{l_1,s_1}:p_{l_2,s_2}) 
&=& \log\left( \frac{\left(s_{1}+s_{2}\right)^{2}+\left(l_{1}-l_{2}\right)^{2}}{4 s_{1} s_{2}}\right).
\end{eqnarray*}

Thus we have
$$
h_\KL(u)=\log\left(1+\frac{1}{2}u\right).
$$

This plays an important role in establishing an equivalence criterion for two infinite products of Cauchy measures. 
See \cite{okamura2020equivalence}.

\subsubsection{LeCam-Vincze triangular divergence}\label{sec:LeCam}

Let us consider another illustrating example: The LeCam-Vincze triangular divergence~\cite{le2012asymptotic,Vincze-1981} 
defined by
$$
D_{\mathrm{LCV}}(p:q) := \int \frac{(p(x)-q(x))^2}{p(x)+q(x)} \dx.
$$
This divergence is a symmetric $f$-divergence obtained for the generator $f_{\mathrm{LCV}}(u)=\frac{(u-1)^2}{1+u}$.
The triangular divergence is a bounded divergence since $f(0)=f^*(0)=1<\infty$,
 and its square root $\sqrt{D_{\mathrm{LCV}}(p:q)}$ yields a metric distance.
The LeCam triangular divergence between a Cauchy standard density $p_{0,1}$ and a Cauchy density $p_{l,s}$ is
$$
D_{\mathrm{LCV}}(p_{0,1}:p_{l,s}) = 2-4\sqrt{\frac{s}{l^2+s^2+2s+1}} \leq 2.
$$

Since $\chi(p_{0,1}:p_{l,s})=\frac{l^2+(s-1)^2}{2s}$, we can express the triangular divergence using the $\chi$-squared divergence as
$$
D_{\mathrm{LCV}}(p_{l_1,s_1}:p_{l_2,s_2}) = 2-4\sqrt{\frac{1}{2(\chi(p_{l_1,s_1},p_{l_2,s_2})+2)}}.
$$
Thus we have the function:
$$
h_{f_{\mathrm{LCV}}}(u)=2-4\sqrt{\frac{1}{2(u+2)}}.
$$

\subsubsection{Total variation distance}\label{sec:tv}

The total variation  distance (TVD) is a metric $f$-divergence obtained for the generator $f_\TV(u)=\frac{1}{2}|u-1|$:
$$
D_\TV(p:q)=I_{f_\TV}(p:q)=\frac{1}{2}\int_\bbR |p(x)-q(x)|\dx.
$$

Consider the TVD between two Cauchy densities $p_{l_1,s_1}$ and $p_{l_2,s_2}$: $D_\TV(p_{l_1,s_1},p_{l_2,s_2})$.

\begin{itemize}

\item When $s_2=s_1=s$, we have one root $r$ for $p_{l_1,s}(x)=p_{l_2,s}(x)$ since the Cauchy standard density $p(x)$ is even: $r=\frac{l_1+l_2}{2}$.
Assume without loss of generality that $l_1<l_2$. Then we have
\begin{eqnarray*}
D_\TV(p_{l_1,s}:p_{l_2,s}) &=&\frac{1}{2}  \left(\int_{-\infty}^{\frac{l_1+l_2}{2}} (p_{l_1,s}(x)-p_{l_2,s}(x))\dx +
\int_{\frac{l_1+l_2}{2}}^\infty (p_{l_2,s}(x)-p_{l_1,s}(x))\dx
  \right),\\
	&=& \frac{2}{\pi} \arctan\left(\frac{|l_2-l_1|}{2s}\right)\leq 1.
\end{eqnarray*}
Notice that we have $\lim_{x\rightarrow\infty} \arctan(x)=\frac{\pi}{2}$.
We can express $D_\TV(p_{l_1,s}:p_{l_2,s})$ using $\chi(p_{l_1,s},p_{l_2,s})=\frac{(l_2-l_1)^2}{2s^2}$:
$$
D_\TV(p_{l_1,s}:p_{l_2,s}) =  \frac{2}{\pi} \arctan\left(\sqrt{\frac{\chi(p_{l_1,s},p_{l_2,s}) }{2}}\right).
$$
 
See also~Appendix~\ref{sec:tvlf} for the total variation between two densities of a location family.

\item We calculate the two roots $r_1$ and $r_2$ of $p_{l_1,s_1}(x)=p_{l_2,s_2}(x)$ when $s_2\not=s_1$:

\begin{eqnarray*}
r_1&=&\frac{\sqrt{s_1\, {{s_2}^{3}}-2 {{s_1}^{2}}\, {{s_2}^{2}}+\left( {{s_1}^{3}}+\left( {{l_2}^{2}}-2 l_1\, l_2+{{l_1}^{2}}\right) \, s_1\right) \, s_2}+l_1s_2-l_2s_1}{s_2-s_1} ,\\
r_2&=&\frac{\sqrt{s_1\, {{s_2}^{3}}-2 {{s_1}^{2}}\, {{s_2}^{2}}+\left( {{s_1}^{3}}+\left( {{l_2}^{2}}-2 l_1\, l_2+{{l_1}^{2}}\right) \, s_1\right) \, s_2}-l_1s_2+l_2s_1}{s_2-s_1}.
\end{eqnarray*}

Then we use the formula for the definite integral:

$$
I(l,s,a,b):=\int_a^b p_{l,s}(x)\dx=\frac{1}{\pi}\left(\arctan\left(\frac{l-a}{s}\right)-\arctan\left(\frac{l-b}{s}\right)
\right),
$$
where $\arctan(-x)=-\arctan(x)$.

It follows that we have
\begin{eqnarray*}
\lefteqn{D_\TV(p_{l_1,s_1}:p_{l_2,s_2})=}\nonumber\\
&&\frac{1}{\pi}\left(
\arctan\left(\frac{l_2-r_1}{s_2}\right)-\arctan\left(\frac{l_2-r_2}{s_2}\right)
+\arctan\left(\frac{l_1-r_1}{s_1}\right)-\arctan\left(\frac{l_1-r_2}{s_2}\right)
\right).
\end{eqnarray*}

Rearranging and simplifying the terms, we get:
\begin{eqnarray*}
D_\TV(p_{l_1,s_1}:p_{l_2,s_1}) &=&  \frac{2}{\pi} \arctan\left(\sqrt{\frac{\chi(p_{l_1,s_1}:p_{l_2,s_1}) }{2}}\right),\\
&=& h_{f_\TV}\left(\chi[p_{l_1,s_1},p_{l_2,s_1}]\right),
\end{eqnarray*}
with
$$
h_{f_\TV}(u)=\frac{2}{\pi}\arctan\left(\sqrt{\frac{u}{2}}\right).
$$

\end{itemize}

\subsubsection{$f$-divergences for polynomial generators}\label{sec:fdivpoly}

First, let us consider the $f$-divergence between two Cauchy densities for $f$ a (convex) monomial.

\begin{proposition}\label{prop:poly}
Let $a \ge 2$ be an integer. 
Let $J_a$ be a function such that 
\[ J_a (\chi(z,w)) = \int_{\mathbb R} p_z(x)^a p_w(x)^{1-a} \dx, \ z, w \in \mathbb H. \]
Then, 
$J_a$ is a polynomial with degree $a-1$. 
\end{proposition}

\begin{proof}
Let $\lambda \in (0,1)$. 
Then, 
$$ 
J_a \left(\frac{(1-\lambda)^2}{2\lambda}\right) = \frac{1}{\pi} \frac{1}{\lambda^{a-1}} \int_{\mathbb R} \frac{(x^2 + \lambda^2)^{a-1}}{(x^2 + 1)^a} \dx. 
$$
Hence it suffices to show that the right hand side is a polynomial of $\lambda + \lambda^{-1}$. 

Let 
\[ R(a,i) := \int_{\mathbb R} \frac{x^{2i}}{(x^2 + 1)^a} dx, \ 0 \le i \le a-1.   \]
Then, by the change-of-variable that $x = 1/y$, 
\[ R(a,i) = R(a,a-1-i), \ 0 \le i \le a-1. \]

By this and the binomial expansion, 
\[ \frac{1}{\lambda^{a-1}} \int_{\mathbb R} \frac{(x^2 + \lambda^2)^{a-1}}{(x^2 + 1)^a} dx = \sum_{i=0}^{a-1} \binom{a-1}{i} R(a,i) \lambda^{a-1-2i} \]
\[= \sum_{i=0}^{a-1} \binom{a-1}{i} R(a,i) \frac{\lambda^{a-1-2i} + \lambda^{2i-a+1}}{2}.\] 

By induction in $n$, it is easy to see that $\lambda^n + \lambda^{-n}$ is a polynomial of $\lambda + \lambda^{-1}$ with degree $n$. 
\end{proof}

It holds that
\begin{eqnarray*}
J_2 (t) &=& t+1,\\ 
J_3 (t) &=&   (3(t+1)^2 -1)/2 = \frac{3}{2}t^2+3t+1 ,\\
J_4 (t) &=& (5(t+1)^3 - 3(t+1))/2 = \frac{5}{2}t^3+\frac{15}{2}t^2+6t+1,\\
J_5 (t) &=& (35(t+1)^4 - 30(t+1)^2 + 3)/8 = \frac{35}{8}t^4+\frac{35}{2}t^3+\frac{45}{2}t^2+10t+1.
\end{eqnarray*}

Notice that the smallest degree coefficient $a_0$ of polynomial $J_d(t)=\sum_{i=0}^{d-1} a_it^i$ is always one  since 
when $\chi(z,w)=0$, we have $z=w$ and therefore
$J_d(0)=\int_{\mathbb R} p_z(x)^a p_w(x)^{1-a}  \dx=\int_{\mathbb R} p_z(x)^a p_z(x)^{1-a}  \dx=\int_{\mathbb R} p_z(x)\dx=1=a_0$.

The result extends for  $f$-divergences between two Cauchy densities for $f(u)=P_d(u)=\sum_{i=1}^d a_iu^i-\sum_{i=1}^d a_i$ a convex polynomial in degree $d$ with $P_d(1)=0$.
Notice that the set of convex polynomials of degree $d$ can be characterized by the set of positive polynomials~\cite{PositivePolynomial-2008} of degree $d-2$ since $P_d(u)$ is convex iff $P''_d(u)\geq 0$.
A positive polynomial can always be decomposed as a sum of two squared polynomials~\cite{PosPolynomial-2000,PosPolynomial-2019}.

\begin{proposition}
The $f$-divergence between two Cauchy densities for a convex polynomial generator $P_d(u)$ of degree $d$ can be expressed as a $d-1$ dimensional polynomial $Q_{d-1}$ of the chi-squared divergence: $I_{P_d}(p_{\lambda_1},p_{\lambda_2})=Q_{d-1}(\chi(p_{\lambda_1},p_{\lambda_2}))$.
\end{proposition}

The proof follows from the fact that $I_{P_d}(p_{\lambda_1},p_{\lambda_2})=\sum_{i=0}^d I_{f_{a_i}}(p_{\lambda_1},p_{\lambda_2})$ where 
$f_{a_i}(u)=a_i^u-a_i$ and Proposition~\ref{prop:poly}.
Notice that $\int_{\mathbb R} p_z(x)^a p_w(x)^{1-a} \dx = \int_{\mathbb R} p_z(x)^{1-a} p_w(x)^{a} \dx$ since 
$J_a (\chi(z,w))=J_a (\chi(w,z))$.

\begin{remark}
In practice, we can estimate the coefficients of  $J_d(t)=\sum_{i=0}^{d-1} a_it^i$ using polynomial regression~\cite{PolynomialRegression-1996} as follows:
Let $a=[a_0,\ldots,a_{d-1}]^\top$ denote the vector of polynomial coefficients of $J_d$.
Let us draw $n\geq d$ random variates   $\lambda_1,\ldots,\lambda_n$ and $\lambda'_1,\ldots,\lambda_n'$.
Define the $n\times d$ matrix $M=[m_{ij}]$ with $m_{ij}=\chi(\lambda_i,\lambda_i')^{j-1}$.
Let $b=[b_1,\ldots,b_n]$ denote the vector with $b_i\simeq\int_{\mathbb R} p_{\lambda_i}(x)^d p_{\lambda_i'}(x)^{1-d}  \dx$ is numerically  approximated (e.g., using a quadrature integration rule or by stochastic Monte Carlo integration).
Then we estimate $a$ by $\hat{a}=M^+b$ where $M^+:=(M^\top M)^{-1}M^\top$ is the pseudo-inverse matrix (with $M^+=M^{-1}$ when $n=d$).
Notice that knowing that $\hat{a}_0$ should be close to one, allows to check the quality of the polynomial regression.
In fact, we know that all coefficients $a_i$'s should be rational.

For example, we find that
\begin{eqnarray*}
\hat{J}_6(t)&=&7.958522957345747 t^5 + 39.020985312326296 t^4 + 70.4495468953682 t^3 + 52.37619399770375 t^2 \\
&&+ 14.951303338589055 t + 1.002873073997123  
\end{eqnarray*}
Running a second time, we find another estimate
\begin{eqnarray*}
\hat{J}_6(t)&=&7.8720651949082665 t^5 + 39.38158109425294 t^4 + 70.00024065301261 t^3 + 52.49185790967846 t^2 \\
&&+ 15.004692984586242 t + 0.9992248562053161  
\end{eqnarray*}

We can also estimate similarly the order-$k$ chi divergence~\cite{fdivchi-2013} 
$$
D_{\chi,k}(p:q)=\int \frac{(p(x)-q(x))^k}{q(x)^{k-1}}\dx,
$$ 
for even integers $k\geq 2$.
The order-$k$ chi divergence $D_{\chi,k}(p:q)$ is an $f$-divergence obtained for the convex generator $f_{\chi,k}(u)=(u-1)^k$.
Using the binomial expansion, we have~\cite{fdivchi-2013}:
$$
D_{\chi,k}(p:q) = \sum_{i=0}^k \binom{k}{i} (-1)^i \int q(x)^{i-k+1} p(x)^{k-i} \dx.
$$

For example,  we find using the polynomial regression for $k=6$:
\begin{eqnarray*}
\hat{h}_{f_{\chi,6}}(u)&=& 7.875095431165917 u^5 + 13.124758716080692 u^4 + 2.4996228229861686 u^3 \\
&& +0.0013068731474561446 u^2 -7.942140169951983 10^{-4} u + 4.2270711867131716 10^{-5}. 
\end{eqnarray*}
Another run yields a close estimate:
\begin{eqnarray*}
\hat{h}_{f_{\chi,6}}(u)&=& 7.884522702454348 u^5 + 13.081028848015308 u^4 + 2.5649432632612843 u^3\\
&& + -0.03537083264961893 u^2 + 0.005359945026839341 u  -1.901249938160987 10^{-4}. 
\end{eqnarray*}

Since the first polynomial coefficient $a_0$ of $h_{f_{\chi,k}}(u)$ should be zero, we can assess the quality of the polynomial regression.

A different set of techniques consist in   estimating symbolically the univariate functions $h_f$ and $k_f$ using {\em symbolic regression}~\cite{SymbolicRegression-2002,SymbolicRegression-2021}.
\end{remark}

\subsubsection{The Jensen-Shannon divergence}\label{sec:JSD}

Consider the Jensen-Shannon divergence~\cite{JS-1991,fuglede2004jensen} (JSD) (a special case of Sibson's information radius~\cite{Sibson-1969} of order $1$ between a $2$-point set):

\begin{eqnarray*}
D_\JS(p:q) &=&  \frac{1}{2}\left(D_\KL\left(p:\frac{p+q}{2}\right) + D_\KL\left(q:\frac{p+q}{2}\right)\right),\\
&=&  h\left(\frac{p+q}{2}\right) - \frac{h(p)+h(q)}{2},
\end{eqnarray*} 
where $h(p)=-\int p(x)\log p(x)\dx$ denotes Shannon entropy.
The JSD can be rewritten as
$$
D_\JS(p:q)=  \frac{1}{2}\left(D_K(p:q)+D_K(q:p)\right),
$$
where the divergence $D_K$~\cite{JS-1991} is defined by
$$
D_K(p:q):=\int p(x)\log\frac{2p(x)}{p(x)+q(x)}\dx.
$$
The divergence $D_K$ is an $f$-divergence for the generator $f_K(u)=u\log\frac{2u}{1+u}$ such that the reverse $K$-divergence is
 ${D_K}^*(p:q):=D_K(q:p) =I_{f^*_K}(p:q)$ with conjugate generator $f_K^*(u)=-\log\frac{1+u}{2}$.
Thus the JSD is an $f$-divergence for $f_\JS(u)=\frac{u}{2}u\log\frac{2u}{1+u}-\frac{1}{2}\log\frac{1+u}{2}$.
Since $f_\JS(0)<\infty$, the JSD is upper bounded. It is bounded by $\log 2$ since $D_K(p:q)\leq \log 2$.

Using the fact that $f$-divergences between Cauchy distributions are symmetric, we have
$$
D_\JS(p_{l_1,s_1}:p_{l_2,s_2})=D_K(p_{l_1,s_1}:p_{l_2,s_2})=D_K\left(p:p_{\frac{l_2-l_1}{s_1},\frac{s_2}{s_1}}\right).
$$

To get the  JSD between two Cauchy distributions, we need to find a closed-form formula for $D_\JS(p:p_{l,s}) = D_K(p:p_{l,s})$.
Let us skew the divergence $D_K$~\cite{symJS-2010} with a parameter $\alpha\in(0,1)$:
\begin{equation}
D_{K_\alpha}(p:q):=D_\KL(p:(1-\alpha)p+\alpha q)=\int p(x)\log \frac{p(x)}{(1-\alpha)p(x)+\alpha q(x)} \dx.
\end{equation}
The divergence $D_{K_\alpha}$ is an $f$-divergence for the generator $f_{K_\alpha}(u):=-u\log\left((1-\alpha)+\frac{\alpha}{u}\right)$~\cite{symJS-2010}.

Let $p_1(x):=p_{0,1}(x)= \frac{1}{\pi (x^2 + 1)}$, 
$p_2(x):=p_{l,s}(x)$ and $m_w(x)=(1-w)p_1(x)+w p_2(x) :=  \left(\frac{1-w}{\pi (x^2 + 1)} + \frac{ws}{\pi ((x-l)^2 + s^2)}\right)$. 

In Proposition~1 of~\cite{KLCauchy-2019} (proven in Appendix A), a closed-form is reported for the following definite integral:
\begin{eqnarray*}
A(a, b, c ; d, e, f) &=&\int_{-\infty}^{\infty} \frac{\log \left(d x^{2}+e x+f\right)}{a x^{2}+b x+c} \dx,\\
&=& \frac{2 \pi\left(\log \left(2 a f-b e+2 c d+\sqrt{4 a c-b^{2}} \sqrt{4 d f-e^{2}}\right)-\log (2 a)\right)}{\sqrt{4 a c-b^{2}}}.
\end{eqnarray*}

Relying on this closed-form formula, we find after calculations that we have:
$$
D_\KL(p_1 : m_w) =\log\left( \frac{l^2 + (s+1)^2}{(1-w)(l^2 + s^2 +1)+2ws+2\sqrt{s^2 + s((1-s)^2+l^2)w(1-w)}}\right).
$$

We remark that $(1-w)(l^2 + s^2 +1)+2ws \ge 2s > 0$ and $s^2 + s((1-s)^2+l^2)w(1-w) \ge s^2 > 0$.  
This is analytic with respect to $w$ on $(0,1)$, because there exists a holomorphic extension of this to an open neighborhood of the closed interval $[0,1]$ in $\mathbb C$. 

We consider now the general case:
Let $p_{l_1,s_1}(x):= \frac{s_1}{\pi ((x-l_1)^2 + s_1^2)}$, 
$p_{l_2,s_2}(x):= \frac{s_2}{\pi ((x-l_2)^2 + s_2^2)}$ and consider the mixture: 
\begin{eqnarray*}
m(x) &:=& (1-w)p_{l_1,s_1}(x)+w p_{l_2,s_2}(x),\\
&=& \left(\frac{(1-w) s_1}{\pi ((x-l_1)^2 + s_1^2)} + \frac{w s_2}{\pi ((x-l_2)^2 + s_2^2)}\right).
\end{eqnarray*}

Then we have:
\begin{eqnarray}
\lefteqn{D_\KL(p_{l_1,s_1}:m) =}&&\nonumber\\ 
&&\log\left( \frac{(l_1 - l_2)^2 + (s_1+s_2)^2}{(1-w)(s_1^2 + s_2^2 +(l_1 - l_2)^2)+ 2ws_1 s_2 +2\sqrt{s_1^2 s_2^2 + s_1 s_2 ((s_1-s_2)^2+(l_1 - l_2)^2)w(1-w)}}\right). \label{eq:klpmix}
\end{eqnarray}

Let us report one example:
$$
D_{K_w}(p:p_{1,1}) = D_\KL(p_1 : (1-w)p_1(x)+wp_2(x)) = \log 5 - \log\left(3-w+2\sqrt{1+w-w^2}\right).
$$

When $w=\frac{1}{2}$, we get $D_K(p_{l_1,s_1}:p_{l_2,s_2})=D_\KL(p_{l_1,s_1}:m)$, and we get the JSD between Cauchy densities 
$p_{l_1,s_1}$ and $p_{l_2,s_2}$: 
\begin{eqnarray}
D_\JS(p_{l_1,s_1}:p_{l_2,s_2}) &=& \log\left( \frac{2\sqrt{(l_1 - l_2)^2 + (s_1+s_2)^2}}{\sqrt{(l_1 - l_2)^2 + (s_1+s_2)^2} + 2\sqrt{s_1s_2}}\right),\\
&=:& h_\JS(\chi(p_{l_1,s_1},p_{l_2,s_2})), \nonumber
\end{eqnarray}
with
$$
h_\JS(u)= \log\left( \frac{2\sqrt{2+u}}{\sqrt{2+u}+\sqrt{2}} \right),
$$
since $\frac{(l_1 - l_2)^2 + (s_1+s_2)^2}{2s_1s_2}-2=\frac{(l_1 - l_2)^2 + (s_1-s_2)^2}{2s_1s_2}$.

Since $D_\JS(p_{l_1,s_1}:p_{l_2,s_2})=h\left(\frac{p_{l_1,s_1}+p_{l_2,s_2}}{2}\right) - \frac{h(p_{l_1,s_1})+h(p_{l_2,s_2})}{2}$ 
and $h(p_{l_s})=\log(4\pi s)$~\cite{KLCauchy-2019}, we get a formula
for the Shannon entropy of the mixture of two Cauchy densities:
\begin{equation}\label{eq:ShannonCauchyMix}
h\left(\frac{p_{l_1,s_1}+p_{l_2,s_2}}{2}\right)=D_\JS(p_{l_1,s_1}:p_{l_2,s_2})+\log (4\pi\sqrt{s_1s_2}).
\end{equation}

Notice that the JSD between two Gaussian distributions is not analytic~\cite{mixtures-2016}.

\begin{remark}
Consider a mixture family~\cite{IG-2016,wmixture-2018}
$$
\mathcal{M}:=\left\{m_\theta(x)=\sum_{i=1}^D \theta_i p_i(x) + \left(1-\sum_{i=1}^D \theta_i \right)p_0(x)\ : \theta_i>0, \sum_{i=1}^D \theta_i<1 \right\}
$$ 
where the $p_i(x)$'s are linearly independent component distributions.
The KLD between  two densities $m_{\theta_1}$ and $m_{\theta_2}$ of $\mathcal{M}$ amount to a Bregman divergence~\cite{IG-2016,wmixture-2018} for the Shannon negentropy $F(\theta):=-h(m_\theta)$:
\begin{eqnarray*}
D_\KL(m_{\theta_1}:m_{\theta_2})=B_F(\theta_1:\theta_2),
\end{eqnarray*}
where 
\begin{eqnarray*}
B_F(\theta_1:\theta_2):=F(\theta_1)-F(\theta_2)-(\theta_1-\theta_2)^\top \nabla F(\theta_2).
\end{eqnarray*}

Since $\frac{1}{2}(m_{\theta_1}+m_{\theta_2})=m_{\frac{\theta_1+\theta_2}{2}}$, we have
\begin{eqnarray*}
D_\JS(m_{\theta_1}:m_{\theta_2})&=& \frac{1}{2}\left( D_\KL\left(m_{\theta_1}:\frac{1}{2}(m_{\theta_1}+m_{\theta_2})\right) + 
D_\KL\left(m_{\theta_2}:\frac{1}{2}(m_{\theta_1}+m_{\theta_2})\right)\right),\\ 
&=& \frac{1}{2}\left(B_F\left(\theta_1:\frac{\theta_1+\theta_2}{2}\right) + B_F\left(\theta_2:\frac{\theta_1+\theta_2}{2}\right)\right),\\
&=& \frac{F(\theta_1)+F(\theta_2)}{2}-F\left(\frac{\theta_1+\theta_2}{2}\right) :=J_F(\theta_1:\theta_2).
\end{eqnarray*}
This last expression is called a Jensen divergence~\cite{nielsen2011burbea} $J_F(\theta_1:\theta_2)$. 
In general, the Shannon entropy of a mixture is not available in closed-form.
However, we have shown that the Shannon entropy of a mixture of two Cauchy distributions is available in closed form in Eq.~\ref{eq:ShannonCauchyMix}.

For example, consider the family of mixtures of two Cauchy distributions with prescribed parameters
  $(l_0,s_0)=(0,1)$ and $(l_1,s_1)=(1,1)$.
Then we have the following generator:
$$
F_{0,1,1,1}(\theta)=-h[(1-\theta)p_{0,1}+\theta p_{1,1}]=\theta\log\frac{2\sqrt{1+\theta-\theta^2}+\theta+2}{2\sqrt{1+\theta-\theta^2}-\theta+3}+\log\frac{2\sqrt{1+\theta-\theta^2}-\theta+3}{20\pi},
$$
and the derivative of $F_{0,1,1,1}(\theta$) is
$$
\eta(\theta)=F_{0,1,1,1}'(\theta)=\log \frac{2\sqrt{1+\theta-\theta^2}+\theta+2}{2\sqrt{1+\theta-\theta^2}-\theta+3}.
$$
It follows that the Bregman divergence $B_{F_{0,1,1,1}}(\theta_1:\theta_2)$ is
\begin{center}
\noindent\scalebox{0.95}{$
B_{F_{0,1,1,1}}(\theta_1:\theta_2)=D_\KL[m_{\theta_1}:m_{\theta_2}]=\theta_1\log\frac{(2\sqrt{1+\theta_1-\theta_1^2}+\theta_1+2)(2\sqrt{1+\theta_2-\theta_2^2}-\theta_2+3)}{
(2\sqrt{1+\theta_1-\theta_1^2}-\theta_1+3)(2\sqrt{1+\theta_2-\theta_2^2}+\theta_2+2)}
+\log\frac{2\sqrt{1+\theta_1-\theta_1^2}-\theta_1+3}{2\sqrt{1+\theta_2-\theta_2^2}-\theta_2+2}.
$}
\end{center}

Let us define the skewed $\alpha$-Jensen-Shannon divergence:
\begin{equation}
D_{\JS,\alpha}(p:q)=(1-\alpha)D_\KL(p:(1-\alpha)p+\alpha q)+\alpha D_\KL(q:(1-\alpha)p+\alpha q).
\end{equation}
It is an $f$-divergence (i.e., $D_{\JS,\alpha}(p:q)=I_{f_{\JS,\alpha}}(p:q)$) for the convex generator:
\begin{equation}
f_{\JS,\alpha}=-(1-\alpha)\log(\alpha u+(1-\alpha)) -\alpha u\log\left(\frac{1-\alpha}{u}+\alpha\right).  
\end{equation}
When $\alpha=\frac{1}{2}$, we have $f_\JS(u)=f_{\JS,\frac{1}{2}}(u)=-\frac{1}{2}\log\frac{1+u}{2}-\frac{1}{2}u\log\left(\frac{1}{2u}+\frac{1}{2}\right)=\frac{1}{2}u\log\frac{2u}{1+u}-\frac{1}{2}\log\frac{1+u}{2}$.
The skewed $\alpha$-Jensen-Shannon divergence can be rewritten as
\begin{equation}
D_{\JS,\alpha}(p:q)=h((1-\alpha)p+\alpha q)-((1-\alpha)h(p)+\alpha h(q)).
\end{equation}
Thus we have
\begin{equation}\label{eq:hmixc}
h((1-\alpha)p+\alpha q)=D_{\JS,\alpha}(p:q)+((1-\alpha)h(p)+\alpha h(q)).
\end{equation}
When $p=p_{l_1,s_1}$ and $q=p_{l_2,s_2}$, using Eq.~\ref{eq:klpmix}, we get a closed-form for $D_{\JS,\alpha}(p_{l_1,s_1}:p_{l_2,s_2})$, 
and hence we have a closed-form for the differential entropy of a mixture of two components $h((1-\alpha)p_{l_1,s_1}+\alpha p_{l_2,s_2})$.
Let $m_\theta:=(1-\theta)p_{l_1,s_1}+\theta p_{l_2,s_2})$.

The skewed $\alpha$-Jensen-Shannon divergence between two mixtures $m_{\theta_1}$ and  $m_{\theta_2}$ amounts to
\begin{equation}
D_{\JS,\alpha}(m_{\theta_1}:m_{\theta_2})=h((1-\alpha)m_{\theta_1}+\alpha m_{\theta_2})-((1-\alpha)h(m_{\theta_1})+\alpha h(m_{\theta_2})).
\end{equation}
Since $(1-\alpha)m_{\theta_1}+\alpha m_{\theta_2}=m_{(1-\alpha)\theta_1+\alpha\theta_2}$, we get a closed-form formula for the 
skewed $\alpha$-Jensen-Shannon divergence between two Cauchy mixtures with two prescribed component distributions.

Similarly, the KLD between two Cauchy mixtures $m_{\theta_1}$ and  $m_{\theta_2}$ is available in closed-form using Eq.~\ref{eq:klpmix}.

\end{remark}

\subsubsection{The Taneja divergence}\label{sec:Taneja}

The Taneja $T$-divergence~\cite{Taneja-1995} (Eq.~14) is a symmetric divergence defined by:
$$
D_T(p,q) := \int \frac{p(x)+q(x)}{2}\log  \frac{p(x)+q(x)}{2 \sqrt{p(x) q(x)}} \dx.
$$

The $T$-divergence can be rewritten as $D_T(p:q)=\int A(p(x),q(x))\log \frac{A(p(x),q(x))}{G(p(x),q(x))} \dx$ where 
$A(a,b):=\frac{a+b}{2}$ and $G(a,b):=\sqrt{ab}$ are the arithmetic mean and the geometric mean of $a>0$ and $b>0$, respectively.
(Thus the  $T$-divergence is also called the arithmetic-geometric mean divergence in~\cite{Taneja-1995,Pinsker-2009}.)
In~\cite{acharyya2013bregman}, Banerjee et al. proved that $\sqrt{\Delta(a,b)}$ with $\Delta(a,b)=\log \frac{A(a,b)}{G(a,b)}$ is a metric distance.

The $T$-divergence is an $f$-divergence for the generator:
$$
f_T(u)=\frac{u+1}{2}\log\frac{u+1}{2\sqrt{u}}.
$$
We have  $D_T(p:q)=I_{f_T}(p:q)$ since $f_T(u)$ is convex ($f_T''(u)=\frac{u^2+1}{4u^2(u+1)}$).

The $T$-divergence satisfies $D_\JS(p:q)+D_T(p:q)=\frac{1}{4}D_J(p:q)$, where $D_J(p:q)$ is the Jeffreys divergence:
$$
D_J(p:q) =D_\KL(p:q)+D_\KL(q:p).
$$

Thus we have
$$
D_T(p:q)=\frac{1}{4}D_J(p:q)-D_\JS(p:q).
$$

Since the Jeffreys divergence is an $f$-divergence for the generator $f_J(u)=(u-1)\log u$,
 we get $f_T(u)=\frac{1}{4}f_J(u)-f_\JS(u)$  since $I_{f_T}(p,q)=I_{\frac{1}{4}f_J}(p,q)-I_{f_\JS}(p,q)=I_{\frac{1}{4}f_J-f_\JS}(p,q)$.
(More generally, $I_{f_1-f_2}=I_{f_1}(p:q)-I_{f_2}(p:q)$ is an $f$-divergence when $f_1-f_2$ is convex and strictly convex at $1$.) 

It follows the following closed-form formula for the Taneja divergence between Cauchy densities:
$$
D_T[p_{l_1,s_1},p_{l_2,s_2}]= \log\left(\frac{1}{2}\left( 1+\sqrt{\frac{(s_1+s_2)^2+(l_1-l_2)^2}{4s_1s_2}}  \right)\right).
$$

We can express the $T$-divergence  between Cauchy densities as a function of the chi-squared divergence as follows:

$$
h_T(u) = \frac{1}{2}h_\KL(u)-h_\JS(u)  = \log \left( \frac{1+\sqrt{1+\frac{u}{2}}}{2} \right).
$$

A related divergence to the $T$-divergence is the Kumar-Chhina divergence~\cite{KumarChhina-2005}:
$$
D_{\KC}(p, q)=\int \frac{(p(x)+q(x))(p(x)-q(x))^{2}}{p(x) q(x)} \log \frac{p(x)+q(x)}{2 \sqrt{p(x) q(x)}} \dx.
$$

It is an $f$-divergence for the generator:
$$
f_\KC(u)=\frac{(u+1)(u-1)^{2}}{u} \log \frac{u+1}{2 \sqrt{u}},
$$
since we $D_{\KC}(p, q)=I_{f_\KC}(p,q)$ for the convex generator $f_\KC$.


\subsection{Maximal invariants (proof of Proposition \ref{prop:McCullagh-mi})}\label{sec:maxinvariant}

This subsection gives details of arguments in the final part of \cite[Section 1]{McCullagh1993}. 
 
\begin{proof}
First, let us show that 

\begin{lemma}
For every $(z,w) \in \mathbb{H}^2$, there exist $\lambda \geq 1$ and $A \in \SL(2, \bbR)$ such that 
$ (A.z, A.w) = (\lambda i, i)$.
\end{lemma}

\begin{proof}
Since the special orthogonal group $\SO(2, \bbR)$ is the isotropy subgroup of $\SL(2, \bbR)$ 
for $i$ and the action is transitive, it suffices to show that for every $z \in \bbH$ there 
exist $\lambda \geq 1$ and $A \in \SO(2, \bbR)$ such that $\lambda i = A.z$.  

Since  we have that for every $\lambda > 0$, 
$$ 
\mattwotwo{0}{-1}{1}{0}.\lambda i = \frac{i}{\lambda},$$
it suffices to show that for every $z \in \bbH$ there exist $\lambda > 0$ and $A \in \SO(2, \bbR)$ 
such that $\lambda i = A.z$.  

We have that 
$$ 
\mattwotwo{\cos \theta}{-\sin \theta}{\sin \theta}{\cos \theta}.z = 
\frac{\frac{|z|^2 -1}{2} \sin 2\theta + \Re(z) \cos 2\theta + i \Im(z)}{\left|z \sin \theta + \cos \theta \right|^2},
$$

Therefore for some $\theta$, we have
$$
\frac{|z|^2 -1}{2} \sin 2\theta + \Re(z) \cos 2\theta  = 0.
$$ 
\end{proof}

By this lemma, we have that for some $\lambda, \lambda^{\prime}  \geq 1$ and $A, A^{\prime} \in \SL(2, \bbR)$, 
$$
(\lambda i, i) = (A.z, A.w), \  (\lambda^{\prime} i, i) = (A^{\prime}.z^{\prime}, A^{\prime}.w^{\prime}),
$$

We see that 
$$
\chi(z, w) = \chi(\lambda i, i)  = \frac{(\lambda - 1)^2}{4\lambda} = \frac{1}{4} \left(\lambda + \frac{1}{\lambda} - 2\right), 
$$
and
$$
\chi(z^{\prime}, w^{\prime}) = \chi(\lambda^{\prime} i, i)  = \frac{(\lambda^{\prime} - 1)^2}{4\lambda^{\prime}} = \frac{1}{4} (\lambda^{\prime} + \frac{1}{\lambda^{\prime}} - 2).$$

If $\chi(z^{\prime}, w^{\prime}) =  \chi(z, w)$, 
then, $\lambda = \lambda^{\prime}$ and hence $(A.z, A.w) = (A^{\prime}.z^{\prime}, A^{\prime}.w^{\prime})$.
\end{proof}

\section{Invariance of $f$-divergences and $f$-divergences between distributions related to the Cauchy distributions}\label{sec:logcauchy}

There are several distributions which are strongly related with the Cauchy distributions. 
In this section, we shall make use of the invariance properties of $f$-divergences to derive results for the circular Cauchy~\cite{kato2013extended,pewsey2013circular}, wrapped Cauchy~\cite{MLE-WrappedCauchy-1988} and log-Cauchy~\cite{olkin2007life} families which are all related to the Cauchy distributions via various transformations either on the parameter space or on the observation space.

First, consider the family of circular Cauchy distributions parameterized by complex parameters $w$ belonging to the unit disk $\bbD=\{w\in\bbC\ :\ |w|<1\}$. A Circular Cauchy distribution (CC) is an angular distribution~\cite{pewsey2013circular} playing an important role in circular and directional statistics~\cite{mardia2009directional} with the following probability density function:
$$
p_w^\cc(\phi):=\frac{1}{2\pi} \frac{1-|w|^2}{|e^{i\phi}-w|^2}\,\mathrm{d}z,\quad \phi\in [-\pi,\pi),
$$
where $z:=e^{i\phi}\in\bbC$.
Let $w=\rho e^{i\phi_0}$ be the polar form of $w$. 
The circular Cauchy density can be rewritten~\cite{kato2013extended} as:
$$
p_{\rho,\phi_0}^\cc(\phi)=\frac{1}{2\pi} \frac{1-\rho^2}{1+\rho^2-2\rho\cos(\phi-\phi_0)}\,\mathrm{d}\phi,\quad \phi\in [-\pi,\pi).
$$

Consider the subgroup of M\"obius transformations $\mathrm{SL}_2(\bbC)$ that maps $\bbD$ onto itself via transformations 
of the holomorphic automorphism group of the complex unit disk~\cite{ungar1994holomorphic,needham1998visual} (informally speaking, hyperbolic motions):
$$
w\mapsto t_{\phi,a}(w):=e^{i\phi} \frac{w+a}{\bar{a}w+1},\quad \phi\in [-\pi,\pi), a\in\bbC.
$$

The following invariance of $f$-divergences with respect to non-degenerate holomorphic mappings $t_{\phi,a}$ of parameters holds:

\begin{proposition}\label{thm:ccfdiv}
We have $I_f(p^\cc_{w_1}:p^\cc_{w_2})=I_f(p^\cc_{t_{\phi,a}(w_1)}:p^\cc_{t_{\phi,a}(w_2)})$ for all $\phi\in [-\pi,\pi)$ and $a\in\bbC$.
\end{proposition}

This proposition relies on the fact that $I_f(p_{\theta_1}:p_{\theta_2})=I_f(p_{\eta_1}:p_{\eta_2})$ for any smooth invertible transformations $\eta(\theta)$ (with smooth inverse $\theta(\eta)$). Here, however the distribution parameters are complex numbers.


Next, McCullagh~\cite{mccullagh1992conditional} noticed that if $X\sim\mathrm{Cauchy}(\theta)$ 
then $Y=\frac{1+iX}{1-iX}$ follows $\mathrm{CCauchy}\left(\frac{1+i\theta}{1-i\theta}\right)$ with parameter complex $w=\frac{1+i\theta}{1-i\theta}$. 
Denote the complex parameter reciprocal conversion functions $\theta\leftrightarrow w$ by $w(\theta)=\frac{1+i\theta}{1-i\theta}$ 
and $\theta(w)=i\frac{1-w}{(1+w)}$.   
Let us write $w=a+ib$ for $a,b\in\bbR$.

\begin{theorem}[$f$-divergences between circular Cauchy distributions]\label{thm:fdivsymCC}
The $f$-divergence between two circular Cauchy distributions amounts to the $f$-divergence between two corresponding Cauchy distributions:
$I_f(p^\cc_{w_1}:p^\cc_{w_2})=I_f(p_{\theta(w_1)}:p_{\theta(w_2)})$.
It follows that all $f$-divergences between circular Cauchy distributions are symmetric and can be expressed as  scalar functions of the chi square divergence.
\end{theorem}

This theorem follows from the invariance of $f$-divergences~\cite{IG-2016,EIG-2020} and Theorem~\ref{thm:fdivsymmetric}.
That is, let $Y=m(X)$ for $m$ a diffeomorphism between continuous random variables $X$ and $Y$.
Denote by $p_X$ and $q_Y$ the probability densities functions with support $\mathcal{X}$.
It is a key property of $f$-divergences that $f$-divergences are invariant under diffeomorphic transformations~\cite{fdivdiffeo-2010,infproj-2021}:
$$
I_f(p_{X_1}:p_{X_2})=I_f(q_{Y_1}:q_{Y_2}).
$$ 
This invariance of $f$-divergences further holds for non-deterministic mappings called sufficiency of stochastic kernels~\cite{liese2006divergences}.
This result is related to the the result obtained for the Kullback-Leibler divergence in~\cite{akaoka2021bahadur} (Lemma 5.1).
It is worth noting that the circular Cauchy distribution can be interpreted as the exit distribution of a Brownian motion starting at $w\in\bbD$ when reaching the unit boundary circle, see~\cite{mccullagh1992conditional}.

Next, consider the wrapped Cauchy distributions (WC)~\cite{MLE-WrappedCauchy-1988} with probability density functions:
$$
p^\wc_{\mu, \gamma}(\phi) = \sum_{n=-\infty}^{\infty} \frac{\gamma}{\pi\left(\gamma^{2}+(\phi-\mu+2 \pi n)^{2}\right)}, \quad -\pi\leq \phi<\pi,
$$
where $\mu\in\bbR$ denotes the peak position of the unwrapped distribution and $\gamma>0$ the scale parameter.
Let $\eta=\mu+i\gamma$.

The density can be rewritten equivalently as
$$
p^\wc_{\mu, \gamma}(\phi) = \frac{1}{2 \pi} \frac{\sinh (\gamma)}{\cosh (\gamma)-\cos (\phi-\mu)}.
$$

Since we have the following identity:
$$
p^\cc_{w}(\phi)=p^\wc\left(\phi, \eta(w)\right),\quad  \eta(w)=\frac{w-i}{w+i}
$$
it follows the following theorem:

\begin{theorem}[$f$-divergences between wrapped Cauchy distributions]\label{thm:fdivsymWC}
The $f$-divergence between two wrapped Cauchy distributions amounts to the $f$-divergence between two corresponding Cauchy distributions:
$I_f(p^\wc_{\eta_1}:p^\wc_{\eta_2})=I_f(p_{\theta(\eta_1)}:p_{\theta(\eta_2)})$.
It follows that the $f$-divergence between wrapped Cauchy distributions is symmetric and can be expressed as a scalar function of the chi square divergence.
\end{theorem}

Finally, consider the family $\calLC$ of Log-Cauchy (LC) distributions (see~\cite{olkin2007life}, p. 443) and~\cite{olive2014statistical}, p. 329):
$$
\calLC:=\left\{ p^\lc_{\mu,\sigma}(y) =\frac{1}{y \pi}\left[\frac{\sigma}{(\log y-\mu)^{2}+\sigma^{2}}\right],\quad \mu>0,\sigma>0 \right\},
$$
defined on the positive real support $\calY=\bbR_{++}$.

If $X\sim\mathrm{Cauchy}(l,s)$ is a random variable following a Cauchy distribution then $Y=\exp(X)$ is a random variable  following a log-Cauchy distribution with $\mu=l$ and $\sigma=s$.
Reciprocally,  if $Y$ follows a log-Cauchy distribution $\mathrm{LogCauchy}(\mu,\sigma)$, then $X=\log(Y)$ follows a Cauchy distribution with $l=\mu$ and $s=\sigma$. 
In particular, if $Y\sim\mathrm{LogCauchy}(0,1)$ then $X=\log(Y)\sim\mathrm{Cauchy}(0,1)$.

We state the symmetric property of $f$-divergences between log-Cauchy distributions:

\begin{theorem}\label{thm:fdivlc}
The $f$-divergences between two Log-Cauchy distributions $\mathrm{LogCauchy}(\mu_1,\sigma_1)$ and $\mathrm{LogCauchy}(\mu_2,\sigma_2)$ amount to the $f$-divergences between the two corresponding Cauchy distributions: 
$I_f(p^\lc_{\mu_1,\sigma_1}:p^\lc_{\mu_2,\sigma_2})=I_f(p_{\mu_1,\sigma_1}:p_{\mu_2,\sigma_2})$.
It follows that the $f$-divergences between two Log-Cauchy distributions are symmetric and can be expressed as a scalar function of the chi square divergence.
\end{theorem}

\begin{proof}
First, let us recall that the generic relationships between the probability density functions $p_X$ and $q_Y$
with corresponding real-valued random variables satisfying $Y=m(X)$ for a differentiable and invertible function $m$ with $m'(x)\not=0$ is
\begin{eqnarray*}
p_X(x) &=& m'(x) \times q_Y(m(x)) = m'(x) \times q_Y(y),\\
q_Y(y) &=&  (m^{-1})'(y) \times p_X(m^{-1}(y))= (m^{-1})'(y) \times p_X(x). \label{eq:densitytransform}
\end{eqnarray*}

Now consider the case  $y=m(x)=\exp(x)$ with $m^{-1}(y)=\log(y)$, and $m'(x)=\exp(x)$ and $({m^{-1}})'(y)= {1/y}$.
Let us make a change of variable in the $f$-divergence integral with $y=\exp(x)$ and $\dy=\exp(x)\dx$.
We have $
p_{l,s}(x)\dx=p^\lc_{\mu,\sigma}(y)\dy$,
with $\frac{\dx}{\dy}=\frac{1}{y}$ and $\frac{\dy}{\dx}=e^y$.
Let $q_{Y_i}\sim \mathrm{LogCauchy}(\mu_i,\sigma_i)$ and $p_{X_i}\sim \mathrm{Cauchy}(\mu_i,\sigma_i)$ for $i\in\{1,2\}$.
By a change of variable, we have:

\begin{eqnarray*}
I_f(q_{Y_1}:q_{Y_2}) &:=&  \int_{\bbR_{++}} q_{Y_1}(y) f\left( \frac{q_{Y_2}(y)}{q_{Y_1}(y)}\right) \dy \\
&=& \int_{\bbR_{++}} ({m^{-1}})'(y) \times p_{X_1}(m^{-1}(y)) f\left( \frac{({m^{-1}})'(y)\times p_{X_2}(m^{-1}(y))}{({m^{-1}})'(y) \times p_{X_1}({m}^{-1}(y))} \right) \dy,\\
&=& \int_{\bbR} p_{X_1}(x)f\left( \frac{p_{X_2}(x)}{p_{X_1}(x)}\right) \dx,\\
&=:& I_f(p_{X_1}:p_{X_2}).
\end{eqnarray*}

Then we use the symmetric property of the $f$-divergences of the Cauchy distributions to deduce the symmetry of the $f$-divergences between
 log-Cauchy distributions: $I_f(p^\lc_{\mu_1,\sigma_1}:p^\lc_{\mu_2,\sigma_2})=I_f(p^\lc_{\mu_2,\sigma_2}:p^\lc_{\mu_1,\sigma_1})$. 
It follows that we have $I_f(p^\lc_{\mu_1,\sigma_1}:p^\lc_{\mu_2,\sigma_2})=h_f(\chi((\mu_1,\sigma_1),(\mu_2,\sigma_2)))$.
\end{proof}

\section{Asymmetric Kullback-Leibler divergence between multivariate Cauchy distributions}\label{sec:fdivasymmetric}
For a symmetric positive-definite $d\times d$ matrix $P\succ 0$ and a $d$-dimensional location vector $\mu$, 
the density of a random variable~\cite{infproj-2021} $X_{\mu,P}:=P X+\mu$ with $X\sim p(x)$ (standard density) is 
\begin{equation}
p_{\mu,P}(x):=|P|^{-1}\, p(P^{-1}(x-\mu)).
\end{equation}

A $d$-dimensional location scale family is formed by the set of densities $\{p_{\mu,P}(x)\ :\ P\succ 0, \mu\in\bbR^d\}$.
For example, the set of multivariate normal distributions (MVNs) form a multidimensional location-scale family~\cite{infproj-2021}.

The probability density function of a $d$-dimensional Cauchy distribution~\cite{press1972multivariate} (MVCs) with parameters $\mu \in \bbR^d$ and $\Sigma\succ 0$ be a $d\times d$ positive-definite symmetric matrix is defined by:
$$
p_{\mu,\Sigma}(x) := \frac{C_d}{(\det \Sigma)^{1/2}} \left( 1 + \left(x - \mu\right)^{\top} \Sigma^{-1} \left(x-\mu\right) \right)^{-(d+1)/2}, \ x \in
 \bbR^d,
$$
where $C_d=\frac{\Gamma\left(\frac{d+1}{2}\right)}{\pi^{\frac{d+1}{2}}}$ is a normalizing constant, and $\Gamma(\cdot)$ denotes the gamma function.
The MVCs form a multivariate location-scale family with standard density: 
$$
p(x):=\frac{\Gamma\left(\frac{d+1}{2}\right)}{\pi^{\frac{d+1}{2}}}\left( 1 + x^\top x \right)^{-(d+1)/2},
$$
where matrix parameter $P=\Sigma^{\frac{1}{2}}$ denotes the symmetric positive-definite square root matrix of $\Sigma\succ 0$.
	
In this section, we shall prove that the $f$-divergences between any two densities of a multidimensional location-scale family with prescribed scale root matrix $P$
and even standard density (i.e., $p(x)=p(-x)$) is symmetric, and then show that the KLD between bivariate Cauchy distributions is asymmetric in general.

First, let us consider the case $\Sigma=I$: The corresponding set of multivariate Cauchy distributions yields a multivariate location subfamily $\{p_{\mu}(x)=p_{\mu,I}(x)\ :\ \mu\in\bbR^d\}$ with standard distribution
$p(x)=p_{0,I}(x)= \frac{C_d}{(\det \Sigma)^{1/2}} \left( 1 + x^\top x \right)^{-(d+1)/2}$.
Since the standard density is even (i.e., $p(x)=p(-x)$), we can extend straightforwardly the result of Proposition~\ref{prop:fdivlocation} using a multidimensional change of variable in the integrals of $f$-divergences:

\begin{proposition}\label{prop:fdivloceven}
The $f$-divergences between any two densities of the multivariate location Cauchy family is symmetric: 
$I_f(p_{\mu_1},p_{\mu_2})=I_f(p_{\mu_2},p_{\mu_1})$.
\end{proposition} 

Next, we consider the case of MVC location subfamilies with prescribed matrix $\Sigma$ (or equivalently $P$).
\begin{proposition}\label{prop:fdivlocSigmaeven}
The $f$-divergences between any two densities of the multivariate location Cauchy family $\{p_{\mu,\Sigma}\ :\ \mu\in\bbR^d\}$ with prescribed matrix $\Sigma$ is symmetric: 
$I_f(p_{\mu_1,\Sigma},p_{\mu_2,\Sigma})=I_f(p_{\mu_2,\Sigma},p_{\mu_1,\Sigma})$.
\end{proposition} 

\begin{proof}
We shall use the following identities of $f$-divergences arising from the location-scale family group structure~\cite{infproj-2021}:
$$
I_{f}\left(p_{l_{1}, P_{1}}: p_{l_{2}, P_{2}}\right)=
I_{f}\left(p: p_{P_{1}^{-1}\left(l_{2}-l_{1}\right), P_{1}^{-1} P_{2}}\right)
=I_{f}\left(p_{P_{2}^{-1}\left(l_{1}-l_{2}\right), P_{2}^{-1} P_{1}}: p\right).
$$

Thus for the MVCs, we have:
$$
I_{f}\left(p_{\mu_{1}, \Sigma_{1}}: p_{\mu_{2}, \Sigma_{2}}\right)=
I_{f}\left(p : p_{\Sigma_1^{-\frac{1}{2}}\left(\mu_{2}-\mu_{1}\right), \Sigma_1^{-\frac{1}{2}} \Sigma_2^{\frac{1}{2}}}\right)
=I_{f}\left(p_{\Sigma_2^{-\frac{1}{2}}\left(\mu_{1}-\mu_{2}\right), \Sigma_2^{-\frac{1}{2}} \Sigma_1^{\frac{1}{2}}}: p\right).
$$

It follows that when $\Sigma_1=\Sigma_2=\Sigma$, we get:
$$
I_{f}\left(p_{\mu_{1}, \Sigma}: p_{\mu_{2}, \Sigma}\right)=
I_{f}\left(p : p_{\Sigma^{-\frac{1}{2}}\left(\mu_{2}-\mu_{1}\right), I}\right)
=I_{f}\left(p_{\Sigma^{-\frac{1}{2}}\left(\mu_{1}-\mu_{2}\right), I}: p\right).
$$

Recasting the equalities using the multivariate location Cauchy family, we obtain:
$$
I_{f}\left(p_{\mu_{1}, \Sigma}: p_{\mu_{2}, \Sigma}\right)=
I_{f}\left(p : p_{\Sigma^{-\frac{1}{2}}\left(\mu_{2}-\mu_{1}\right)}\right)
=I_{f}\left(p_{\Sigma^{-\frac{1}{2}}\left(\mu_{1}-\mu_{2}\right)}: p\right).
$$

Since we proved in Proposition~\ref{prop:fdivloceven} for the multivariate Cauchy location family that $I_f(p_{\mu_1},p_{\mu_2})=I_f(p_{\mu_2},p_{\mu_1})$ (with $p_{\mu}(x) := p_{\mu, I}(x)$),
 it follows that we have:

$$
I_{f}\left(p_{\mu_{1}, \Sigma}: p_{\mu_{2}, \Sigma}\right)=
I_{f}\left(p : p_{\Sigma^{-\frac{1}{2}}\left(\mu_{2}-\mu_{1}\right)}\right)
=
I_{f}\left(p_{\Sigma^{-\frac{1}{2}}\left(\mu_{2}-\mu_{1}\right)}:p\right)
=
I_{f}\left(p_{\mu_{2}, \Sigma}: p_{\mu_{1}, \Sigma}\right).
$$
\end{proof}

However, contrary to the family of univariate Cauchy distributions, we have the following result: 

\begin{proposition}\label{prop:locDiffSigma}
There exist two bivariate Cauchy densities $p_{\mu_1, \Sigma_1}$ 
and $p_{\mu_2, \Sigma_2}$ such that 
$
D_\KL\left(p_{\mu_1, \Sigma_1}: p_{\mu_2, \Sigma_2} \right) \not=
D_\KL\left(p_{\mu_2, \Sigma_2}: p_{\mu_1, \Sigma_1} \right)$.
\end{proposition}

\begin{proof}
We let $d=2$.
By the change of variable in the integral~\cite{infproj-2021}, we have
$$
D_\KL\left(p_{\mu_1, \Sigma_1} : p_{\mu_2, \Sigma_2} \right) = 
D_\KL\left(p_{0, I_2} \ { : } \ p_{\Sigma_1^{-1/2}(\mu_2 - \mu_1), \Sigma_1^{-1/2} \Sigma_2 \Sigma_1^{-1/2}}\right), 
$$
where $I_2$ denotes the unit {$2 \times 2$} matrix. 

Let 
$$
\mu_1 = 0, \Sigma_1 = I_2, \  \mu_2 = (0,1)^\top,  \Sigma_2 = \mattwotwo{n}{0}{0}{\frac{1}{n}},
$$
where $n$ is a natural number. 
{We will show that $D_\KL\left(p_{\mu_1, \Sigma_1}: p_{\mu_2, \Sigma_2} \right) \not=
D_\KL\left(p_{\mu_2, \Sigma_2}: p_{\mu_1, \Sigma_1} \right)$ for sufficiently large $n$.}
Then, 
$$
D_\KL\left(p_{\mu_1, \Sigma_1} : p_{\mu_2, \Sigma_2} \right) = 
\frac{3{C_2}}{2} \int_{\bbR^2}\frac{\log(1+ x_1^2 / n + n x_2^2) - \log(1 + x_1^2 + x_2^2)}{(1 + x_1^2 + x_2^2)^{3/2}} \dx_1 \dx_2 
$$
and
\begin{eqnarray*}
 D_\KL\left(p_{\mu_2, \Sigma_2} : p_{\mu_1, \Sigma_1} \right) &=& 
D_\KL\left(p_{0, I_2} : p_{-\Sigma_1^{-1/2}\mu_1, \Sigma_1^{-1}}\right) ,\\
 &=& \frac{3{C_2}}{2} \int_{\bbR^2}\frac{\log(1+ x_1^2 / n + n (x_2 + \sqrt{n})^2) - \log(1 + x_1^2 + x_2^2)}{(1 + x_1^2 + x_2^2)^{3/2}} \dx_1 \dx_2. 
 \end{eqnarray*}

Hence it suffices to show that 
$$
\int_{\bbR^2}\frac{\log(1+ x_1^2 / n + n (x_2 + \sqrt{n})^2) - \log(1+ x_1^2 / n + n x_2^2)}{(1 + x_1^2 + x_2^2)^{3/2}} \dx_1 \dx_2 \ne 0 
$$
for some $n$.
 
We see that $\log(1+ x_1^2 / n + n (x_2 + \sqrt{n})^2) > \log(1+ x_1^2 / n + n x_2^2)$ if and only if $x_2 > -\sqrt{n}/2$. 
Since $\{(x_1, x_2) : x_2 > -\sqrt{n}/2 \} \to \bbR^2, n \to \infty$, 
we see that by Fatou's lemma ~\cite{Kesavan-2019} (p. 93), 
$$
\lim_{n \to \infty} \int_{x_2 > -\sqrt{n}/2}\frac{\log(1+ x_1^2 / n + n (x_2 + \sqrt{n})^2) - \log(1+ x_1^2 / n + n x_2^2)}{(1 + x_1^2 + x_2^2)^{3/2}} \dx_1 \dx_2 = +\infty. 
$$
Hence it suffices to show that 
\begin{equation}\label{eq:lower-int-finite}
\liminf_{n \to \infty} \int_{x_2 \le -\sqrt{n}/2}\frac{\log(1+ x_1^2 / n + n (x_2 + \sqrt{n})^2) - \log(1+ x_1^2 / n + n x_2^2)}{(1 + x_1^2 + x_2^2)^{3/2}} \dx_1 \dx_2 > -\infty. 
\end{equation}

If $x_2 \le -\sqrt{n}/2$, then, 
$$
\log(1+ x_1^2 / n + n (x_2 + \sqrt{n})^2) - \log(1+ x_1^2 / n + n x_2^2) = \log\left(1+ \frac{n^{3/2} (n^{1/2} + 2x_2)}{1+ x_1^2 / n + n x_2^2}\right)
$$
$$
\ge \log\left(1+ \frac{n^{3/2} (n^{1/2} + 2x_2)}{1+ n x_2^2}\right).
$$
Let $f(x) :=  \frac{n^{1/2} + 2x}{1+ n x^2}, \ x < -\sqrt{n}/2$. 
Then, $f$ is decreasing on $\left(-\infty, -\frac{\sqrt{n}}{2} - \sqrt{\frac{n^2+4}{4n}} \right]$ and increasing on $\left[-\frac{\sqrt{n}}{2} - \sqrt{\frac{n^2+4}{4n}}, -\frac{\sqrt{n}}{2} \right]$. 
Since $-\frac{\sqrt{n}}{2} - \sqrt{\frac{n^2+4}{4n}} > -\frac{3}{2}\sqrt{n}$ for $n \ge 2$, 
it holds that for $n \ge 2$, 
$$
 \int_{x_2 \le -3\sqrt{n}/2}\frac{\log(1+ x_1^2 / n + n (x_2 + \sqrt{n})^2) - \log(1+ x_1^2 / n + n x_2^2)}{(1 + x_1^2 + x_2^2)^{3/2}} \dx_1 \dx_2 
 $$
\begin{equation}\label{eq:lower-int-finite-part1}
\ge \log\left(\frac{4+n^2}{4+9n^2}\right)  \int_{\bbR^2}\frac{\dx_1 \dx_2}{(1 + x_1^2 + x_2^2)^{3/2}} \ge -2\pi\log 5. 
\end{equation}

If $x_2 = -\frac{\sqrt{n}}{2} - \sqrt{\frac{n^2+4}{4n}}$, then, 
$$
 \log\left(1+ \frac{n^{3/2} (n^{1/2} + 2x_2)}{1+ n x_2^2}\right) = 2\log 2 - 2 \log\left(n+\sqrt{n^2 + 4}\right) \ge -\log(n^2 +4). 
$$
Hence, 
$$
\int_{-3\sqrt{n}/2 \le x_2 \le -\sqrt{n}/2}\frac{\log(1+ x_1^2 / n + n (x_2 + \sqrt{n})^2) - \log(1+ x_1^2 / n + n x_2^2)}{(1 + x_1^2 + x_2^2)^{3/2}} \dx_1 \dx_2 
$$
$$
\ge -\log(n^2 +4) \int_{-3\sqrt{n}/2 \le x_2 \le -\sqrt{n}/2} \frac{\dx_1 \dx_2}{(1 + x_1^2 + x_2^2)^{3/2}} 
$$
\begin{equation}\label{eq:lower-int-finite-part2}
\ge - \sqrt{n} \log(n^2 +4) \int_{\mathbb R} \frac{\dx_1}{(1 + x_1^2 + n^2 /2)^{3/2}} = -\frac{4\sqrt{n} \log(n^2 +4)}{n^2 + 2} \to 0, \ n \to \infty. 
\end{equation}
By Eq.~(\ref{eq:lower-int-finite-part1}) and (\ref{eq:lower-int-finite-part2}), we have Eq.~(\ref{eq:lower-int-finite}). 

\end{proof}

\begin{remark}
By numerical computations, 
we have that 
\[ \int_{-\infty}^{\infty}\int_{-\infty}^{\infty} \frac{\log(1+ x^2/100 + 100(y+10)^2)}{(1+x^2+y^2)^{3/2}} dxdy = 57.953 \]
and
\[ \int_{-\infty}^{\infty}\int_{-\infty}^{\infty} \frac{\log(1+ x^2/100 + 100y^2)}{(1+x^2+y^2)^{3/2}} dxdy = 30.1523. \]
\end{remark}

\section{Taylor series of $f$-divergences}\label{sec:Taylor}

In this section, we aim at rewriting the $f$-divergences as   converging infinite series of power chi divergences~\cite{fdivchi-2013,nielsen2019power}.
The Pearson power chi divergence $D_{\chi,k}^P$ of order $k$ (for any integer $k\in\{2,\ldots,\}$) is a dissimilarity obtained for the generator $f_{\chi,k}^P(u)=(u-1)^k$ which generalizes the Pearson $\chi_2$-divergence ($k=2$):
\begin{eqnarray*}
D_{\chi,k}^P(p:q) &=& \int p(x)f_{\chi,k}^P\left(\frac{q(x)}{p(x)}\right)\dmu(x),\\
&=& \int p(x)\left(\frac{q(x)}{p(x)}-1\right)^k \dmu(x),\\
&=& \int  \frac{(q(x)-p(x))^k}{p(x)^{k-1}} \dmu(x).
\end{eqnarray*}
We have $D_{\chi,2}^P(p:q)=D_\chi^P(p:q):=\int \frac{(p(x)-q(x))^2}{p(x)}\dmu(x)$.
For even integers $k\geq 4$, the Pearson power chi divergence are non-negative dissimilarities 
since $f_{\chi,k}^P(u)$ is strictly convex  (we have ${f_{\chi,k}^P}''(u)=k(k-1)(u-1)^{k-2}\geq 0$).
For odd integers $k\geq 3$, the Pearson power chi divergence may be negative.
Similarly, we can define the Neyman power chi divergence $D_{\chi,k}^N$ of order $k$:
$$
D_{\chi,k}^N(p:q) = D_{\chi,k}(q:p)= \int  \frac{(p(x)-q(x))^k}{q(x)^{k-1}} \dmu(x).
$$
We have $D_{\chi,2}^N(p:q)=D_\chi^N(p:q):=\int \frac{(p(x)-q(x))^2}{q(x)}\dmu(x)$.
When $k$ is even it is a $f$-divergence, otherwise $D_{\chi,k}^N$ may fail the positive-definiteness property of $f$-divergences.
We note $D_{\chi,k}(p:q)=D_{\chi,k}^P(p:q)$ below.

We first state a general framework to obtain power chi divergence expansions of $f$-divergences. 

\begin{theorem}\label{thm:div-expansion}
Let $X$ be a topological space and $\mu$ be a Borel measure on $X$ with full support. 
Let $\{p_{\theta}(x)\}_{\theta}$ be a family of probability density functions on $(X, \mu)$. 
Assume that for each $\theta$, 
$p_{\theta}(x)$ is positive and continuous with respect to $x$. 
We also assume that for each $\theta_1$ and $\theta_2$ there exists 
$C = C(\theta_1, \theta_2)$ such that $p_{\theta_1}(x) \le C p_{\theta_2}(x)$ for every $x \in X$. 
Let $f(z) = \sum_{n=1}^{\infty} a_n (z-1)^n$ be an analytic function ($f\in C^\omega$), and denote by $r_f$ be the convergence radius of $f$. 
Assume that $r_f \ge 1$. 
Let $I_f$ be the induced $f$-divergence. 
Then, \\
(i) If $\frac{p_{\theta_2}(x)}{p_{\theta_1}(x)} < 1+r_f$ for every $x$, then, 
\[ I_f (p_{\theta_1} : p_{\theta_2}) = \sum_{n=2}^{\infty} a_n \int_{X}  \left( \frac{p_{\theta_2}(x)}{p_{\theta_1}(x)} - 1 \right)^n p_{\theta_1}(x) \dmu(x)=\sum_{n=2}^{\infty} a_n D_{\chi,n}(p_{\theta_1} : p_{\theta_2}). \]
(ii) If $\frac{p_{\theta_2}(x)}{p_{\theta_1}(x)} > 1+r_f$ for some $x$, then, 
the infinite sum \\
$\sum_{n=2}^{\infty} a_n \int_{X} \left( \frac{p_{\theta_2}(x)}{p_{\theta_1}(x)} - 1 \right)^n p_{\theta_1}(x) \mu(dx)$ diverges. 
\end{theorem}

\begin{proof}
(i) By the assumption and $r_f \ge 1$, 
$\inf_{x \in X} \frac{p_{\theta_2}(x)}{p_{\theta_1}(x)} > 1-r_f$. 
Hence, 
$\sup_{x \in X} \left|\frac{p_{\theta_2}(x)}{p_{\theta_1}(x)} - 1 \right| < r_f$.
Thus we have the Taylor series:
$$
 f\left( \frac{p_{\theta_2}(x)}{p_{\theta_1}(x)} \right) =  \sum_{n=2}^{\infty} a_n  \left( \frac{p_{\theta_2}(x)}{p_{\theta_1}(x)} - 1 \right)^n,
$$
and the convergence is uniform with respect to $x$. 
By noting that $p_{\theta}(x)$ is a probability density function, 
we have the assertion. 

(ii) Since $\frac{p_{\theta_2}(x)}{p_{\theta_1}(x)}$ is continuous with respect to $x$, 
there exist $\delta_0 > 0$ and an open set $U_{0}$  such that 
$$
\inf_{x \in U_{0}} \frac{p_{\theta_2}(x)}{p_{\theta_1}(x)} \ge \delta_0 + 1+ r_f \ge \delta_0 + 2.
$$
Then, 
\[ a_n \int_{\frac{p_{\theta_2}(x)}{p_{\theta_1}(x)} \ge 1} \left( \frac{p_{\theta_2}(x)}{p_{\theta_1}(x)} - 1 \right)^n p_{\theta_1}(x) \mu(dx) \ge a_n (\delta_0+ r_f)^n \int_{U_0} p_{\theta_1}(x) \mu(dx). \]

Since $r_f \ge 1$, 
\[ a_n \int_{\frac{p_{\theta_2}(x)}{p_{\theta_1}(x)} < 1} \left| \frac{p_{\theta_2}(x)}{p_{\theta_1}(x)} - 1 \right|^n p_{\theta_1}(x) \mu(dx) \le a_n \left( 1 - \inf_{x \in \mathbb R} \frac{p_{\theta_2}(x)}{p_{\theta_1}(x)} \right)^n \to 0, \ n \to \infty. \]

By the assumptions, $\int_{U_0} p_{\theta_1}(x) \mu(dx) > 0$. 
Thus we see that 
\[ \lim_{n \to \infty} a_n \int_{X} \left( \frac{p_{\theta_2}(x)}{p_{\theta_1}(x)} - 1 \right)^n p_{\theta_1}(x) \mu(dx) = +\infty. \]
\end{proof}

Now we deal with the particular case of Cauchy distributions. 
We first remark that for every $(l_1, s_1)$ and $(l_2, s_2)$, 
\[ \max_{x \in \mathbb{R} \cup \{\pm\infty\}} \frac{p_{l_2,s_2}(x)}{p_{l_1,s_1}(x)} = \max_{x \in \mathbb{R} \cup \{\pm\infty\}} \frac{p_{l_1,s_1}(x)}{p_{l_2,s_2}(x)}, \]
because there exists $A \in \mathrm{SL}(2, \mathbb R)$ such that $\theta_1 = A.\theta_2$ and $\theta_2 = A.\theta_1$ where $\theta_j = \ell_j + i s_j, \ j = 1,2$.

We first deal with the case that the convergence radius is $1$. 
We denote the Kullback-Leibler, $\alpha$-divergence, Jensen-Shannon and squared Hellinger divergences by $D_\KL$, $I_{\alpha}$, $D_\JS$ and $D_H^2$, respectively. 

\begin{lemma}
(i) If $l^2 + (s - 4/5)^2 < 9/16$, 
then, $\sup_{x \in \mathbb R} \frac{p_{0,1}(x)}{p_{l,s}(x)} < 2,$ 
and hence, 
\begin{eqnarray*}
D_\KL(p_{l, s} : p_{0, 1}) &=& \sum_{n=2}^{\infty} \frac{(-1)^n}{n} D_{\chi,n}(p_{l, s} : p_{0, 1}), \\
 I_{\alpha}(p_{l, s} : p_{0, 1}) &=& \sum_{n=2}^{\infty} \frac{-4}{1-\alpha^2} \binom{(1+\alpha)/2}{n} D_{\chi,n}(p_{l, s} : p_{0, 1})\\
 D_\JS(p_{l, s} : p_{0, 1}) &=& \sum_{n=2}^{\infty} \frac{(-1)^n (2^{n-1}-1)}{n(n-1)2^{n-1}} D_{\chi,n}(p_{l, s} : p_{0, 1}), \\
D_H^2(p_{l, s} : p_{0, 1}) &=& \sum_{n=2}^{\infty} \frac{(-1)^n (2n-3)!!}{2^{n-1} n!} D_{\chi,n}(p_{l, s} : p_{0, 1}),
\end{eqnarray*}
where we used the  generalized binomial coefficient for the $\alpha$-divergences.\\
(ii) If $l^2 + (s - 4/5)^2 > 9/16$, 
then, $\sup_{x \in \mathbb R} \frac{p_{0,1}(x)}{p_{l,s}(x)} > 2,$ 
and hence, 
all of the infinite sums in (i) diverge. 
\end{lemma}

We now deal with the case that the convergence radius is $2$. 
Let $D_\HM(p:q)=\int \frac{2p(x)q(x)}{p(x)+q(x)} \dx$ be the harmonic (mean) divergence~\cite{HM-2007,dragomir2010refinement}. 

\begin{lemma}
(i) If  $l^2 + (s - 5/3)^2 < 16/9$, 
then, $\sup_{x \in \mathbb R} \frac{p_{0,1}(x)}{p_{l,s}(x)} < 3$
and hence, 
\[ D_\HM(p_{l, s} : p_{0, 1}) = \sum_{n=2}^{\infty} \frac{(-1)^{n+1}}{2^n} \int_{\mathbb R} \left( \frac{p_{0,1}(x)}{p_{l,s}(x)} - 1 \right)^n p_{l,s}(x) \dx=\sum_{n=2}^{\infty} \frac{(-1)^{n+1}}{2^n} D_{\chi,n}(p_{l, s} : p_{0, 1}). \]
(ii) If $l^2 + (s - 5/3)^2 > 16/9$, 
then, $\sup_{x \in \mathbb R} \frac{p_{0,1}(x)}{p_{l,s}(x)} > 3$
and hence, 
the infinite sum in (i) diverges. 
\end{lemma}

Other expansions are available in Table~3 of~\cite{nielsen2019power} (e.g., Jeffreys' divergence).
We refer to the Appendix~\ref{sec:maximataylor} for an implementation of the calculation of $f$-divergences using these series.

We finally consider the total variation distance between the Cauchy distributions. 
Then, we {\it cannot} expect power chi expansions. 

\begin{proposition}
Let $f(u) := \frac{|u-1|}{2}$. 
Then, for every $a_1, \cdots, a_n$, 
\[ \lim_{(l,s)\to (l_0, s_0)} \frac{I_f (p_{l,s}, p_{l_0, s_0}) - \sum_{j=2}^{n} a_j \int_{\mathbb R} \left( \frac{p_{l, s}(x)}{p_{l_0, s_0}(x)} - 1 \right)^j p_{l_0, s_0}(x) dx}{\left|\int_{\mathbb R} \left( \frac{p_{l, s}(x)}{p_{l_0, s_0}(x)} - 1 \right)^n p_{l_0, s_0}(x) \dx\right|} = +\infty. \]
\end{proposition}

\begin{proof}

\begin{lemma}\label{sup-ratio-Cauchy}
\[ \sup_{x \in \mathbb R} \left| \frac{p_{l, s}(x)}{p_{l_0, s_0}(x)} - 1 \right| = O\left( \sqrt{(l - l_0)^2 + (s - s_0)^2} \right), \ (l,s)\to (l_0, s_0). \]
\end{lemma}

\begin{proof}
We see that 
\[ \frac{p_{l, s}(x)}{p_{l_0, s_0}(x)} - 1 = \frac{s}{s_0} - 1 + \left(\frac{s}{s_0} - 1\right) \left(\frac{(x - l_0)^2 + s_0^2}{(x - l)^2 + s^2}  -1 \right) + \frac{(x - l_0)^2 + s_0^2}{(x - l) + s^2}  - 1.   \]
Since 
\[ \frac{(x - l_0)^2 + s_0^2}{(x - l)^2 + s^2}  - 1 = \frac{2 (l - l_0) (x - l) + (l - l_0)^2 + s_0^2 - s^2}{(x - l)^2 + s^2} = O\left( \sqrt{(l - l_0)^2 + (s - s_0)^2} \right), \]
we have the assertion. 
\end{proof}

By this lemma, we see that 
\[  \sum_{j=2}^{n} a_j \int_{\mathbb R} \left( \frac{p_{l, s}(x)}{p_{l_0, s_0}(x)} - 1 \right)^j p_{l_0, s_0}(x) dx = O\left( (l - l_0)^2 + (s - s_0)^2 \right). \]

On the other hand, 
\[ I_f (p_{l,s}, p_{l_0, s_0}) = \frac{2}{\pi} \arctan\left( \frac{1}{2} \sqrt{\frac{(l - l_0)^2 + (s -s_0)^2}{s s_0}} \right). \]
Hence, 
\[ \lim_{(l,s)\to (l_0, s_0)} \frac{I_f (p_{l,s}, p_{l_0, s_0})}{(l - l_0)^2 + (s - s_0)^2} =+\infty. \]

Thus we see that 
\[ \lim_{(l,s)\to (l_0, s_0)} \frac{I_f (p_{l,s}, p_{l_0, s_0}) - \sum_{j=2}^{n} a_j \int_{\mathbb R} \left( \frac{p_{l, s}(x)}{p_{l_0, s_0}(x)} - 1 \right)^j p_{l_0, s_0}(x) \dx}{(l - l_0)^2 + (s - s_0)^2} = +\infty. \]
By Lemma \ref{sup-ratio-Cauchy}, 
we see that for $n \ge 2$, 
\[ \int_{\mathbb R} \left( \frac{p_{l, s}(x)}{p_{l_0, s_0}(x)} - 1 \right)^n p_{l_0, s_0}(x) \dx = O\left( \left((l - l_0)^2 + (s - s_0)^2 \right)^{n/2} \right), \ (l,s)\to (l_0, s_0).  \]
Thus we have the assertion. 
\end{proof}

\begin{remark}
Consider the exponential family of exponential distributions $\{p_\lambda(x)=\lambda\exp(-\lambda x),\ \lambda\in\bbR_{++}\}$ defined on the positive half-line support $\calX=\bbR_+$. The criterion $\frac{p_{\theta_2}}{p_{\theta_1}}<1+r_f$ is satisfied for $\lambda_1<\lambda_2<(1+r_f)\lambda_1$.
Moreover the Pearson order-$k$ power chi divergences are available in closed form for integers $k>1$ since $\lambda_1<\lambda_2$ by adapting Lemma 3~ of~\cite{fdivchi-2013} (i.e., when $\lambda_1<\lambda_2$, it is enough to have conic natural parameter spaces instead of affine spaces).
Thus we can calculate the KLD between $p_{\lambda_1}$ and $p_{\lambda_2}$ as converging Taylor chi series.
In this case, the KLD is also known to be in closed-form as a Bregman divergence for exponential distributions:
$$
D_\KL(p_{\lambda_1}:p_{\lambda_2})=\frac{\lambda_2}{\lambda_1}-\log\frac{\lambda_2}{\lambda_1}-1.
$$
However, if we choose the exponential family of normal distributions, we cannot bound their density ratio, and therefore the Taylor chi series diverge.
\end{remark}

Notice that even if the series diverge, the $f$-divergences may be finite 
(e.g., when the ratio of densities fails to be bounded by $1+r_f$).
In that case, we cannot represent $I_f$ by a Taylor series.
By truncating the distributions, we may potentially find a validity range where to apply the Taylor expansion.

\section{Metrization of $f$-divergences between Cauchy densities}\label{sec:metrization}

Recall that $f$-divergences can always be symmetrized by taking the generator $s(u)=f(u)+uf(1/u)$.
Metrizing $f$ divergences consists in finding the largest exponent $\alpha$ such that $I_s^\alpha$ is a metric distance satisfying the triangle inequality~\cite{Kafka-1991,OsterreicherVajda-2003,Vajda-MetricDivergence-2009}.
For example, the square root of the Jensen-Shannon divergence~\cite{fuglede2004jensen} yields a metric distance which is moreover Hilbertian~\cite{acharyya2013bregman}, i.e., meaning that there is an embedding $\phi(\cdot)$ into a Hilbert space $\calH$ such that $D_\JS(p:q)=\|\phi(p)-\phi(q)\|_{\calH}$. That is, $\sqrt{\mathrm{JSD}}$ admits of Hilbert embedding.

We will show that the square roots of the Kullback-Leibler divergence and the  Bhattacharyya divergence are distances on the upper-half plane  in Theorems \ref{thm:metrization-KLD} and \ref{thm:sqrtBhat} below respectively. 
We also show that the square root of the KLD is isometrically embeddable into a Hilbert space in Theorem \ref{thm:embeddable}. 

\subsection{Metrization of the Kullback-Leibler diveregnce}

The following is a generalization of Theorem 3 in \cite{CauchyVoronoi-2020}. 

\begin{theorem}\label{thm:metrization-KLD}
Let  $0 < \alpha \leq 1$. 
Then $D_\KL(p_{\theta_1} : p_{\theta_2})^{\alpha}$ is a metric on $\bbH$ if and only if $0 < \alpha \leq 1/2$. 
\end{theorem}

In the following we also give full details of the proof of Theorem 3 in \cite{CauchyVoronoi-2020}. 

\begin{proof}
We proceed as in~\cite{CauchyVoronoi-2020} by letting 
$$
t(u) := \log\left( \frac{1+\cosh(\sqrt{2} u)}{2} \right), u \geq 0. 
$$
Let us consider the properties of $F_2 (u) := t(u)^{\alpha}/u$.

$$ 
F_2^{\prime}(u) = -2\frac{t(u)^{\alpha-1}}{u^2} G(u/\sqrt{2}),
$$
where 
$$ G_2(w) := (2+e^{2w} + e^{-2w}) \log\left( \frac{e^w + e^{-w}}{2}  \right) - \alpha w (e^{2w} - e^{-2w}). $$
If we let $x := e^w$, then, 
$$ G_2(w) = (x +x^{-1}) \left( (x +x^{-1}) \log (\frac{x^2 +1}{2x} ) - \alpha (x - x^{-1}) \log x \right).
$$
Let 
$$ H_2(x) := x \left( (x +x^{-1}) \log (\frac{x^2 +1}{2x} ) - \alpha (x - x^{-1}) \log x \right).
$$
Then, $H_2(1) = 0$ and 
$$
H_2^{\prime}(x) = 4\left( x \log (\frac{x^2 +1}{2} ) - (1+\alpha)x \log x + \frac{x^3}{x^2 +1} - \alpha x   \right). 
$$
Let 
$$
I_2(x) := x \log (\frac{x^2 +1}{2} ) - (1+\alpha)x \log x + \frac{x^3}{x^2 +1} - \alpha x. 
$$
Then, 
$I_2(1) = 1/2 - \alpha$ and 
$$
I_2^{\prime}(x) = \log (\frac{x^2 +1}{2} ) - (1+\alpha) \log x + \frac{x^2 (3x^2 + 5)}{(x^2 +1)^2} - (1+2\alpha). 
$$

Consider the case that $\alpha > 1/2$. 
Then, $I_2(x) < 0$ for every $x > 1$ which is sufficiently close to $1$. 
Hence, $G_2(w) < 0$ for every $w > 0$ which is sufficiently close to $0$. 
Hence, $F_2^{\prime}(u) > 0$ for every $u > 0$ which is sufficiently close to $0$. 
This means that $F_2$ is strictly increasing near the origin.

Hence there exists $u_0 > 0$ such that 
$$ 2 t(u_0)^{\alpha} < t(2u_0)^{\alpha}. $$
Take $x_0, z_0 \in \bbH$ such that $\rho_{\FR}(x_0, z_0) = 2u_0$, where $\rho_{\FR}$ is the Fisher metric distance on $\bbH$. 
By considering the geodesic between $x_0$ and $z_0$, we can take $y_0 \in \bbH$ such that $\rho_{\FR}(x_0, y_0) = \rho_{\FR}(y_0, z_0) = u_0$. 

Finally we consider the case that $\alpha = 1/2$. 
Let 
$$ J_2(x) := (x^2 +1)^2 \log (\frac{x^2 +1}{2} ) - \frac{3}{2} (x^2 +1)^2 \log x + x^2 (3x^2 + 5) - 2(x^2 +1)^2.   $$
Then, $J_2(1) = 0$. 
If we let $y := x^2$, then, 
$$ J_2(x) = (y +1)^2 \log \left(\frac{y +1}{2} \right) - \frac{3}{4} (y +1)^2 \log y +(y^2 + y -2).    $$
Let $K_2(y) := J(\sqrt{y})$. 
Then, 
\begin{eqnarray*}
K_2^{\prime}(y) &=& 2(y +1) (\log (\frac{y +1}{2} ) +1)- \frac{3}{2} (y +1) \log y - \frac{3(y +1)^2}{4y}  +(2y+1),\\
 &=& y + (y+1)\left(2 \log \left(y +1 \right) - \frac{3}{2} \log y + \frac{9}{4} - \frac{3}{4y} - 2\log 2\right). 
\end{eqnarray*}
If $y > 1$, then, 
$$ 2 \log \left(y +1 \right) > \frac{3}{2} \log y
$$
and 
$$ \frac{9}{4} - \frac{3}{4y} - 2\log 2 > \frac{3}{2} - 2\log 2 > 0.$$
Then, $J_2(x) > J(1) = 0$ for every $x > 1$. 
Hence, $I_2(x) > I(1) = 0$ for every $x > 1$. 
Hence, $G_2(w) > 0$ for every $w > 0$. 
Hence, $F_2^{\prime}(u) < 0$ for every $u > 0$. 
This means that $F_2$ is strictly decreasing on $[0, \infty)$.
Thus we proved that $D_\KL(p_{\theta_1} : p_{\theta_2})^{1/2}$ gives a distance, hence $D_\KL(p_{\theta_1} : p_{\theta_2})^{\alpha}$ is also a distance for every $\alpha \in (0, 1/2)$.
\end{proof}


\subsection{Metrization of the Bhattacharyya divergence}

The Bhattacharyya divergence~\cite{bhattacharyya1943measure} is defined by
$$
D_{\Bhat}(p:q) := -\log \left( \int \sqrt{p(x)q(x)} \dx \right). 
$$

The term $\int \sqrt{p(x)q(x)} \dx$ is called the Bhattacharyya coefficient.
It is easy to see that $D_{\Bhat}(p:q) = 0$ iff $p=q$, and $D_{\Bhat}(p:q) =D_{\Bhat}(q:p)$. 

\begin{theorem}\label{thm:sqrtBhat}
$\sqrt{D_{\Bhat}(p_{\theta_1}:p_{\theta_2})}$ is a distance on $\mathbb H$.
\end{theorem}

For exponential families, see \cite[Proposition 2]{CauchyVoronoi-2020} and \cite{nielsen2011burbea}. 
We cannot apply the method of \cite[Proposition 2]{CauchyVoronoi-2020} in a direct manner. 
We state the reason in the end of this section. 
We can also show that $D_{\Bhat}(p_{\theta_1}:p_{\theta_2})^{\alpha}$ is not a metric if $\alpha > 1/2$ in the same manner as in the proof of Theorem \ref{thm:metrization-KLD}.

\begin{proof}
We show the triangle inequality. 
We follow the idea in the proof of Theorem 3 in~\cite{CauchyVoronoi-2020}. 
We construct the metric transform $t_{\FR \rightarrow \Bhat}$ and show that $t_{\FR\rightarrow \Bhat}(s)$ is increasing and $\sqrt{t_{\FR \rightarrow \Bhat}(s)}/s$ is decreasing. 

Let $\rho_{\FR}$ be the Fisher-Rao distance. 
Then, by following the argument in the proof of   \cite[Theorem 3]{CauchyVoronoi-2020}, 
\[ \chi(z,w) = F_3(\rho_{\FR}(z,w)), \]
where we let 
\[ F_3(s) := \cosh(\sqrt{2} s) -1. \]

Let 
\[ I_3(z,w) := \int \sqrt{p_z (x) p_w(x)} dx. \]
Then, by the invariance of the $f$-divergences, 
\[ I_3(A.z, A.w) = I_3(z,w). \]
Hence we have that for some function $J_3$, $J_3(\chi(z,w)) = I_3(z,w)$. 
Hence, 
\[ \sqrt{D_{\Bhat}(p_{\theta_1}:p_{\theta_2})} = \sqrt{-\log J_3\left(F_3(\rho_{\FR}(\theta_1, \theta_2))\right)}. \]
We have that 
\[ t_{\FR \rightarrow \Bhat}(s) = -\log J_3(F_3(s)).  \]

It holds that for every $a \in (0,1)$, 
\[ J\left( \chi(ai,i)) \right) = I(ai,i). \]

By the change-of-variable $x = \tan \theta$ in the integral of $I(a i, i)$, 
it is easy to see that 
\[ I_3(ai,i) = \frac{2 \sqrt{a} \mathbf{K}(1-a^2)}{\pi}, \]
where $\mathbf{K}$ is the elliptic integral of the first kind. 
It is defined by\footnote{This is a little different from the usual definition. The usual one is $\mathbf{K}(t) = \int_0^{\pi/2} \frac{1}{\sqrt{1 - t^2 \sin^2 \theta}} d\theta$.} 
\[ \mathbf{K}(t) := \int_0^{\pi/2} \frac{1}{\sqrt{1 - t \sin^2 \theta}} d\theta, \ 0 \le t < 1. \]

Hence, 
\[ J_3 \left(\frac{(1-a)^2}{2a}\right) = \frac{2 \sqrt{a} \mathbf{K}(1-a^2)}{\pi}. \]

Since $$F_3(s) = \cosh(\sqrt{2} s) -1 = \frac{(1 - e^{-\sqrt{2} s})^2}{2 e^{-\sqrt{2} s}}, $$
we have that 
\[ J_3(F_3(s)) = \frac{2 e^{-s/\sqrt{2}} \mathbf{K}(1-e^{-2\sqrt{2} s})}{\pi}. \]
Since the above function is decreasing with respect to $s$, 
$t_{\FR \rightarrow \Bhat}(s)$ is increasing. 

Furthermore, we have that 
\begin{equation}\label{eq:sqrt-Bhat} 
\frac{\sqrt{t_{\FR \rightarrow \Bhat}(s)}}{s} = \sqrt{{-\frac{1}{s^2}}\log \left(  \frac{2 e^{-s/\sqrt{2}} \mathbf{K}(1-e^{-2\sqrt{2} s})}{\pi} \right) }. 
\end{equation}

This function is decreasing with respect to $s$. 
See Figure \ref{fig:fig1}. 
We can show this fact by using the results for the complete elliptic integrals. 
The full proof is somewhat complicated. 
See Section \ref{sec:cei}.
\end{proof}

\begin{remark}
It holds that 
\[ \lim_{s \to +0} \frac{\sqrt{t_{\FR \rightarrow \Bhat}(s)}}{s} = \frac{1}{8}, \textup{ and } \lim_{s \to +\infty} \frac{\sqrt{t_{\FR \rightarrow \Bhat}(s)}}{s} = 0. \]
\end{remark}

\begin{remark}
The squared Hellinger distance $H^2(p:q):=\frac{1}{2} \int \left(\sqrt{p(x)}-\sqrt{q(x)}\right)^2\dx$ (an $f$-divergence for $f_{\mathrm{Hellinger}}(u)=\frac{1}{2}(\sqrt{u}-1)^2$) satisfies that
\[ H^2(p_{\theta_1} : p_{\theta_2}) = 1 - \exp\left( -D_{\Bhat}(p_{\theta_1}:p_{\theta_2}) \right) = 1 - J_3(F_3(\rho_{\FR}(\theta_1, \theta_2))) \]
\[=  1- \frac{2 e^{- \rho_{\FR}(\theta_1, \theta_2)/\sqrt{2}} \mathbf{K}(1-e^{-2\sqrt{2} \rho_{\FR}(\theta_1, \theta_2)})}{\pi} \]
\[ = 1-  \frac{2 K\left( 1 - \left(1+ \chi(\theta_1, \theta_2) + \sqrt{\chi(\theta_1, \theta_2)(2+\chi(\theta_1, \theta_2))}\right)^{-2}\right)}{\pi \sqrt{1+ \chi(\theta_1, \theta_2) + \sqrt{\chi(\theta_1, \theta_2)(2+\chi(\theta_1, \theta_2))}}}. \label{eq:squaredHellinger} \]
The Hellinger distance $H(p_{\theta_1} : p_{\theta_2})$ is known to be a metric distance.
Notice that 
$$
h_{f_{\mathrm{Hellinger}}}(u)=1-  \frac{2 \mathbf{K}\left( 1 - \left(1+ u + \sqrt{u(2+u)}\right)^{-2}\right)}{\pi \sqrt{1+u + \sqrt{u(2+u)}}}
$$ 
and we check that $h_{f_{\mathrm{Hellinger}}}(0)=0$ 
since $\mathbf{K}(0)=\frac{\pi}{2}$.
\end{remark}

\begin{figure*}
\begin{center}
\includegraphics[width= 8cm, height= 6cm, bb= 0 0 846 594]{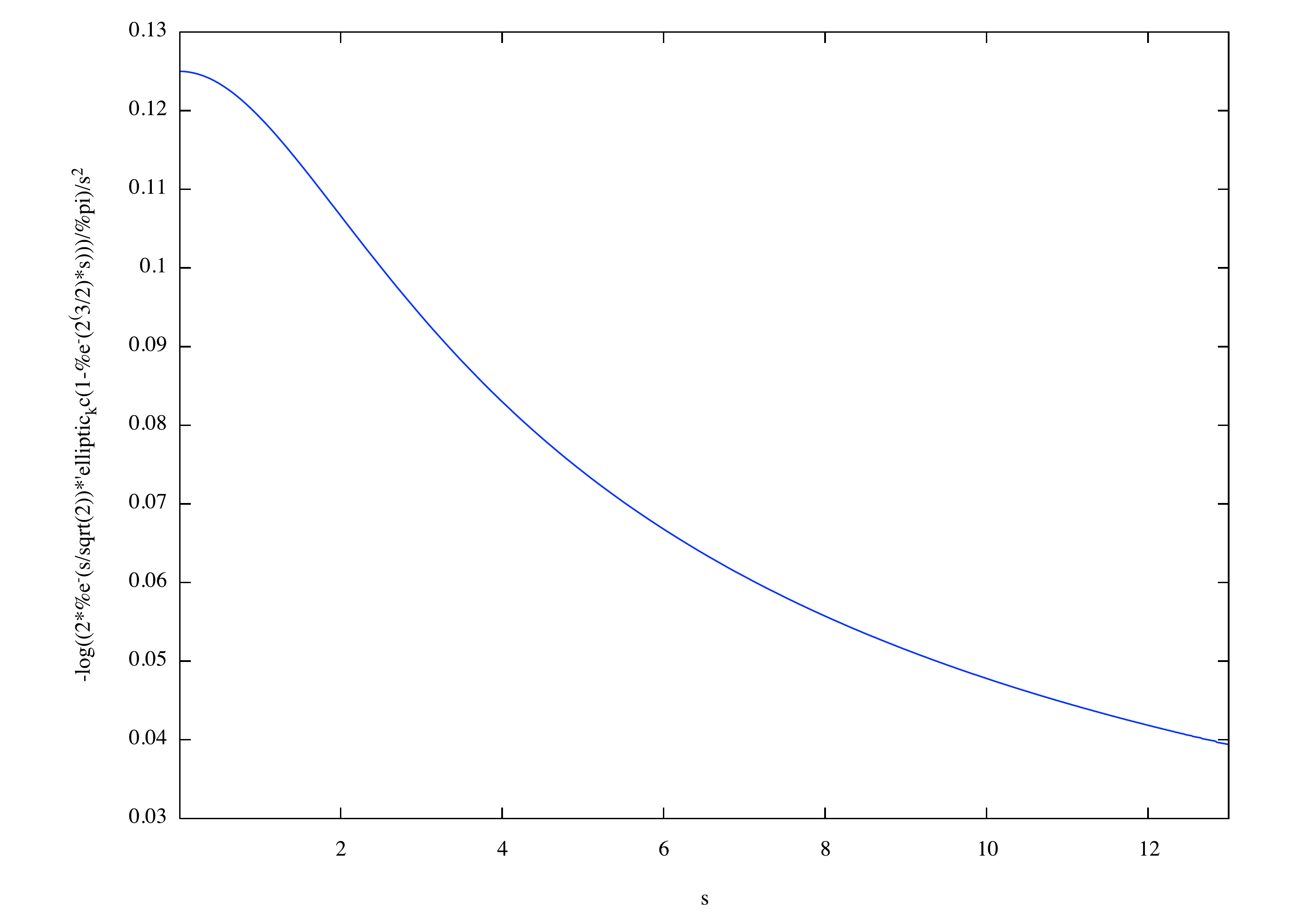}
\caption{Graph of $\frac{\sqrt{t_{\FR \rightarrow \Bhat}(s)}}{s}$}\label{fig:fig1}
\end{center}
\end{figure*}

\begin{remark}
More generally, let $\BC_\alpha[p:q]:=\int_\bbR p(x)^\alpha q(x)^{1-\alpha}\dx$ denote the $\alpha$-skewed Bhattacharyya coefficient for 
$\alpha\in \bbR\backslash\{0,1\}$ (also called the $\alpha$-Chernoff coefficient~\cite{Chernoff-2011,Chernoff-2013}).
The $\alpha$-skewed Bhattacharyya divergence is defined by
$$
D_{\Bhat,\alpha}(p:q) := -\log  \BC_\alpha[p:q] = -\log \int_\bbR p(x)^\alpha q(x)^{1-\alpha}\dx.   
$$

Using a computer algebra system\footnote{\url{https://maxima.sourceforge.io/}}, we can compute the $\alpha$-skewed Bhattacharyya coefficients for integers $\alpha$ in closed form.
For example, we find  the following closed-form for the definite integrals:
{\tiny
\begin{eqnarray*}
\BC_2[p:p_{l,s}]&=& \frac{{{s}^{2}}+{{l}^{2}}+1}{2 s},\\
\BC_3[p:p_{l,s}]&=& \frac{3 {{s}^{4}}+\left( 6 {{l}^{2}}+2\right) \, {{s}^{2}}+3 {{l}^{4}}+6 {{l}^{2}}+3}{8 {{s}^{2}}},\\
\BC_4[p:p_{l,s}]&=&\frac{5 {{s}^{6}}+\left( 15 {{l}^{2}}+3\right) \, {{s}^{4}}+\left( 15 {{l}^{4}}+18 {{l}^{2}}+3\right) \, {{s}^{2}}+5 {{l}^{6}}+15 {{l}^{4}}+15 {{l}^{2}}+5}{16 {{s}^{3}}}, \ \ \ \ and,\\
\BC_5[p:p_{l,s}] &=&\frac{35 {{s}^{8}}+\left( 140 {{l}^{2}}+20\right) \, {{s}^{6}}+\left( 210 {{l}^{4}}+180 {{l}^{2}}+18\right) \, {{s}^{4}}+\left( 140 {{l}^{6}}+300 {{l}^{4}}+180 {{l}^{2}}+20\right) \, {{s}^{2}}+35 {{l}^{8}}+140 {{l}^{6}}+210 {{l}^{4}}+140 {{l}^{2}}+35}{128 {{s}^{4}}}.
\end{eqnarray*}
}
\end{remark}

Furthermore, we give some remarks about the complete elliptic integrals of the first and second kinds.

\begin{remark}
(i) In practice, we can calculate efficiently $\mathbf{K}(t)$ using the arithmetic-geometric mean (AGM): 
$$
\mathbf{K}(t)=\frac{\pi}{ 2 \mathrm{AGM}(1,\sqrt{1-t^2})}
$$ 
where 
$\mathrm{AGM}(a,b)=\lim_{n\rightarrow \infty} a_n=\lim_{n\rightarrow \infty} g_n$ with $a_0=a$, $g_0=b$, $a_{n+1}=\frac{a_n+g_n}{2}$ and $g_{n+1}=\sqrt{a_ng_n}$.
The mean is called the arithmetic-geometric mean because it falls in-between the geometric mean and the arithmetic mean:
  $g_n\leq \mathrm{AGM}(a,b)\leq a_n$, where $g_n$ is an increasing sequence and $a_n$ is a decreasing sequence.
We see that 
$$
\mathrm{AGM}(a,b)=\frac{\pi}{4} \frac{a+b}{\mathbf{K}\left(\frac{a-b}{a+b}\right)}.
$$
One way to show this relation is using the invariance of the Cauchy distribution with respect to the Boole transform which is mentioned in Section \ref{sec:igls}. \\
(ii) Let $\mathbf{K}$ and $\mathbf{E}$ be the complete elliptic integrals of the first and second kinds respectively. 
We let\footnote{This is also a little different from the usual definition. The usual one is $E(t) := \int_0^{\pi/2} \sqrt{1 - t^2 \sin^2 \theta} d\theta.$} 
\[ \mathbf{E}(t) := \int_0^{\pi/2} \sqrt{1 - t \sin^2 \theta} d\theta. \]
The following expansion by C. F. Gauss in 1818 is well-known: 
\[ 1 - \frac{\mathbf{E}(x)}{\mathbf{K}(x)}  = \frac{x}{2} + \sum_{n \ge 1} 2^{n-1} (a_{n} - b_{n})^2, \ \ x \in (0,1), \]
where $(a_0, b_0) = (1,\sqrt{1-x})$ and $(a_{n+1}, b_{n+1}) = \left(\frac{a_n + b_n}{2}, \sqrt{a_n b_n}\right)$, $n \ge 0$. 
See \cite{Salamin} for more details. 

By investigating of the behaviors of $\frac{\sqrt{t_{\FR \rightarrow \Bhat}(s)}}{s}$ in Eq.~\ref{eq:sqrt-Bhat}, 
we get some approximation formulae of $1 - \frac{\mathbf{E}(x)}{\mathbf{K}(x)}$.  
See Lemma \ref{Gauss-AGM} below for example. 
By numerical computations, 
it holds that 
\[ 1 - \frac{\mathbf{E}(x)}{\mathbf{K}(x)}  =  \frac{x}{2} + \frac{x^2}{16} + \frac{x^3}{32} + \frac{41}{2048}x^4 +  \frac{59}{4096} x^5 + \frac{727}{65536} x^6 + O(x^7),  \]
\[ x \left( \frac{3}{2} + 4\frac{\log(2 K(x)/\pi)}{\log(1-x)} \right) = \frac{x}{2} + \frac{x^2}{16} + \frac{x^3}{32} + \frac{251}{12288}x^4 +  \frac{123}{8192} x^5 + \frac{34781}{2949120} x^6 + O(x^7)  \]
and 
\[ \frac{x\left(4-x-\sqrt{(4-3x)^2 + 4(2-x)(1-x)\log(1-x)}\right)}{4x+2(x-1)\log(1-x)} \]
\[= \frac{x}{2} + \frac{x^2}{16} + \frac{x^3}{32} + \frac{49}{3072} x^4 + \frac{41}{6144} x^5 + \frac{259}{491520} x^6 + O(x^7).  \]
See also \cite[Lemma 6.2]{KMY}. 
They are very close to each other  if $x > 0$ is close to $0$. 
For just a few of recent results about complete elliptic integrals and its applications, see \cite{KMY}, \cite{YQCZ} and the references therein. 
 
Table~\ref{tab:summary} summarizes the symmetric closed-form $f$-divergences 
$I_f(p_{\lambda}:p_{\lambda'})=h_f(\chi[p_{\lambda}:p_{\lambda'}])$ 
 between two univariate Cauchy densities $p_{\lambda}$ and $p_{\lambda'}$
that we obtained as a function $h_f$ of the chi-squared divergence $\chi[p_{\lambda}:p_{\lambda'}]=\frac{\|\lambda-\lambda'\|^2}{2\lambda_2\lambda_2'}$ (with $h_f(0)=0$).
\end{remark}

\begin{remark}
The proof of \cite[Proposition 2]{CauchyVoronoi-2020} is {\it not} applicable to the proof of Theorem \ref{thm:sqrtBhat} above, because it cannot be a Bregman divergence. 
See \cite{acharyya2013bregman}.
\end{remark}


\begin{table}
\centering
{\renewcommand{\arraystretch}{1.5}
\begin{tabular}{|l|l|l|}\hline
$f$-divergence name & $f(u)$ & $h_f(u)$ for $I_f[p_{\lambda_1}:p_{\lambda_2}]=h_f(\chi[p_{\lambda_1}:p_{\lambda_2}])$ \\ \hline\hline
Chi squared divergence & $(u-1)^2$ & $u$\\
Total variation distance & $\frac{1}{2}|u-1|$ & $\frac{2}{\pi}\arctan\left(\sqrt{\frac{u}{2}}\right)$\\
Kullback-Leibler divergence & $-\log u$ & $\log(1+\frac{1}{2}u)$\\
Jensen-Shannon divergence & $\frac{u}{2}\log\frac{2u}{1+u}-\frac{1}{2}\log\frac{1+u}{2}$ & $\log\left( \frac{2\sqrt{2+u}}{\sqrt{2+u}+\sqrt{2}} \right)$\\
Taneja $T$-divergence & $\frac{u+1}{2}\log\frac{u+1}{2\sqrt{u}}$ & $\log \left( \frac{1+\sqrt{1+\frac{u}{2}}}{2} \right)$,\\
LeCam-Vincze divergence & $\frac{(u-1)^2}{1+u}$ & $2-4\sqrt{\frac{1}{2(u+2)}}$\\
squared Hellinger divergence & $\frac{1}{2}(\sqrt{u}-1)^2$ & $1-  \frac{2 K\left( 1 - \left(1+ u + \sqrt{u(2+u)}\right)^{-2}\right)}{\pi \sqrt{1+u + \sqrt{u(2+u)}}}$ \\ \hline
\end{tabular}
}

\caption{Closed-form $f$-divergences between two univariate Cauchy densities expressed as a function $h_f$ of the chi-squared divergence $\chi[p_{\lambda}:p_{\lambda'}]=\frac{\|\lambda-\lambda'\|^2}{2\lambda_2\lambda_2'}$. 
The square root of the KLD, LeCam and squared Hellinger divergences between Cauchy densities yields metric distances.\label{tab:summary}  }
\end{table}

\subsection{The Chernoff information}\label{sec:CID}

The Chernoff information~\cite{Chernoff-2013} between two densities $p_1$ and $p_2$ is defined by:
$$
C(p_1 : p_2) := -\log \min_{a \in (0,1)} \int p_1 (x)^a p_2(x)^{1-a} \dx. 
$$

The Chernoff information provides an upper bound for the error probabilities of Bayes hypothesis testing~\cite{CT-2012} (Chapter 11). 

\begin{theorem}
For the univariate Cauchy location-scale families, the Chernoff information is equal to the Bhattacharyya divergence. 
\end{theorem}

\begin{proof}
Let 
$$
\Lambda(a) := \log \int_{\mathbb R} p_{\theta_1} (x)^a p_{\theta_2} (x)^{1-a} \dx.
$$
This is finite for every $\bbR$, and is in  $C^{\infty}$ class on $\bbR$. 

We see that for every $a \in \mathbb R$, 
$$
\Lambda^{\prime}(a)  = \frac{ \int_{\mathbb R}  p_{\theta_1} (x)^a p_{\theta_2} (x)^{1-a} \log \frac{p_{\theta_1}(x)}{p_{\theta_2}(x)} \dx}{\int_{\mathbb R}  p_{\theta_1} (x)^a p_{\theta_2}(x)^{1-a} \dx}.
$$

By the symmetry of $f$-divergences, 
$$
\int_{\mathbb R}  p_{\theta_1} (x)^a p_{\theta_2} (x)^{1-a} \log \frac{p_{\theta_1}(x)}{p_{\theta_2}(x)} \dx =
 \int_{\mathbb R}  p_{\theta_2} (x)^a p_{\theta_1} (x)^{1-a} \log \frac{p_{\theta_2}(x)}{p_{\theta_1}(x)} \dx. 
$$

Hence, for $a=1/2$, 
\begin{equation}\label{eq:sym} 
\int_{\mathbb R}  p_{\theta_1} (x)^{1/2} p_{\theta_2} (x)^{1/2} \log \frac{p_{\theta_1}(x)}{p_{\theta_2}(x)} dx = 0. 
\end{equation} 
Hence, 
$ \Lambda^{\prime}(1/2) = 0$. 
By the Cauchy-Schwarz inequality, $ \Lambda^{\prime\prime}(a) \ge 0$. 
Hence $\Lambda(a)$ takes its minimum at $a=1/2$. 
\end{proof}

Thus the Chernoff information between two Cauchy distributions $p_{\lambda_1}$ and $p_{\lambda_2}$ can be computed from the Bhattacharyya coefficient $\BC_\alpha[p_{\lambda_1}:p_{\lambda_2}]:=\int \sqrt{p_{\lambda_1}(x)p_{\lambda_2}(x)}\dx$:
$$
C(p_{\lambda_1} : p_{\lambda_2})=-\log \BC_\alpha[p_{\lambda_1}:p_{\lambda_2}].
$$

Since the Bhattacharyya coefficient can be recovered from the squared Hellinger divergence:
$$
\BC_\alpha[p_{\lambda_1}:p_{\lambda_2}]=1-H^2 (p_{\lambda_1}:p_{\lambda_2}),
$$ 
we use the closed-form of the squared Hellinger divergence (Eq.~\ref{eq:squaredHellinger}) to recover the closed-form formula of the  Bhattacharyya coefficient.
The Bhattacharyya and Chernoff divergences are not $f$-divergences because they are not separable divergences.
Nevertheless, by abuse of notation, let us write $h_{\mathrm{Chernoff}}(u)=-\log \left(1-h_{\mathrm{Hellinger}}(u)\right)$.

\begin{remark}
We can compute Eq. \ref{eq:sym} by using the two formulas 4.386.3 and 4.386.4 in p.~588 of~\cite{gradshteyn2014table}.
However such approach is much more tedious than the above proof. 
\end{remark}

Additional material is available at \url{https://franknielsen.github.io/CauchyFdivergences/}

\section{Geometric properties of the metrizations of $f$-divergences}\label{sec:properties}

If a divergence $D$ is given, then we can define an associated Riemannian metric $g_D$ on the parameter space by following Eguchi \cite{Eguchi1983, Eguchi1992}. 
(See also Remark~\ref{rmk:connection}.)
Specifically, by regarding $D$ is a smooth function on $M \times M$ where $M$ is the space of parameters, 
we let 
\[ (g_D)_r (X_r, Y_r)  := -X_p Y_q D(p,q) |_{p=q=r}, \ r \in M, \]
where $X, Y$ are vector fields on $M$. 

It is known that if $D$ is the Kullback-Leibler divergence, then, $g_D$ is the Fisher metric. 
If $D$ is not the Kullback-Leibler divergence, then, we are not sure whether $g_D$ is the Fisher metric. 
However, $g_D$ is the Fisher metric for every smooth $f$-divergence between the Cauchy distribution. 

\begin{proposition}
Let $D_f$ be the $f$-divergence between the univariate Cauchy densities. 
Let $F$ be a function such that 
$$ D_f (p_{\theta_1} : p_{\theta_2}) = F(\chi(\theta_1, \theta_2)), \ \theta_1, \theta_2 \in \mathbb H.$$
Assume that $F$ is in $C^2 ([0, \infty))$. 
Then, the Riemannian metric $g_D$ is $F^{\prime}(0) \rho$, where $\rho$ is the Poincar\'e metric on $\mathbb H$. 
\end{proposition}

For the (dual) connections induced by the $f$-divergence, see Remark~\ref{rmk:connection}. 
We remark that $\sqrt{2}\rho_{\textup{FR}}$ is identical with the Poincar\'e distance on $\mathbb H$.  

\begin{proposition}\label{prop:nongeo}
Let $\sqrt{D_{\textup{Bhat}}}$ and $\sqrt{D_{\textup{KL}}}$ be the distances between Cauchy densities. 
Then, neither $(\mathbb H, \sqrt{D_{\textup{Bhat}}})$ nor $(\mathbb H, \sqrt{D_{\textup{KL}}})$ is a geodesic metric space. 
\end{proposition}

\begin{proof}
Recall that $p_z (x) = \frac{\textup{Im}(z)}{\pi |x-z|^2}, z \in \mathbb H$. 

Assume that $(\mathbb H, \sqrt{D_{\textup{KL}}})$ is a geodesic metric space. 
Then, for every $A > 0$, 
there exists a continuous map $\gamma : [0,1] \to \mathbb H$ such that $\gamma(0) = i, \gamma(1) = Ai$, and 
\[ \sqrt{D_{\textup{KL}}\left(p_{i} : p_{Ai}\right)} = \sqrt{D_{\textup{KL}}(p_{i} : p_{\gamma(t)})} + \sqrt{D_{\textup{KL}}(p_{\gamma(t)} : p_{Ai})}  \]
for every $t \in (0,1)$.

Let $\widetilde \gamma (t) := \textup{Im}(\gamma (t)) i$. 
Then, 
\[ \chi(\widetilde \gamma (t_1), \widetilde \gamma (t_2)) \le \chi(\gamma (t_1), \gamma (t_2)), \ \ t_1, t_2 \in [0,1]. \]
Since $ \sqrt{D_{\textup{KL}}\left(p_{z} : p_{w}\right)}$ is increasing as a function of  $\chi(z,w)$, 
\[ D_{\textup{KL}}\left(p_{\widetilde \gamma (t_1)}: p_{\widetilde \gamma (t_2)}\right) \le D_{\textup{KL}}(p_{\gamma (t_1)}: p_{\gamma (t_2)}), \ \ t_1, t_2 \in [0,1]. \]
Since $\sqrt{D_{\textup{Bhat}}}$ is a distance, 
\[ \sqrt{D_{\textup{KL}}\left(p_{i} : p_{Ai}\right)} = \sqrt{D_{\textup{KL}}(p_{i} : p_{\widetilde\gamma(t)})} + \sqrt{D_{\textup{KL}}(p_{\widetilde\gamma(t)} : p_{Ai})}  \]
for every $t \in (0,1)$. 
Since $\widetilde \gamma$ is continuous, 
we see that 
\[ \sqrt{D_{\textup{KL}}\left(p_{i} : p_{Ai}\right)} = \sqrt{D_{\textup{KL}}(p_{i} : p_{Bi})} + \sqrt{D_{\textup{KL}}(p_{Bi} : p_{Ai})}, \ \  \ B \in (1,A),\]  
by the intermediate value theorem.

Let $a > 0$. 
Then, 
\[ \rho_{\textup{FR}}(i, a^2 i) = \rho_{\textup{FR}}(i, a i) + \rho_{\textup{FR}}(a i, a^2 i) = 2\rho_{\textup{FR}}(i, ai). \]
Hence, 
\[ \frac{7}{5} \sqrt{\rho_{\textup{FR}}(i, a^2 i)} < \sqrt{\rho_{\textup{FR}}(i, a i)} + \sqrt{\rho_{\textup{FR}}(a i, a^2 i)}. \]
Since 
\begin{equation}\label{eq:KL-infty}  
\lim_{\chi(z,w) \to \infty} \frac{D_{\textup{KL}}(p_z, p_w)}{\rho_{\textup{FR}}(z,w)} = \frac{1}{\sqrt{2}}, 
\end{equation}
we see that 
\[ \sqrt{D_{\textup{KL}}(p_i : p_{a^2 i})} < \sqrt{D_{\textup{KL}}(p_i : p_{a i})} + \sqrt{D_{\textup{KL}}(p_{ai} : p_{a^2 i})} \]
for sufficiently large $a > 0$.
Thus we see that  $(\mathbb H, \sqrt{D_{\textup{KL}}})$ is not a geodesic metric space. 

The proof for $\sqrt{D_{\textup{Bhat}}}$ goes in the same manner, because 
\begin{equation}\label{eq:Bhat-infty}  
\lim_{\chi(z,w) \to \infty} \frac{D_{\textup{Bhat}}(p_z:p_w)}{\rho_{\textup{FR}}(z,w)} = \frac{1}{\sqrt{2}}. 
\end{equation}

\end{proof}

\begin{proposition}\label{prop:complete}
The metric spaces $(\mathbb H, \sqrt{D_{\textup{KL}}})$ and $(\mathbb H, \sqrt{D_{\textup{Bhat}}})$ are both complete. 
\end{proposition}

\begin{proof}
Assume that $(z_n)_n$ is a Cauchy sequence with respect to $ \sqrt{D_{\textup{KL}}}$. 
Since $\sqrt{D_{\textup{KL}}\left(p_{z} : p_{w}\right)}$ is increasing as a function of  $\chi(z,w)$, 
we see that 
$\chi(z_n, z_m) \to 0, n,m \to \infty$. 
We see that 
$\chi(z,w) \le \delta$ if and only if 
$$\left|w - (\textup{Re}(z) + i (1+\delta)\textup{Im}(z)) \right| \le \sqrt{\delta (\delta+2)} \textup{Im}(z). $$
Hence $(z_n)_n$ is bounded. 
Let $z$ be an accumulation point of $(z_n)_n$. 
Then, $z_{k_n} \to z, \ n \to \infty$ with respect to the Euclid distance. 
Hence, 
$\chi(z_{k_n}, z) \to 0, n \to \infty$. 
Hence, 
$\sqrt{D_{\textup{KL}}\left(p_{z_{k_n}} : p_{z}\right)} \to 0, n \to \infty$. 
Since $(z_n)_n$ is a Cauchy sequence with respect to $ \sqrt{D_{\textup{KL}}}$, 
we see that 
$\sqrt{D_{\textup{KL}}\left(p_{z_{n}} : p_{z}\right)} \to 0, n \to \infty$. 
\end{proof}

Now by Hopf-Rinow's theorem (see \cite[Theorem 16]{Petersen2006}) and Propositions \ref{prop:nongeo} and \ref{prop:complete}, 
\begin{proposition}\label{prop:nonRiem}
Let $\sqrt{D_{\textup{Bhat}}}$ and $\sqrt{D_{\textup{KL}}}$ be the distances between Cauchy densities. 
Then, neither $\sqrt{D_{\textup{Bhat}}}$ or $\sqrt{D_{\textup{KL}}}$ between Cauchy densities is a Riemannian distance. 
\end{proposition}

\begin{remark}[alternative proof of Proposition \ref{prop:nonRiem}]
For $A \in SL(2, \mathbb R)$, 
let $\varphi_A (\theta) := A.\theta, \ \theta \in \mathbb H$. 
We first remark that every Riemannian distance $d$ on $\mathbb H$ which is preserved by every $\varphi_A$ has a form of $c \rho$ for some non-negative constant $c$. 
This is shown by two classical results in Riemannian geometry. 
We remark that every $SL(2, \mathbb{R})$ action to $\mathbb H$ is smooth and bijective. 
By the Myers-Steenrod theorem (see \cite[Theorem 18]{Petersen2006}), every $\varphi_A$ is a Riemannian isometry with respect to the Riemannian metric associated with $d$. 
It is well-known that 
if a Riemannian metric on $\mathbb H$ is a Riemannian isometry for every $\varphi_A$, 
then, it has a form of $c\rho_{\textup{FR}}$ for some constant $c$. 
This is usually stated in the much more general framework for homogeneous spaces. 
See \cite[Proposition X.3.1 and Theorem XI.8.6]{KobayashiNomizu} for example. 
By \eqref{eq:KL-infty} and \eqref{eq:Bhat-infty}, 
neither $\sqrt{D_{\textup{KL}}}$ or $\sqrt{D_{\textup{Bhat}}}$ has a form of $c\rho_{\textup{FR}}$ for some constant $c$. 
\end{remark}

We finally consider isometric embedding into a Hilbert space. 

\begin{theorem}\label{thm:embeddable}
The square root of the Kullback-Leibler divergence between Cauchy densities is isometrically embeddable  into a Hilbert space. 
\end{theorem}

\begin{proof}
By \cite{Schoenberg1938}, 
it suffices to show that 
$$\sum_{i,j=1}^{n} c_i c_j D_{\textup{KL}} (p_{z_i}:p_{z_j}) \le 0$$ for every $(c_1, \cdots, c_n)$ such that $\sum_{i=1}^{n} c_i = 0$ and every $z_1, \cdots, z_n \in \Theta$. 

Let the hyperboloid model be 
\[ \mathbb L  := \{(x,y,z) \in \mathbb R^3 : z > 0, x^2 + y^2 - z^2 = -1\}. \]
Let 
\[ d_{\mathbb L} \left((x_1,y_1,z_1), (x_2,y_2,z_2)\right) := \cosh^{-1}\left( z_1 z_2 - x_1 x_2 - y_1 y_2 \right), \ \ (x_1,y_1,z_1), (x_2,y_2,z_2) \in \mathbb L.\]

Let $\phi_1 : \mathbb L \to \mathbb D$ be the map defined by 
\[ \phi_1(x,y,z) = \left(\frac{x}{1+z}, \frac{y}{1+z} \right). \]
 Let $\phi_2 : \mathbb D \to \mathbb H$ be the map defined by 
\[ \phi_2 (x,y) = \left(-\frac{2y}{(1-x)^2 + y^2}, \frac{1-x^2-y^2}{(1-x)^2 + y^2}\right). \]
Then, $\phi_1$ and $\phi_2$ are both bijective. 
Hence $\phi_2 \circ \phi_1$ is a bijection between $\mathbb H$ and $\mathbb L$. 

Hence it suffices to show that for $(x_1, y_1, z_1), \cdots, (x_n, y_n, z_n) \in \mathbb L$, 
\[ \sum_{i,j=1}^{n} c_i c_j \log\left(1+ \frac{\chi\left(\phi_2 (\phi_1(x_i,y_i,z_i)), \phi_2 (\phi_1(x_j,y_j,z_j))\right)}{2}\right) \le 0. \]

Since
\[ \chi(\phi_2 (w_1), \phi_2 (w_2)) = \frac{2|w_1 - w_2|^2}{(1-|w_1|^2)(1-|w_2|^2)}, \ \ w_1, w_2 \in \mathbb D, \]
we see that 
\[ \chi\left(\phi_2 (\phi_1(x_1,y_1,z_1)), \phi_2 (\phi_1(x_2,y_2,z_2))\right) = z_1 z_2 - x_1 x_2 - y_1 y_2 - 1 \]
\[=\cosh\left(d_{\mathbb L}\left((x_1,y_1,z_1), (x_2,y_2,z_2)\right)\right)-1, \ \ (x_1,y_1,z_1), (x_2,y_2,z_2) \in \mathbb L. \]

Hence,
it suffices to show that for $(x_1, y_1, z_1), \cdots, (x_n, y_n, z_n) \in \mathbb L$, 
\[ \sum_{i,j=1}^{n} c_i c_j \log\left(\frac{1+\cosh\left(d_{\mathbb L}\left((x_i,y_i,z_i), (x_j,y_j,z_j)\right)\right)}{2}\right) \le 0. \]
Since $2(\cosh (x/2))^2 = 1 + \cosh(x), x \in \mathbb R$, 
it suffices to show that for $(x_1, y_1, z_1), \cdots, (x_n, y_n, z_n) \in \mathbb L$, 
\[ \sum_{i,j=1}^{n} c_i c_j 2\log\left( \cosh \left(\frac{d_{\mathbb L}\left((x_i,y_i,z_i), (x_j,y_j,z_j)\right)}{2}\right) \right) \le 0. \]

Now we can apply Theorem 7.5 in Faraut-Harzallah \cite{Faraut-1974} in order to show the last inequality. 
\end{proof}

\begin{remark}\label{rem:FH}
(i)The proof of Theorem 7.5 in Faraut-Harzallah \cite{Faraut-1974} heavily depends on Takahashi's long paper \cite{Takahashi-1963} in representation theory.  
Faraut-Harzallah \cite{Faraut-1972} gave another derivation of  Theorem 7.5 in Faraut-Harzallah \cite{Faraut-1974}. 
However it heavily depends on Helgason's long paper \cite{Helgason-1970} in representation theory.  
By following the outline of \cite{Faraut-1972}, 
we give an elementary proof of Theorem \ref{thm:embeddable} without using the terminologies of representation theory. 
See Appendix \ref{sec:FH}.  \\
(ii) It is natural to consider whether the square root of the Bhattacharyya divergence $\sqrt{D_{\textup{Bhat}}}$  is isometrically embeddable into a Hilbert space. 
The squared Hellinger distance $H^2$ satisfies that  
\[ D_{\Bhat}(p_{z}:p_{w}) = -\log\left(1-H^2(p_{z} : p_{w}) \right) = -\log\left( \int_{\mathbb R} \sqrt{p_{z}(x)} \sqrt{p_{w}(x)} dx \right). \]
$\sqrt{D_{\textup{Bhat}}}$  is isometrically embeddable into a Hilbert space if and only if for every $s > 0$, 
As a function of $(z,w)$, $\left( \int_{\mathbb R} p_{z}(x) p_{w}(x) dx \right)^s$ is a positive definite kernel on $\mathbb H$. 
By the definition of the squared Hellinger distance, $\int_{\mathbb R} p_{z}(x) p_{w}(x) dx $ is positive definite. 
However, to our knowledge, it is not known whether $\left( \int_{\mathbb R} p_{z}(x) p_{w}(x) dx \right)^s$ is positive definite or not for $s \ne 1$.  
For Cauchy densities, 
we can show that 
\[ \int_{\mathbb R} p_{z}(x) p_{w}(x) dx = \frac{1}{\pi} \int_{0}^{\pi} \left(\cosh(d(z,w)) +\cos\theta \sinh (d(z,w)) \right)^{-1/2} d\theta, \ z, w \in \mathbb H, \]
where $d$ is the Poincar\'e distance. 
See Appendix \ref{sec:FH} for more details.  
\end{remark}

It is also natural to consider whether $(\mathbb H, \sqrt{D_{\textup{KL}}})$ or $(\mathbb H, \sqrt{D_{\textup{Bhat}}})$ is Gromov-hyperbolic. 

\begin{definition}
Let $(M,d)$ be a metric space. \\
(i) Let the {\it Gromov product} be 
\[ (x|y)_z := \frac{d(x,z) + d(y,z) - d(x,y)}{2}, \ \ \ x, y, z \in M.  \]
(ii) Let $\delta > 0$. 
We say that $(M,d)$ is {\it $\delta$-hyperbolic} if 
\[ (x|z)_w \ge \min\{(x|y)_w, (y|z)_w\} - \delta, \ \ x, y, z, w \in M. \]
We say that $(M,d)$ is {\it Gromov-hyperbolic} if it is $\delta$-hyperbolic for some $\delta > 0$. 
\end{definition}

It is known that $\mathbb H$ equipped with the Poincar\'e metric is Gromov-hyperbolic. (see Proposition 1.4.3 in \cite{Coornaert1990})

\begin{theorem}
Neither $(\mathbb H, \sqrt{D_{\textup{KL}}})$ or $(\mathbb H, \sqrt{D_{\textup{Bhat}}})$ is Gromov-hyperbolic. 
\end{theorem}

\begin{proof}
By Proposition 1.6 in \cite{Coornaert1990}, 
$(M,d)$ is not Gromov-hyperbolic if and only if 
\[ \sup_{x,y,z,w \in M} \left(d(x,y) + d(z,w) - \max\{d(x,z) + d(y,w), d(x,w) + d(y,z) \} \right) = +\infty. \]

We first consider $(\mathbb H, \sqrt{D_{\textup{KL}}})$. 
For $0 < a < b$, 
\[ \sqrt{D_{\textup{KL}}(p_{ai}:p_{bi})} = \sqrt{\log\left(\frac{b}{4a} + \frac{a}{4b} + \frac{1}{2} \right)}. \]
Hence, 
for $k \ge 1$, 
\[ \lim_{n \to \infty} \sup_{a > 0} \left|  \sqrt{D_{\textup{KL}}(p_{ai}:p_{an^k i})} - \sqrt{k\log n} \right| = \lim_{n \to \infty} \left|  \sqrt{D_{\textup{KL}}(p_{i}:p_{n^k i})} - \sqrt{k\log n} \right| = 0. \]
Hence, 
\begin{multline*} 
\lim_{n \to \infty} \biggl(\sqrt{D_{\textup{KL}}(p_{i}:p_{n^2 i})} + \sqrt{D_{\textup{KL}}(p_{ni}:p_{n^3 i})} \\
- \max\{\sqrt{D_{\textup{KL}}(p_{i}:p_{n i})} + \sqrt{D_{\textup{KL}}(p_{n^2 i}:p_{n^3 i})}, \sqrt{D_{\textup{KL}}(p_{i}:p_{n^3 i})} + \sqrt{D_{\textup{KL}}(p_{ni}:p_{n^2 i})} \} \biggr) 
\end{multline*}
\[ = \lim_{n \to \infty} \sqrt{D_{\textup{KL}}(p_{i}:p_{n^2 i})} + \sqrt{D_{\textup{KL}}(p_{ni}:p_{n^3 i})} - \sqrt{D_{\textup{KL}}(p_{i}:p_{n^3 i})} - \sqrt{D_{\textup{KL}}(p_{ni}:p_{n^2 i})} = +\infty. \]

We second consider $(\mathbb H, \sqrt{D_{\textup{Bhat}}})$. 
For $0 < a < b$, 
\[ \sqrt{D_{\textup{Bhat}}(p_{ai}: p_{bi})} = \sqrt{\frac{1}{2} \log \frac{b}{a} - \log\left( \frac{2}{\pi} \mathbf{K}\left(1 - \frac{a^2}{b^2}\right)\right)}. \]

By Lemma \ref{AVV2} in Appendix, 
\[ \log(4m) \le \mathbf{K}\left(1 - \frac{1}{m^2}\right) \le 2\log(4m), \ m \ge 2.  \]

Hence, 
for $k \ge 1$, 
\[ \lim_{n \to \infty} \sup_{a > 0} \left|  \sqrt{D_{\textup{Bhat}}(p_{ai}:p_{an^k i})} - \sqrt{\frac{k}{2}\log n} \right| = \lim_{n \to \infty} \left|  \sqrt{D_{\textup{Bhat}}(p_{i}:p_{n^k i})} - \sqrt{\frac{k}{2}\log n} \right| = 0. \]

Hence, 
\begin{multline*}
\lim_{n \to \infty} \biggl(\sqrt{D_{\textup{Bhat}}(p_{i}:p_{n^2 i})} + \sqrt{D_{\textup{Bhat}}(p_{ni}: p_{n^3 i})} -\\
 \max\{\sqrt{D_{\textup{Bhat}}(p_{i}: p_{n i})} + \sqrt{D_{\textup{Bhat}}(p_{n^2 i}:p_{n^3 i})}, \sqrt{D_{\textup{Bhat}}(p_{i}:p_{n^3 i})} + \sqrt{D_{\textup{Bhat}}(p_{ni}: p_{n^2 i})} \} \biggr)  = +\infty.
 \end{multline*}
\end{proof}

Now we see that both of the metrics $\sqrt{D_{\textup{KL}}}$ and $\sqrt{D_{\textup{Bhat}}}$ are locally related with the Poincar\'e metric, however, in global, they are completely different from the Poincar\'e metric. 

\bibliographystyle{plain}
\bibliography{fdivCauchy}

\appendix
\section{Information geometry of location-scale families}\label{sec:igls}\label{app:sympstatmdf}

The Fisher information matrix~\cite{locationscaleMfdMitchell-1988,CauchyVoronoi-2020} (FIM) of a location-scale family with continuously differentiable standard density $p(x)$ with full support $\bbR$ is
 $$
I(\lambda)=\frac{1}{s^2}\mattwotwo{a^2}{c}{c}{b^2},$$
 where 
\begin{eqnarray*}
a^2 &=& E_p\left[\left(\frac{p'(x)}{p(x)}\right)^2\right],\\
b^2 &=& E_p\left[\left(1+x\frac{p'(x)}{p(x)}\right)^2\right],\\
c &=& E_p\left[\frac{p'(x)}{p(x)} \left(1+ x \frac{p'(x)}{p(x)} \right) \right].
\end{eqnarray*}
When the standard density is even (i.e., $p(x)=p(-x)$), we get a diagonal Fisher matrix that
can reparameterize with 
$$
\theta(\lambda)=\left(\frac{a}{b}\lambda_1,\lambda_2\right)
$$ so that the Fisher matrix with respect to $\theta$
becomes 
$$
I_\theta(\theta)=\frac{b^2}{\theta_2^2}\, \mattwotwo{1}{0}{0}{1}.
$$ 
It follows that the Fisher-Rao geometry is hyperbolic with curvature
$\kappa=-\frac{1}{b^2}<0$, and that the Fisher-Rao distance is 
$$
\rho_{p}(\lambda_1,\lambda_2)= b\ \rho_U\left(\left(\frac{a}{b}l_1,s_1\right),\left(\frac{a}{b}l_2,s_2\right)\right)
$$
 where 
$$
\rho_U(\theta_1,\theta_2)=\arccosh\left(1+\chi(\theta_1,\theta_2)\right),$$ 
where $\arccosh(u)=\log(u+\sqrt{u^2-1})$ for $u>1$. 

For the Cauchy family, we have $a^2=b^2=\frac{1}{2}$ (curvature $\kappa=-\frac{1}{b^2}=-2$) and the Fisher-Rao distance is
$$
\rho_\FR(p_{\lambda_1}:p_{\lambda_2})=\frac{1}{\sqrt{2}} \, \arccosh(1+\chi(\lambda_1,\lambda_2)).
$$
Notice that if we let $\theta=l+is$ then the metric in the complex upper plane $\bbH$ is $\frac{|\dtheta|^2}{\Im(\theta)^2}$
where $|x+iy|=\sqrt{x^2+y^2}$ denotes the complex modulus, and $\theta\in\bbH:=\{x+iy \st x\in\bbR, y\in\bbR_{++}\}$.

It has been shown that Amari's dual $\pm\alpha$-connections~\cite{IG-2016} $\leftsup{\alpha}\Gamma$ all coincide with the Levi-Civita metric connection~\cite{locationscaleMfdMitchell-1988} $\Gamma=\leftsup{g}\Gamma$ for the Cauchy family since the Amari-Chentsov's totally symmetric cubic tensor $T$ vanishes (i.e., $T_{ijk}=0$). That is, the $\alpha$-geometry coincides with the Fisher-Rao geometry for the Cauchy family~\cite{CauchyVoronoi-2020}, for all $\alpha\in\bbR$.
The $2^3=8$ Christoffel functions defining the Levi-Civita metric connection~\cite{locationscaleMfdMitchell-1988} for the Cauchy family are:
\begin{eqnarray*}
\Gamma_{11}^1&=&\Gamma_{22}^1=\Gamma_{12}^2=\Gamma_{21}^2=0,\\
\Gamma_{12}^1&=&\Gamma_{21}^1=\Gamma_{22}^2=-\frac{1}{s},\\
\Gamma_{11}^2&=&\frac{1}{s}.
\end{eqnarray*}

Next, we recall the symplectic manifold construction of Goto and Umeno~\cite{goto2018maps} for the family of Cauchy distributions (see also~\cite{Noda-2011} for additional details):
The Fisher information metric tensor (FIm) is
$$
g_{l,s}=\frac{\dl^2+\ds^2}{2 s^2}.
$$

A vector field $K$ is a Killing vector field when the Lie derivative $\calL$ of the metric $g$ with respect to $K$ is zero: $\calL_K g=0$, i.e. the vector field $K$ preserves the metric (the flow induced by Killing vector field $K$ is a continuous isometry).
The three Killing vector fields on $TM$ are
\begin{eqnarray*}
K_1 &=& (l^2-s^2)\partial_l+2ls\partial_v,\\
K_2 &=& l\partial_l+s\partial_s,\\
K_3 &=& \partial_l.
\end{eqnarray*}

Consider the almost complex structure $J=\ds\otimes\partial_l-\dl\otimes\partial_s$ and the Levi-Civita connection $\nabla^\LC$ induced by
 the Fisher information metric.
Then $(M,g,J,\nabla^\LC)$ is a symplectic statistical manifold (Definition 4.14 of~\cite{goto2018maps}, see also~\cite{Noda-2011}) equipped with the symplectic form
$\omega=-\frac{1}{2s^2}\dl\wedge\ds$ with the set of canonical coordinates $(l,\frac{1}{2s})$.
We have $\calL_{K_1} \omega=\calL_{K_2} \omega=\calL_{K_3} \omega=0$.

The information geometry of the wrapped Cauchy family is investigated in~\cite{cao2014statistical}.
Goto and Umeno~\cite{goto2018maps} regards the Cauchy distribution as an invariant measure of the generalized Boole transforms and they
model the Cauchy manifold is modeled as a symplectic statistical manifold.
The Boole  transform $\frac{1}{2}\left(X-\frac{1}{X}\right)$ of a standard Cauchy random variable $X$  yields a standard Cauchy random variable. 
See Subsection \ref{subsec:Boole} below. 
See~\cite{CauchyLetac-1977} for a description of the functions preserving Cauchy distributions.

\section{Relationship between the parametric family}\label{sec:Knight}

We can interpret that the invariance of Cauchy $f$-divergence in Lemma \ref{lemma:finv} arises from  a relationship between the parametric family as in Assumption \ref{ass:inv} below rather than the definition of the Cauchy density itself, 
although it is shown that they are equivalent to each other by \cite{mccullagh1996mobius,goto2018maps}. 
This measure-theoretic viewpoint is clear and  useful. 
As an application, we can give a simple, alternative proof of \cite[Proposition 3.1 and Theorem 3.1]{goto2018maps}. 

\subsection{measure-theoretic framework}

Let $(X, \mu)$ be a measure space. 
Let $\varphi : \Theta \cup X \to \Theta \cup X$ be a map such that $\varphi(\Theta) \subset \Theta$ and $\varphi(X) \subset X$. 
Assume that $\varphi|_X$ is measurable. 
For $\theta \in \mathbb H$, let 
$P_{\theta}(dx) := p_{\theta} (x) \mu(dx)$, 
where $p_{\theta}$ is non-negative measurable function on $X$ and $P_{\theta}(dx)$ is a probability measure on $X$. 

\begin{assumption}\label{ass:inv}
$P_{\varphi(\theta)} = P_{\theta} \circ \varphi^{-1}$ for every $\theta$ and $\varphi$. 
\end{assumption}

We consider one-dimensional location-scale families. 
We assume that $X = \mathbb R$, $\Theta = \mathbb H$ and $\mu$ is the Lebesgue measure. 
Let $(U_i)_i$  be at most countable disjoint open sets of $\mathbb R$ such that $\mu(\mathbb R \setminus (\cup_i U_i)) = 0$ and $\varphi|_{U_i}$ is smooth and injective for each $i$. 

\begin{lemma}\label{lem:density}
\[ p_{\varphi(\theta)}(x) = \sum_i \frac{p_{\theta}(\varphi_{i}^{-1}(x))}{\left|\varphi^{\prime}(\varphi_{i}^{-1}(x))\right|} 1_{\varphi(U_i)}(x), \ \textup{ a.e. } x. \]
\end{lemma}

\begin{proof}
By Assumption \ref{ass:inv} and the change of variable formula, 
it holds that for every nonnegative measurable function  $f$, 
 \[ \int_{\mathbb R} f(x) p_{\varphi (\theta)}(x)  dx =   \int_{\mathbb R} f(\varphi(x))  p_{\theta}(x) dx = \sum_i \int_{U_i} f(\varphi(x))  p_{\theta}(x) dx\]
\begin{equation}\label{eq:inv-1} 
 = \sum_i \int_{\varphi(U_i)} f(y)  \frac{p_{\theta}(\varphi_{i}^{-1}(y))}{\left|\varphi^{\prime}(\varphi_{i}^{-1}(y))\right|} dy. 
\end{equation}
Thus we have the assertion. 
\end{proof}

\begin{proposition}\label{prop:div-inv}
Assume that $f : (0,\infty) \to \mathbb R$ is smooth and $f(1) = 0$ and convex. 
Let $D_f (p_{\theta_1} : p_{\theta_2})$ be the $f$-divergence between $p_{\theta_1}$ and $p_{\theta_2}$, that is, 
\[ D_f (p_{\theta_1} : p_{\theta_2}) := \int_{\mathbb R} f\left( \frac{p_{\theta_1}(x)}{p_{\theta_2}(x)} \right) p_{\theta_1}(x) dx. \]
Assume that $\{\varphi_i (U_i)\}_i$ are disjoint. 
Then, 
\[ D_f \left(p_{\varphi(\theta_1)} : p_{\varphi(\theta_2)} \right) = D_f (p_{\theta_1} : p_{\theta_2}). \]
\end{proposition}

The assumption that $\{\varphi_i (U_i)\}_i$ are disjoint is crucial. 
See Remark \ref{rem:inv-fail}. 

\begin{proof}
By using \eqref{eq:inv-1} and the fact that $\varphi$ is bijective except a measure zero set, 
we see that 
\begin{equation*}
p_{\theta}(x) = p_{\varphi (\theta)}(\varphi(x)) |\varphi^{\prime}(x)| , \ \ \textup{ a.e. } x. 
\end{equation*}
Hence, 
\[ D_f \left(p_{\varphi(\theta_1)} : p_{\varphi(\theta_2)} \right) = \int_{\mathbb R}  f\left( \frac{p_{\varphi(\theta_1)}(\varphi(x))}{p_{\varphi(\theta_2)}(\varphi(x))} \right) p_{\theta_1}(x) dx = \int_{\mathbb R}  f\left( \frac{p_{\theta_1}(x)}{p_{\theta_2}(x)} \right) p_{\theta_1}(x) dx.  \]
\end{proof}

\subsection{M\"obius transformations}

For $A = \begin{pmatrix} a & b \\ c & d \end{pmatrix}  \in SL(2, \mathbb{R})$, 
let 
$\varphi_A (z) = A \cdot z := \dfrac{az+b}{cz+d}$.   
This is well-defined on $\overline{\mathbb H}$ if $c = 0$, and on $\overline{\mathbb H} \setminus \{-d/c\}$ if $c \ne 0$. 
If $c \ne 0$, then we let 
$\varphi_A (-d/c) := a/c$.
Then, $\varphi_A$ is a bijection on $\mathbb R$. 
This also holds if $c=0$. 

For $\theta \in \mathbb H$, let 
$P_{\theta}(dx) := p_{\theta} (x) dx$. 
Assume that $p_{\theta} (x) > 0$ for every $x \in \mathbb R$. 
This is a probability measure on $\mathbb R$. 

\begin{lemma}
$P_{\varphi_A (\theta)} = P_{\theta} \circ \varphi_{A}^{-1}$ for every $A \in SL(2, \mathbb R)$ and $\theta \in \mathbb H$. 
\end{lemma}

By \cite{KnightCauchy-1976}, such parametric location-scale family is restricted to the univariate Cauchy distribution. 
Let $\chi$ be the maximal invariant. 
By Proposition \ref{prop:div-inv}, 
we have Theorem \ref{thm:fdivsymmetric}.

\begin{proposition}
Let 
\[ C(x; \ell, s) := \frac{s}{\pi} \frac{1}{(x-\ell)^2 + s^2}, \ \ x, \ell \in \mathbb{R}, s > 0. \]
Assume that $x \ne -d/c$ if $c \ne 0$. 
Let $x^{\prime} := \varphi_A  (x)$, 
$$\ell^{\prime} := \textup{Re}(\varphi_A  (\ell + is)) = \frac{(a\ell+b)(c\ell + d) + acs^2}{(c\ell + d)^2 + c^2 s^2}$$ 
and 
$$s^{\prime} := \textup{Im}(\varphi_A  (\ell + is)) = \frac{s}{(c\ell + d)^2 + c^2 s^2}.$$
Then, 
$\varphi_A^{-1}(x^{\prime}) = \{x\}$, and 
\[ C(x^{\prime}; \ell^{\prime}, s^{\prime}) = \frac{C(x; \ell, s)}{\left|\varphi_A^{\prime} (x) \right|}.  \]
\end{proposition}

This assertion essentially corresponds to \cite[Proposition 3.1 and Theorem 3.1]{goto2018maps}. 

\begin{proof}
We remark that $\varphi_A$ is bijective. 
Hence $\varphi_A^{-1}(x^{\prime}) = \{x\}$. 
By Lemma \ref{lem:density}, 
\[ C(x^{\prime}; \ell^{\prime}, s^{\prime}) = \frac{C(x; \ell, s)}{\left|\varphi_A^{\prime} (x) \right|}, \ \textup{ a.e. } x.  \]
Since the functions in the left and right hand sides in the above display are both continuous on $\mathbb R \setminus \{-d/c\}$, 
we have the assertion. 
\end{proof}

\begin{remark}
(i) Since $\varphi_A  (\mathbb H) \subset \mathbb H$, $\varphi_A $ defines a flow on $\mathbb H$. \\
(ii) If $c \ne 0$, then, 
$\left\{-d/c + yi : y \in \mathbb R \right\}$ and $\mathbb R$ are invariant manifolds. 
The map restricted on $\left\{-d/c + yi : y \in \mathbb R \right\}$ is $s \mapsto 1/(c^2 s)$, and the map restricted on $ \mathbb R $ is $\ell \mapsto \varphi_A  (\ell)$. 
\end{remark}

\subsection{Boole transformations}\label{subsec:Boole}

We give an alternative simultaneous proof of \cite[Proposition 3.1 and Theorem 3.1]{goto2018maps} themselves. 

For $a > 0$, let 
$$\varphi_{a}(z) := \begin{cases} a(z - z^{-1}) \ z \in \overline{\mathbb{H}}  \setminus\{0\} \\ 0 \ \ \ \ \  \  \ \ \ \ \ z = 0\end{cases}.$$

\begin{proposition}
Assume that $x \ne 0$. 
Let $x^{\prime} := \varphi_a  (x)$, 
$$\ell^{\prime} := \textup{Re}(\varphi_a  (\ell + is)) = a \ell \frac{\ell^2 + s^2 - 1}{\ell^2 + s^2}$$ 
and 
$$s^{\prime} := \textup{Im}(\varphi_a  (\ell + is)) = a s \frac{\ell^2 + s^2 + 1}{\ell^2 + s^2}. $$
Then, 
$\varphi_{a}^{-1}(x^{\prime}) = \{x, -1/x\}$, and 
\[ C(x^{\prime}; \ell^{\prime}, s^{\prime}) = \frac{C(x; \ell, s)}{\left|\varphi_a^{\prime} (x) \right|} + \frac{C(-1/x; \ell, s)}{\left|\varphi_a^{\prime} (-1/x) \right|}.  \]
\end{proposition}

\begin{proof}
For ease of notation we let $\theta := \ell + si$. 
We remark that 
\begin{equation}\label{eq:sym-GU} 
\varphi_a (y) = \varphi_a (-1/y), \ \ y \ne  0, 
\end{equation}
and 
\begin{equation}\label{eq:difference}
|\varphi_a (x) - \varphi_a (\theta)| = a |x - \theta| \frac{|x \theta + 1|}{|x \theta|}, \ \ x \ne 0, \theta \in \mathbb{H}. 
\end{equation}
Let $F$ be a non-negative Borel measurable function. 
By \eqref{eq:difference}  and the change of variable formula with $y = \varphi_a (x)$, 
\[ \int_{\mathbb R} F(y) C(y; \varphi_a (\theta)) dy = \int_0^{\infty} F(\varphi_a (x)) C(x;\theta) \frac{(x^2 + 1)(|\theta|^2 + 1)}{|x \theta + 1|^2}  dx \]
\[= \int_{-\infty}^0 F(\varphi_a (x)) C(x;\theta) \frac{(x^2 + 1)(|\theta|^2 + 1)}{|x \theta + 1|^2} dx. \]
Hence, 
\[ \int_{\mathbb R} F(y) C(y; \varphi_a (\theta)) dy = \int_{\mathbb R} F(\varphi_a (x)) C(x;\theta) \frac{(x^2 + 1)(|\theta|^2 + 1)}{2 |x \theta + 1|^2}  dx \]
\[ =  \int_{\mathbb R} F(\varphi_a (x)) C(x;\theta) dx  +  \int_{\mathbb R} F(\varphi_a (x)) C(x;\theta) \left( \frac{(x^2 + 1)(|\theta|^2 + 1)}{2 |x \theta + 1|^2}  -1\right) dx. \]
Since 
$$\frac{(x^2 + 1)(|\theta|^2+1)}{|x\theta + 1|^2} = 1 + \frac{|x - \theta|^2}{|x\theta + 1|^2},$$
it holds that 
\[ \int_{\mathbb R} F(\varphi_a (x)) C(x;\theta) \left( \frac{(x^2 + 1)(|\theta|^2 + 1)}{2 |x \theta + 1|^2}  -1\right) dx\]
\[  = \frac{1}{2}\int_{\mathbb R} F(\varphi_a (x)) C(x;\theta) \left( \frac{|x - \theta|^2}{|x\theta + 1|^2}  -1\right) dx\]
\[ = \frac{s}{2\pi}\int_{\mathbb R} \frac{F(\varphi_a (x))}{|x\theta + 1|^2}  dx - \frac{1}{2}\int_{\mathbb R} F(\varphi_a (x)) C(x;\theta) dx. \]
By the change of variable formula and \eqref{eq:sym-GU}, 
\[ \frac{s}{\pi}\int_{\mathbb R} \frac{F(\varphi_a (x))}{|x\theta + 1|^2}  dx = \int_{\mathbb R} F(\varphi_a (x)) C(x;\theta) dx, \]
and hence, 
\[ \int_{\mathbb R} F(\varphi_a (x)) C(x;\theta) \left( \frac{(x^2 + 1)(|\theta|^2 + 1)}{2 |x \theta + 1|^2}  -1\right) dx = 0. \]
Thus we obtain that 
\[ \int_{\mathbb R} F(y) C(y; \varphi_a (\theta)) dy  =  \int_{\mathbb R} F(\varphi_a (x)) C(x;\theta) dx. \]

Let two functions $\varphi_{a, \pm}$ be the restrictions of $\varphi_a$ to $(0, \infty)$ and $(-\infty, 0)$ respectively. 
Then, by Lemma \ref{lem:density}, 
\[ C(y; \varphi_a (\theta))  = \frac{C(\varphi_{a,+}^{-1}(y);\theta)}{\varphi_a^{\prime}(\varphi_{a,+}^{-1}(y))} + \frac{C(\varphi_{a,-}^{-1}(y);\theta)}{\varphi_a^{\prime}(\varphi_{a,-}^{-1}(y))}, \ \textup{ a.e. } y.  \]
Since the functions in the left and right hand sides in the above display are both continuous on $\mathbb R \setminus \{0\}$, 
\[ C(y; \varphi_a (\theta))  = \frac{C(\varphi_{a,+}^{-1}(y);\theta)}{\varphi_a^{\prime}(\varphi_{a,+}^{-1}(y))} + \frac{C(\varphi_{a,-}^{-1}(y);\theta)}{\varphi_a^{\prime}(\varphi_{a,-}^{-1}(y))}, \ \textup{ for every } y \ne 0.  \] 
\end{proof}

\begin{remark}\label{rem:inv-fail}
\[ D_f \left(p_{\varphi_a (\theta_1)} : p_{\varphi_a (\theta_2)} \right) \ne D_f (p_{\theta_1} : p_{\theta_2})\] 
for $a = 2, \theta_1 = i$ and $\theta_2 = 2i$. 
\end{remark}

\section{Revisiting the KLD between Cauchy densities}\label{sec:KLDCauchy}

We shall prove the following result~\cite{KLCauchy-2019} using complex analysis:
$$
D_\KL(p_{l_1,s_1}:p_{l_2,s_2}) =
\log\left( \frac{\left(s_{1}+s_{2}\right)^{2}+\left(l_{1}-l_{2}\right)^{2}}{4 s_{1} s_{2}}\right).
$$
 
\begin{proof}
$$
D_\KL(p_{l_1, s_1} : p_{l_2, s_2} ) 
=  \frac{s_1}{\pi} \int_{\bbR} \frac{\log ((z-l_2)^2  + s_2^2) }{(z - l_1)^2 + s_1^2} \dz 
$$

\begin{equation}\label{eq:expand} 
- \frac{s_1}{\pi} \int_{\bbR} \frac{\log ((z-l_1)^2  +s_1^2)}{(z - l_1)^2 + s_1^2} \dz + 
\log \frac{s_1}{s_2}. 
\end{equation}

As a function of $z$, 
$$
 \frac{\log (z-l_2 + i s_2)}{z - l_1 + i s_1} 
 $$
 is holomorphic on the upper-half plane $\{x+yi : y > 0\}$. 
By the Cauchy integral formula~\cite{needham1998visual}, we have that for sufficiently large $R$, 
$$
\frac{1}{2\pi i } \int_{C_R^{+}} \frac{\log (z-l_2 + i s_2)}{(z - l_1)^2 + s_1^2} \dz = 
\frac{\log(l_1 - l_2 + i(s_2 + s_1))}{2 s_1 i},
$$
where 
$$
C_R^{+} := \{z : |z| = R, \Im(z) > 0\} \cup \{z : \Im(z) = 0, |\Re(z)| \leq R \}.
$$

Hence, {by $R \to +\infty$,}  we get 
\begin{equation}\label{eq:upper}
\frac{s_1}{\pi} \int_{\bbR} \frac{\log (z-l_2 + i s_2)}{(z - l_1)^2 + s_1^2} \dz = 
\log(l_1 - l_2 + i(s_2 + s_1)).
\end{equation}

As a function of $z$, 
$$
 \frac{\log (z-l_2 - i s_2)}{z - l_1 - i s_1} 
 $$
 is holomorphic on the lower-half plane $\{x+yi : y < 0\}$. 
By the Cauchy integral formula again, we have that for sufficiently large $R$, 
$$
\frac{1}{2\pi i } \int_{C_R^{-}} \frac{\log (z-l_2 - i s_2)}{(z - l_1)^2 + s_1^2} \dz = 
\frac{\log(l_1 - l_2-  i(s_2 + s_1))}{-2 s_1 i},  
$$
where 
$$
C_R^{-} := \{z : |z| = R, \Im(z) < 0\} \cup \{z : \Im(z) = 0, |\Re(z)| \le R \}.
$$

Hence,  {by $R \to +\infty$,}  we get
\begin{equation}\label{eq:lower}
\frac{s_1}{\pi} \int_{\bbR} \frac{\log (z-l_2 - i s_2)}{(z - l_1)^2 + s_1^2} \dz = 
\log(l_1 - l_2 - i(s_2 + s_1)).
\end{equation}

{By Eq.~\ref{eq:upper} and Eq.~\ref{eq:lower},} we have that 
\begin{equation}\label{eq:different}  
\frac{s_1}{\pi} \int_{\bbR} \frac{\log ((z-l_2)^2  + s_2^2) }{(z - l_1)^2 + s_1^2} \dz =
 \log\left( (l_1 - l_2)^2 + (s_1 + s_2)^2 \right).
\end{equation}

In the same manner, we have that 
\begin{equation}\label{eq:same}  
\frac{s_1}{\pi} \int_{\bbR} \frac{\log ((z-l_{{1}})^2  + s_{{1}}^2) }{(z - l_1)^2 + s_1^2} \dz =
 \log(4 s_1^2). 
\end{equation}

By substituting Eq.~\ref{eq:different} and Eq.~\ref{eq:same} into Eq.~\ref{eq:expand}, we obtain the formula Eq.~\ref{eq:kldcauchy}.
\end{proof}

\begin{remark}
Thomas Simon \cite{Simon-2020} also obtained an alternative proof of \cite{KLCauchy-2019}, 
which uses the L\'evy-Khintchine formula and the potential formula for the infinitely divisible distributions, and the Frullani integral. 
\end{remark}

\section{Revisiting the chi-squared divergence between Cauchy densities}\label{sec:chisquared}

\begin{proposition}
\begin{equation}\label{eq:chidcauchy}
D_\chi^{N} (p_{l_1, s_1} : p_{l_2, s_2} ) = \frac{(l_1 - l_2)^2 + (s_1-s_2)^2}{2s_1s_2}. 
\end{equation}
\end{proposition}

\begin{proof}
We first remark that 
$$
D_\chi^{N} (p_{l_1, s_1} : p_{l_2, s_2} ) 
=  \int_{\bbR} \frac{p_{l_2, s_2}^2(x)}{p_{l_1, s_1}(x)} \dx -1.  
$$

Let 
$
F(z) := \frac{(z-l_1)^2 + s_1^2}{(z - l_2 + i s_2)^2}
$. 
Then, this is holomorphic on the upper-half plane $\bbH$, and, 
$$ 
\frac{p_{l_2, s_2}(x)^2}{p_{l_1, s_1}(x)}  = \frac{s_2^2}{\pi s_1} \frac{F(x)}{(x - l_2 - i s_2)^2}. 
$$

By the Cauchy integral formula~\cite{needham1998visual}, we have that for sufficiently large $R$, 
$$
\frac{1}{2\pi i } \int_{C_R^{+}} \frac{F(z)}{(z - l_2 - i s_2)^2} \dz = F^{\prime}(l_2 + is_2),
$$
where 
$
C_R^{+} := \{z : |z| = R, \Im(z) > 0\} \cup \{z : \Im(z) = 0, |\Re(z)| \leq R \}
$.

Since 
$$
F^{\prime}(z) = 2\frac{(z-s_1)(z-l_2 + is_2) - (z-l_1)^2 -s_1^2}{(z - l_2 + i s_2)^3}, 
$$
 we have that 
$$
\int_{C_R^{+}} \frac{F(z)}{(z - l_2 - i s_2)^2} \dz = \frac{\pi}{2} \frac{(l_1 - l_2)^2 + s_1^2 + s_2^2}{s_2^3}. 
$$

Now, by $R \to \infty$, we obtain the formula Eq.~\ref{eq:chidcauchy}.
\end{proof}

\section{Total variation between densities of a location family}\label{sec:tvlf}
Consider a location family with {\em even} standard density $p(-x)=p(x)$.
Then $p(x-l_1)=p(x-l_2)=p(l_2-x)$ when $x=\frac{l_1+l_2}{2}$.
Let $\Phi(a)=\int_{-\infty}^a p(x)\dx$ denote the standard cumulative density function, 
$\Phi_{l,s}(a)=\int_{-\infty}^a p(\frac{x-l}{s})\dx=\Phi(\frac{a-l}{s})$ 
with  $\Phi_{l,s}(-\infty)=0$ and $\Phi_{l,s}(+\infty)=1$.
We have $\int_{a}^b p(x)\dx=\Phi(b)-\Phi(a)$ and $\int_a^{+\infty} p_{l,s}(x)\dx=1-\Phi(\frac{a-l}{s})$.

Then the total variation distance between $p_{l_1}$ and $p_{l_2}$ is
\begin{eqnarray*}
D_\TV(p_{l_1}:p_{l_2}) &=& \frac{1}{2}\left( \int_{-\infty}^{\frac{l_1+l_2}{2}} |p_{l_1}(x)-p_{l_2}(x)|\dx
+\int_{\frac{l_1+l_2}{2}}^{+\infty} |p_{l_2}(x)-p_{l_1}(x)|\dx\right)\\
&=& 2\Phi\left(\frac{|l_1-l_2|}{2s}\right) -1 \leq 1
\end{eqnarray*}

\begin{proposition}
The total variation between two densities $p_{l_1}$ and $p_{l_2}$ of a location family with even standard density is 
$2\Phi\left(\frac{|l_1-l_2|}{2s}\right) -1$.
\end{proposition}

For the Cauchy distribution, since we have
$$
\Phi_{l,s}(x)=\frac{1}{\pi}\arctan\left(\frac{x-l}{s}\right)+\frac{1}{2},
$$
we recover $D_\TV(p_{l_1}:p_{l_2}) =\frac{2}{\pi} \arctan\left(\frac{|l_2-l_1|}{2s}\right)$.

The total variation formula extends to any fixed scale location families.

\section{Complete elliptic integrals}\label{sec:cei}

This section is devoted to the details of the proof of \eqref{eq:sqrt-Bhat} in the proof of Theorem \ref{thm:sqrtBhat}. 

\begin{proof}
Let
\[ F_4(u) := \frac{-\log\left(2 e^{-u/4} \mathbf{K}(1-e^{-u}) /\pi \right)}{u^2}. \]

We consider the derivative. 
\[ F_4^{\prime}(u) = \frac{-1}{u^2} \left( \frac{1}{4} + e^{-u} \frac{\mathbf{K}^{\prime}(1-e^{-u})}{\mathbf{K}(1-e^{-u})} - \frac{2}{u} \log\left(2\mathbf{K}(1-e^{-u}) /\pi \right) \right). \]

Now it suffices to show that for every $u > 0$, 
\[ \frac{1}{4} + e^{-u} \frac{\mathbf{K}^{\prime}(1-e^{-u})}{\mathbf{K}(1-e^{-u})} - \frac{2}{u} \log\left(2 \mathbf{K}(1-e^{-u}) /\pi \right) > 0. \] 

Let $x := 1 - e^{-u}$. 
Then, it suffices to show that for every $x \in (0,1)$, 
\[  \frac{1}{4} + (1-x) \frac{\mathbf{K}^{\prime}(x)}{\mathbf{K}(x)} + \frac{2}{\log(1-x)} \log\left(2 \mathbf{K}(x) /\pi \right) > 0. \]
Let 
\[ G_4(x) :=  \log\left(2 \mathbf{K}(x) /\pi \right) + (\log(1-x))\left( \frac{1}{8} + \frac{1-x}{2} \frac{\mathbf{K}^{\prime}(x)}{\mathbf{K}(x)}  \right). \]
It suffices to show that $G_4(x) < 0$ for every $x \in (0,1)$. 

We see that $G_4(0) = 0$. 
Hence it suffices to show that $G_4^{\prime}(x) < 0$ for every $x \in (0,1)$. 
By Lemma \ref{K-deri} below, 
\[ G_4(x) =   \log\left(2 \mathbf{K}(x) /\pi \right) + (\log(1-x))\left( \frac{3}{8} + \frac{1}{4x} \left( .\frac{\mathbf{E}(x)}{\mathbf{K}(x)} - 1\right) \right). \]

By Lemmas \ref{K-deri} and \ref{EK-deri} below, 
\[ G_4^{\prime}(x) = -\frac{H_4(x)}{8x^2 (1-x)}, \]
where we let 
\[ H_4(x) := (x(2-x) + (x-1)\log(1-x)) \mathbf{K}(x)^2 -2x \mathbf{K}(x)\mathbf{E}(x) + \log(1-x) \mathbf{E}(x)^2. \]
Then it suffices to show that $H_4(x) > 0$ for every $x \in (0,1)$. 
Since $-2x < 0$ and $ \log(1-x) < 0$, 
by noting Lemma \ref{Gauss-AGM} below, 
it holds that 
\[ \frac{H_4(x)}{\mathbf{K}(x)^2} \ge  (x(2-x) + (x-1)\log(1-x)) -2x I_4(x) + \log(1-x) I_4(x)^2, \]
where we let 
\[ I_4(x) := \frac{1}{2} - \frac{x}{4} + \frac{\sqrt{1-x}}{2}. \]

Our main idea is to use different estimates for $H(x)/\mathbf{K}(x)^2$ on a neighborhood of $1$ and on the compliment of it. 

\begin{lemma}
For $x \le 0.998$, 
\[ (x(2-x) + (x-1)\log(1-x)) -2x I(x) + \log(1-x) I(x)^2 > 0.\]
\end{lemma}

\begin{proof}
Let $y := \sqrt{1-x}$. 
Then, 
\[ (x(2-x) + (x-1)\log(1-x)) -2x I(x) + \log(1-x) I(x)^2 > 0 \]
is equivalent with 
\[ \log y > 4\frac{y^2 -1}{y^2 + 6y + 1}. \]
Let $$P_4(y) := \log y - 4\frac{y^2 -1}{y^2 + 6y + 1}.$$
Then, $P_4(1) = 0$. 
By considering the derivative of $P$, 
it is increasing $y < 5 - 2\sqrt{6}$ and decreasing $y > 5 - 2\sqrt{6}$. 

We see that 
$
P_4 (y) > 0, \ \ y > 0.041.
$
Now the assertion follows from the fact that 
\[ 0.998 < 1 - (0.041)^{2}. \]
\end{proof}

Now it suffices to show that $H(x) > 0$ for $x > 0.998$.

\begin{lemma}\label{lem:1st}
\[ x(2-x) + (x-1)\log(1-x) \ge 1, \ \ x \in (0.998,1). \]
\end{lemma}

\begin{proof}
Let $g_4(x) := x(2-x) + (x-1)\log(1-x)$. 
Then, $g_4(1) = 1$ and 
\[ g_4^{\prime}(x) = 3-2x + \log(1-x).\]
This is negative if $x > 0.9$. 
\end{proof}

\begin{lemma}\label{lem:2nd}
\[ 2x .\frac{\mathbf{E}(x)}{\mathbf{K}(x)} < \frac{1}{2}, \ \ x \in (0.998,1). \]
\end{lemma}

\begin{proof}
We see that 
\[ \frac{d}{dx} \left( x .\frac{\mathbf{E}(x)}{\mathbf{K}(x)}\right) \le  2.\frac{\mathbf{E}(x)}{\mathbf{K}(x)}  -\frac{1}{2}. \]

By Lemma \ref{EK-deri} below and the fact that 
\[ \frac{\mathbf{E}(0.995)}{\mathbf{K}(0.995)} < \frac{1}{4},\] 
we see that 
\[  2.\frac{\mathbf{E}(x)}{\mathbf{K}(x)}  \le \frac{1}{2}, \ \ x > 0.995. \]

Hence, 
\[ 2x .\frac{\mathbf{E}(x)}{\mathbf{K}(x)} < 2 \frac{\mathbf{E}(0.995)}{\mathbf{K}(0.995)} < \frac{1}{2}. \]
\end{proof}

\begin{lemma}\label{lem:3rd}
\[ -\log(1-x) \left( .\frac{\mathbf{E}(x)}{\mathbf{K}(x)}  \right)^2 < \frac{1}{2}, \ \ x \in (0.998,1). \]
\end{lemma}

\begin{proof}
We use Lemma \ref{AVV} below. 
It suffices to show that 
\[  \frac{2 x^{1/2}}{\log(1+x^{1/2}) - \log(1-x^{1/2})} \le \sqrt{\frac{1}{-2\log(1-x)}}, \ \ x \in (0.998,1). \]
This is equivalent with 
\[ h_4(x) := \left( \log(1+x^{1/2}) - \log(1-x^{1/2}) \right)^2 + 8x\log(1-x) \ge 0, \ \ x \in (0.998,1).  \]
We see that 
\[ h_4^{\prime}(x) = -2\frac{\log(1-\sqrt{x}) - \log(1+\sqrt{x}) + 2\sqrt{x}(x + (x-1)\log(1-x))}{(1-x)\sqrt{x}}. \]

It is easy to see that 
\[ \log(1-\sqrt{x}) - \log(1+\sqrt{x}) + 2\sqrt{x}(x + (x-1)\log(1-x)) < 0,  \ \ x \in (0.998,1). \]
Hence $h_4$ is increasing at least on $(0.998,1)$.  
Now use the fact that $h_4(0.998) > 0$. 
\end{proof}

By Lemmas \ref{lem:1st}, \ref{lem:2nd} and \ref{lem:3rd}, 
we see that $H_4(x) > 0$ for $x > 0.998$. 
The proof of Eq. \ref{eq:sqrt-Bhat} is completed. 
\end{proof}

\subsection{Some Lemmas concerning the complete elliptic integrals}

In this subsection, we collect standard results about the complete elliptic integrals. 

\begin{lemma}\label{K-deri}
\[ \mathbf{K}^{\prime}(x) = -\frac{\mathbf{K}(x)}{2x} + \frac{\mathbf{E}(x)}{2x(1-x)}. \]
\end{lemma}

\begin{lemma}\label{EK-deri}
\[ \frac{d}{dx}\left(.\frac{\mathbf{E}(x)}{\mathbf{K}(x)} \right) = -\frac{1}{2x} + \frac{1}{x} .\frac{\mathbf{E}(x)}{\mathbf{K}(x)} - \frac{1}{2x(1-x)} \left( .\frac{\mathbf{E}(x)}{\mathbf{K}(x)}\right)^2 \le 0. \]
In particular, 
$\mathbf{E}/ \mathbf{K}$ is strictly decreasing. 
\end{lemma}

\begin{lemma}\label{Gauss-AGM}
\[ .\frac{\mathbf{E}(x)}{\mathbf{K}(x)} \le \frac{1}{2} - \frac{x}{4} + \frac{\sqrt{1-x}}{2}, \ x \in [0,1).\]
\end{lemma}

The following is due to Anderson,  Vamanamurthy, and Vuorinen \cite{AVV}. 

\begin{lemma}[{\cite[Theorem 3.6]{AVV}}]\label{AVV}
\[  .\frac{\mathbf{E}(x)}{\mathbf{K}(x)} \le \frac{2 x^{1/2}}{\log(1+x^{1/2}) - \log(1-x^{1/2})}, \ x \in [0,1). \]
\end{lemma}

The following is due to Eq. (1.1) in \cite{Carlson1985}. See also Eq. (6.2) in \cite{AVV}. 

\begin{lemma}\label{AVV2}
\[ \log \left(\frac{4}{\sqrt{1-x}}\right) \le \mathbf{K}(x) \le \frac{4}{3+x} \log \left(\frac{4}{\sqrt{1-x}}\right), \ \ x \in [0,1). \]
\end{lemma}

\section{Negative definiteness of the KLD between Cauchy densities}\label{sec:FH}

In this section, we give an elementary proof of Theorem \ref{thm:embeddable}. 
We first give an outline of the proof. 
Our proof follows the strategy of \cite{Faraut-1972} and consists of three steps. 
We do not need to introduce the hyperboloid space $\mathbb L$. 

Step 1. \ Let $d$ be the Poincar\'e distance on $\mathbb H$. 
We remark that $d = \sqrt{2}\rho_{\FR}$. 
Then, $\cosh(d(z,w)) = 1 + \chi(z,w)$ and  
\[ 2 \log \cosh \left(\frac{d(z,w)}{2} \right) = \log \left( 1 + \frac{\chi(z,w)}{2} \right). \]

We see that for every $r \ge 0$, 
\[ 2 \log \cosh \left(\frac{r}{2} \right) = \lim_{s \to +0} \frac{1}{s} \left( 1 - \frac{1}{2\pi} \int_{-\pi}^{\pi}  \left( \cosh(r) + \cos \theta \sinh (r) \right)^{-s} d\theta \right).\]

Hence it suffices to show that 
\[ H_s (z,w) := \frac{1}{2\pi} \int_{-\pi}^{\pi}  \left( \cosh(d(z,w)) + \cos \theta \sinh (d(z,w)) \right)^{-s} d\theta, \ z,w \in \mathbb H, \]
is positive definite for every $s \in (0,1)$. 

Step 2. \ Let $$P(z,x) := \frac{\textup{Im}(z)}{|x-z|^2} (x^2 + 1), \ z \in \mathbb H, x \in \mathbb R,$$
and $\mu(dx) := \dfrac{dx}{\pi (x^2 + 1)}$. 

Then we see that 
\[ H_s (z,w) = \int_{\mathbb R} P(z,x)^s P(w,x)^{1-s} \mu(dx) \]
\[ = C(s) \iint_{\mathbb R^2} P(w,x)^{1-s} P(z,y)^{1-s} \left( \frac{(x-y)^2}{(x^2 + 1) (y^2 + 1)} \right)^{-s} \mu(dx)\mu(dy), \]
where $C(s)$ is a positive constant depending only on $s$. 

Step 3. \ Let $z_1, \cdots, z_n \in \mathbb H$ and $c_1, \cdots, c_n \in \mathbb R$ with $\sum_{i=1}^{n} c_i = 0$. 
Let 
$$\varphi_s (x) := \sum_{i=1}^{n} c_i P(z_i, x)^{1-s}, \ x \in \mathbb R,$$
which is continuous on $\mathbb R$. 

Let $k_s (x,y) := \left( \frac{(x-y)^2}{(x^2 + 1) (y^2 + 1)} \right)^{-s}$, 
which is a positive definite kernel on $\mathbb R$. 

Thus we see that 
\[ \sum_{i.j=1}^{n} c_i c_j H_s (z_i, z_j) = \frac{C(s)}{\pi^2} \iint_{\mathbb R^2} \frac{\varphi_s (x) \varphi_s (y) k_s (x,y)}{(x^2 + 1)(y^2 + 1)} dxdy \ge 0. \]
 
Now we proceed to the full proof. 

Step 1. \  
It is known that (see formula no.4.224.9 in \cite{gradshteyn2014table})
\[ 2 \log \cosh \left(\frac{r}{2}\right) = \frac{1}{2\pi} \int_{-\pi}^{\pi} \log\left(\cosh(r) + \cos\theta \sinh (r) \right) d\theta, \ \ r \ge 0. \]

We see that for $r \ge 0$, 
\[ \left|\log (\cosh(r) + \cos\theta \sinh (r)) \right| \le r.  \]

Since for $t > 0$, $\lim_{s \to +0} \frac{1 - t^{-s}}{s} = \log t$ and $|\frac{1 - t^{-s}}{s}| \le |\log t|$, 
\[ \int_{-\pi}^{\pi} \log\left(\cosh(r) +\cos\theta \sinh(r) \right) d\theta = \lim_{s \to +0} \int_{-\pi}^{\pi} \frac{1 - \left(\cosh(r) +\cos\theta \sinh (r) \right)^{-s}}{s}d\theta, \ r > 0, \]
by the Lebesgue convergence theorem. 
This convergence also holds for $r=0$. 

Step 2. \ 
\begin{lemma}\label{lem:transfer}
\[ \frac{1}{2\pi} \int_{-\pi}^{\pi} \left(\cosh(r) +\cos\theta \sinh (r) \right)^{-s} d\theta = \int_{\mathbb R} P(e^r i, x)^s \mu(dx).  \]
\end{lemma}

\begin{proof}
Let $x = \tan \frac{\theta}{2}$. 
Then, 
$d\theta = \dfrac{2}{1+x^2} dx$ and 
\[ \cosh(r) +\cos\theta \sinh (r) = \frac{e^{2r} + x^2}{e^r (1+x^2)} = \frac{1}{P(e^r i, x)}. \]
\end{proof}

\begin{lemma}\label{lem:rot-inv}
For $A \in SO(2)$ and $z \in \mathbb H$, 
\[ \int_{\mathbb R} P(A.z,x)^s \mu(dx) =  \int_{\mathbb R} P(z,x)^s \mu(dx). \]
\end{lemma}

\begin{proof}
Let 
$A = \begin{pmatrix} \cos\theta & -\sin\theta \\ \sin\theta & \cos\theta \end{pmatrix}$. 
Let $y \in \mathbb R$ such that $x = A.y$. 
Then, 
\[ P(A.z, A.y) = P(z,y) \]
and 
\[ \mu(dx) = \frac{1}{\pi} \frac{1}{(A.y)^2 + 1} \frac{dx}{dy} dy = \frac{1}{\pi} \frac{1}{y^2 + 1} dy = \mu(dy). \]
\end{proof}

Now we introduce a group structure on $\mathbb H$. 
For $z = z_1+iz_2$ and $w = w_1+iw_2$, let 
\[ zw := (z_1 + z_2 w_1) + i z_2 w_2. \]
This gives a group structure on $\mathbb H$. 
It holds that 
$$ w^{-1} = \frac{-w_1 + i}{w_2}, \ w = w_1+iw_2$$
and the unit element is the imaginary unit $i$. 

We see that 
$\chi(w^{-1}z,i) = \chi(z,w), z,w \in \mathbb H$ and hence 
\begin{equation}\label{eq:gi-dist}
d(w^{-1}z,i) = d(z,w), z,w \in \mathbb H.
\end{equation}

\begin{lemma}\label{lem:inverse}
For $z,w \in \mathbb H$, 
\[ H_s (z,w) = \int_{\mathbb R} P(w^{-1} z, x)^s \mu(dx).  \]
\end{lemma}

\begin{proof}
By \eqref{eq:gi-dist}, we can assume that $w=i$. 
Then there exists $A \in SO(2)$ such that $e^{d(z,i)} i = A.z$. 
Now the assertion follows from Lemmas \ref{lem:transfer} and \ref{lem:rot-inv}. 
\end{proof}

For $w = w_1 + iw_2 \in \mathbb H$ and $x \in \mathbb R$, 
we let $wx := w_2 x + w_1$.

\begin{lemma}\label{lem:ratio}
\[ P(w^{-1}z, x) P(w, wx) = P(z, wx), \ \ z, w \in \mathbb H, x \in \mathbb R. \]
\end{lemma}

\begin{proof}
Since 
\[ w^{-1}z = \frac{z_1 - w_1 + i z_2}{w_2}, \ \ z = z_1 + iz_2, w = w_1 + iw_2, \]
we see that 
\[ P(w^{-1}z, x) = \frac{z_2 w_2}{(z_1 - w_1 - w_2 x)^2 + z_2^2} (x^2 + 1). \]
We also see that 
\[ P(z,wx) = \frac{z_2 ((w_2 x + w_1)^2+1)}{(z_1 - w_1 - w_2 x)^2 + z_2^2} \]
and \[ P(w,wx) = \frac{(w_2 x + w_1)^2+1}{w_2 (x^2 + 1)}. \]
The assertion follows from these identities. 
\end{proof}

\begin{proposition}\label{prop:prop1}
\[ H_s (z,w) = \int_{\mathbb R} P(z,x)^s P(w,x)^{1-s} \mu(dx), \ \ z, w \in \mathbb H.\] 
\end{proposition}

\begin{proof}
By Lemmas \ref{lem:inverse} and \ref{lem:ratio}, 
\[ H_s (z,w) = \int_{\mathbb R} P(z, wx)^s P(w, wx)^{-s} \mu(dx). \]
Let $y = wx = w_2 x + w_1$. 
Then, $\mu(dx) = \dfrac{w_2}{\pi |y-w|^2} dy$. 
Hence, 
\[ \int_{\mathbb R} P(z, wx)^s P(w, wx)^{-s} \mu(dx) = \int_{\mathbb R} P(z,y)^s P(w,y)^{1-s} \mu(dy). \]
\end{proof}

\begin{lemma}\label{lem:intertwin-pre}
For every $s \in (0,1/2)$, there exists a positive constant $C(s)$ such that for every $a \in \mathbb R$ 
\[ (1+a^2)^{-s} = \frac{C(s)}{\pi} \int_{\mathbb R} \frac{|x+a|^{-2s}}{(1+x^2)^{1-s}} dx. \]
\end{lemma}

\begin{proof}
Let $x = \tan \theta, |\theta| < \pi/2$. 
Then, $d\theta = \cos^2 \theta dx = \frac{1}{1+x^2} dx$ and  
\[ \frac{(x+a)^{2}}{1+x^2} = (\sin \theta + a \cos \theta)^2. \]
Hence, 
\[ \int_{\mathbb R} \frac{|x|^{-2s}}{(1+(x-a)^2)^{1-s}} dx = \int_{-\pi/2}^{\pi/2} \left| \sin \theta + a \cos \theta \right|^{-2s} d\theta. \]
By symmetry, 
\[ \int_{-\pi/2}^{\pi/2} \left| \sin \theta + a \cos \theta \right|^{-2s} d\theta = \frac{1}{2} \int_{-\pi}^{\pi} \left| \sin \theta + a \cos \theta \right|^{-2s} d\theta = \pi (1+a^2)^{-s} \int_{-\pi}^{\pi} \left| \cos \theta \right|^{-2s} d\theta. \]
The assertion holds if we let $C(s) := \left( \int_{-\pi}^{\pi} \left| \cos \theta \right|^{-2s} d\theta \right)^{-1}$. 
\end{proof}

\begin{lemma}[intertwining formula]\label{lem:intertwin}
For every $s \in (0,1/2)$, $w \in \mathbb H$ and $y \in \mathbb R$, 
\begin{equation}\label{eq:intertwin} 
P(w,y)^{s} = C(s) \int_{\mathbb R} P(w,x)^{1-s} \left( \frac{(x-y)^2}{(x^2 + 1) (y^2 + 1)} \right)^{-s} \mu(dx). 
\end{equation} 
\end{lemma}

\begin{proof}
Let $\xi := w - y$ and $t := x-y$. 
Then, \eqref{eq:intertwin} holds if and only if 
\begin{equation}\label{eq:1st-reduction}
\left(\frac{\textup{Im}(\xi)}{|\xi|^2}\right)^s  = \frac{C(s)}{\pi} \int_{\mathbb R} \left( \frac{\textup{Im}(\xi)}{|\xi - t|^2} \right)^{1-s} |t|^{-2s} dt. 
\end{equation}

Let $u := (t - \textup{Re}(\xi))/\textup{Im}(\xi)$. 
Then, 
\[ \int_{\mathbb R} \left( \frac{\textup{Im}(\xi)}{|\xi - t|^2} \right)^{1-s} |t|^{-2s} dt = (\textup{Im}(\xi))^{-s} \int_{\mathbb R} \left(\frac{1}{1+u^2}\right)^{1-s} \left| u + \frac{\textup{Re}(\xi)}{\textup{Im}(\xi)} \right|^{-2s} du. \]
Hence \eqref{eq:1st-reduction} holds if and only if 
\[ \left( \left( \frac{\textup{Re}(\xi)}{\textup{Im}(\xi)} \right)^2+ 1 \right)^{-s} = \frac{C(s)}{\pi} \int_{\mathbb R} \left(\frac{1}{1+u^2}\right)^{1-s} \left| u + \frac{\textup{Re}(\xi)}{\textup{Im}(\xi)} \right|^{-2s} du, \]
which follows from Lemma \ref{lem:intertwin-pre}. 
\end{proof}

By Proposition \ref{prop:prop1} and Lemma \ref{lem:intertwin}, 

\begin{proposition}\label{prop:prop2}
For every $s \in (0,1/2)$, 
\[ H_s (z,w) = C(s) \iint_{\mathbb R^2} P(w,x)^{1-s} P(z,y)^{1-s} \left( \frac{(x-y)^2}{(x^2 + 1) (y^2 + 1)} \right)^{-s} \mu(dx)\mu(dy), \ \ z, w \in \mathbb H.\] 
\end{proposition}

Step 3. \ 
 \begin{lemma}
 $k_s (x,y)$ is a positive definite kernel on $\mathbb R$. 
 \end{lemma}
 
 \begin{proof}
 For $r \in (0,1)$, let 
 \[ k^{(r)}_s (x,y) := \left(1-r \frac{(xy + 1)^2}{(x^2 + 1)(y^2 + 1)}\right)^{-s}. \]
 
 Since $(x,y) \mapsto \frac{1}{(x^2 + 1)(y^2 + 1)}$ and $(x,y) \mapsto (xy)^2 + 2xy + 1$ are both positive definite kernels on $\mathbb R$, 
 $(x,y) \mapsto \frac{(xy + 1)^2}{(x^2 + 1)(y^2 + 1)}$ is also a positive definite kernel on $\mathbb R$. 
 
 By the Taylor expansion, 
 \[ (1-x)^{-s} = \sum_{n=0}^{\infty} a_n x^n, \ |x| < 1,  \]
 for $a_n \ge 0, n = 0,1, \cdots$. 
 Hence $ k^{(r)}_s (x,y)$ is a positive definite kernel on $\mathbb R$. 
 Since $\lim_{r \to 1-0}  k^{(r)}_s (x,y) = k_s (x,y)$, 
$ k^{(r)}_s (x,y)$ is a positive definite kernel on $\mathbb R$. 
 \end{proof}
 
By this and the quadrature rule for the Riemannian integral for continuous functions, 
it holds that for every $a < b$ and $r \in (0,1)$, 
 \[ \iint_{[a,b]^2} \frac{\varphi_s (x) \varphi_s (y) k^{(r)}_s (x,y)}{(x^2 + 1)(y^2 + 1)} dxdy \ge 0. \]

Since $0 \le k^{(r)}_s (x,y) \le k_s (x,y)$, 
 \[ \iint_{\mathbb R^2} \frac{   |\varphi_s (x) \varphi_s (y)|  k^{(r)}_s (x,y)}{(x^2 + 1)(y^2 + 1)} dxdy \le \iint_{\mathbb R^2} \frac{   |\varphi_s (x) \varphi_s (y)|  k_s (x,y)}{(x^2 + 1)(y^2 + 1)} dxdy \le \sum_{i.j=1}^{n} |c_i| |c_j| H_s (z_i, z_j) < + \infty.  \]
By the Lebesgue convergence theorem,  
we see that for every $r \in (0,1)$, 
 \[ \iint_{\mathbb R^2}  \frac{  \varphi_s (x) \varphi_s (y)  k^{(r)}_s (x,y)}{(x^2 + 1)(y^2 + 1)} dxdy = \lim_{n \to \infty} \iint_{[-n,n]^2} \frac{\varphi_s (x) \varphi_s (y) k^{(r)}_s (x,y)}{(x^2 + 1)(y^2 + 1)} dxdy \ge 0. \]
and furthermore, 
 \[ \iint_{\mathbb R^2}  \frac{  \varphi_s (x) \varphi_s (y)  k_s (x,y)}{(x^2 + 1)(y^2 + 1)} dxdy = \lim_{r \to 1-0} \iint_{\mathbb R^2}  \frac{  \varphi_s (x) \varphi_s (y)  k^{(r)}_s (x,y)}{(x^2 + 1)(y^2 + 1)} dxdy \ge 0. \]
This completes the proof.

\section{Code snippet for Taylor expansions of $f$-divergences}\label{sec:maximataylor}
We provide below a code using the {\sc Maxima}\footnote{\url{https://maxima.sourceforge.io/}} software to calculate the truncated Taylor series of $f$-divergences between two Cauchy distributions.

{\small
\begin{verbatim}
Cauchy(x,l,s) := (s/(%pi*((x-l)**2+s**2)));
KLCauchy(l1,s1,l2,s2) := log(((s1+s2)**2+(l1-l2)**2)/(4*s1*s2)) ;
l1:0;
s1:1;
l2:0.6;
s2:6/5;
k:40;
testcond: (9/16)-(l2**2+(s2-(4/5))**2);
print("Is condition>0 for Taylor expansion?:",testcond);
Cauchy1:Cauchy(x,l1,s1);
Cauchy2:Cauchy(x,l2,s2);
print("Exact KL");
KLCauchy(l1,s1,l2,s2);
ExactKL:float(%);
print("KL numerical integration:");
kla: quad_qagi( Cauchy1*log(Cauchy1/Cauchy2), x, minf, inf,'epsrel=1d-10);
NumKL:float(kla[1]);

for i:2 while (i<=k)  
do( r[i]: quad_qagi( (Cauchy1-Cauchy2)**i/Cauchy2**(i-1), x, minf, inf,'epsrel=1d-10),
 print(i,r[i][1]));

print("KL Taylor truncated series:");
TaylorKL: sum( (((-1)**i)/i)*r[i][1], i, 2, k);
print("Exact:",ExactKL,"Numerical:",NumKL,"Trunc. Taylor", TaylorKL);
print("Error |Taylor-Exact|",abs(TaylorKL-ExactKL));
\end{verbatim}
}

\end{document}